\documentclass[a4paper,12pt,numbers=noenddot]{article}
\usepackage[a4paper,left=3cm,right=3cm,top=2.5cm,bottom=2.5cm]{geometry}
\usepackage{graphicx} 
\usepackage[utf8]{inputenc}
\usepackage{amsmath}
\usepackage{amsthm}
\usepackage{amssymb}
\usepackage{physics}
\usepackage{faktor}
\usepackage[colorlinks=true, allcolors=blue]{hyperref}
\graphicspath{ {./images/} }
\usepackage{bbm}
\usepackage{authblk}
\usepackage{graphicx}
\usepackage{float}
\usepackage{subfig}
\usepackage{wrapfig}
\usepackage{caption}
\newcommand{\R}{\mathbb{R}}
\newcommand{\C}{\mathbb{C}}

\newcommand{\T}{\mathbb{T}}
\newcommand{\Z}{\mathbb{Z}}
\newcommand{\Sp}{\mathbb{S}}
\newcommand{\Hp}{\mathcal{H}}
\newcommand{\F}{\mathcal{F}}
\newcommand{\N}{\mathbb{N}}

\newcommand{\Pp}{\mathbb{P}}
\newcommand{\E}{\mathbb{E}}

\newtheorem{theorem}{Theorem}
\newtheorem{definition}{Definition}
\newtheorem{lemma}{Lemma}
\newtheorem{proposition}{Proposition}

\newtheorem{remark}{Remark}
\theoremstyle{definition}
\newtheorem{example}{Example}

\providecommand{\keywords}[1]
{
  \small	
  \textbf{Keywords --} #1
}

\providecommand{\msc}[1]
{
  \small	
  \textbf{Mathematics Subject Classification --} #1
}

\title{Dynamical Localization and Transport properties of Quantum Walks on the hexagonal lattice}
\author{Andreas Schaefer%
  \thanks{Electronic address: \texttt{andreas.schaefer@univ-grenoble-alpes.fr}}}
\affil{Institut Fourier, Université Grenoble Alpes, France}
\date{\today}

\begin{document}

\maketitle

\begin{abstract}
    We study coined Random Quantum Walks on the hexagonal lattice, where the strength of disorder is monitored by the coin matrix. Each lattice site is equipped with an i.i.d. random variable that is uniformly distributed on the torus and acts as a random phase in every step of the QW. We show exponential decay of the fractional moments of the Green function in the regime of strong disorder, that is whenever the coin matrix is sufficiently close to the fully localized case, using a fractional moment criterion and a finite volume method. In the decorrelated case, we deduce dynamical localization. Moreover, we adapt a topological index to our model and thereby obtain transport for some coin matrices.
\end{abstract}

\noindent \keywords{Quantum walks, localization, fractional moment criterion, transport, topological index}

\noindent \msc{81Q35}

\section{Introduction}
Quantum Walks (QWs) are the quantum mechanical analogue of classical Random Walks and have attracted considerable attention in recent years, owing to their various applications in fields such as quantum information or condensed matter physics. A comprehensive introduction to QWs can be found in works such as \cite{Portugal:2013}, \cite{VenegasAndraca:2012}, \cite{Kempe:2003} and \cite{Qiang:2024}. In the realm of quantum information, QWs are studied for their potential to exploit quantum phenomena such as entanglement or ballistic motion (see e.g. \cite{Venancio:2023}, \cite{Ahlbrecht:2011}) to obtain algorithms that outperform their classical counterparts. One notable example of this are quantum search algorithms, see \cite{Grover:1996}, \cite{Shenvi:2003}, \cite{Szegedy:2004}, \cite{Ambainis:2019} , \cite{Portugal:2013} among others. In condensed matter physics, on the other hand, QWs can be used as an approximation of the more complex behavior of physical systems, see for example \cite{Chalker:1988}, \cite{Karski:2009}, \cite{Zaehringer:2010}.

It therefore comes as no surprise that QWs have been the subject of extensive mathematical investigation. This research has encompassed a variety of aspects, including the rigorous study of decoherence (\cite{Ahlbrecht:2011}, \cite{Ahlbrecht:2012}, \cite{Joye_2011}, \cite{HJ:2012}), their spectral properties (\cite{Bourget:2003}, \cite{Joye:2014}, \cite{HJ:2014}, \cite{Higuchi:2014}, \cite{Higuchi:2018}, \cite{Richard:2017}), continuum limits (\cite{Molfetta:2018}, \cite{Molfetta:2022}, \cite{Molfetta:2024}) and topological features (\cite{Cedzich:2017}, \cite{Cedzich:2018}), which can be used to show the existence of transport (\cite{ABJ:2017}, \cite{ABJ:2020}, \cite{AschMouneime:2019}). In contrast, localization has been observed for specific types of QWs. While localization has been extensively studied for self-adjoint operators (\cite{FroehlichSpencer:1983}, \cite{AizenmanSchenker:2001}, \cite{Stollmann:2001}, \cite{AizenmanMolchanov:1993}, \cite{AENSS:2005}, \cite{Warzel:15}), methods to extend these results to the unitary case have been developed (\cite{Joye:2005}, \cite{HJS:09}) and applied to specific models (\cite{Koshovets:1991}, \cite{HJS:2006}, \cite{JM:2010}, \cite{Joye:2012}, \cite{ABJ:2012}, \cite{Klausen:2023}, \cite{Boumaza:2025}, see also the reviews \cite{Joye:2011}, \cite{ABJ:2015} and references therein). In particular, several recent works have established dynamical localization for quantum walks and related CMV matrices in disordered and quasi-periodic settings (\cite{ASW:2011}, \cite{CedWer:2021}, \cite{WangDam:2019}, \cite{CedFil:2023}, \cite{CedFilLi:2024}).

The main goal of this paper is to adapt the method developed in \cite{HJS:09} to a specific coined QW on the hexagonal lattice with the aim of showing dynamical localization in the strong disorder regime, which is monitored by the coin matrix. We establish this result using the fractional moment criterion and a finite volume method. A novelty of this work is the introduction of small correlations, which to the best of the authors knowledge have only been investigated in \cite{Klausen:2023} (therein the preprint "Quantum walks in random magnetic fields" together with A. H. Werner and C. Cedzich). In contrast to most previous works, where random phases are represented by an infinite diagonal matrix with i.i.d. random phases along the diagonal, our model incorporates randomness through an infinite diagonal matrix in which each triplet of diagonal entries shares the same phase. However, this comes with a caveat: The fact that exponential decay of the fractional moments implies dynamical localization relies on the independence of the randomness, and does not directly generalize to the correlated case. As a way out, we can add additional random variables to our model, rendering the randomness independent and obtaining dynamical localization. Although the extension of the fractional moment criterion to the correlated case has been conjectured \cite{Klausen:2023}, it remains unresolved and may serve as the subject of future research. For further details, we refer to Remark \ref{remark: correlations}. To complete the picture,  we address the translation-invariant case and investigate transport properties using the topological methods developed in \cite{ABJ:2020}, \cite{AschMouneime:2019}. By adapting a topological index to our model, we derive criteria for the absolutely continuous spectrum to span the entire unit circle, which serves as an indicator of Quantum transport. These criteria depend on the absolute value of the entries of the coin matrices at infinity and thus remain valid when random phases are introduced to the model. We note that, under the right assumptions, these results can be generalized to a broader class of QWs, using a scattering approach as presented in \cite{joye:2024}. The choice of the hexagonal lattice is partially dictated by the specificities of electronic transport in graphene samples, see e.g. \cite{Fefferman:2012} for the Hamiltonian classification of the phenomenon.

The paper is structured as follows: We start in section \ref{sec: model and results} by defining the considered QW and stating the main result concerning dynamical localization, namely Theorem \ref{thm:main_thm}. Section \ref{sec: trans-inv case} deals with the translation-invariant case and its band structure. In section \ref{sec: proof of dyn loc total} we prove exponential decay of the fractional moments as stated in Theorem \ref{thm:main_thm}. We first show the boundedness of fractional moments in \ref{sec: frac mom bound} and define a finite volume restriction in \ref{sec: box definition}. We then obtain an iterative step in \ref{sec: iterative step} that we repeatedly apply in \ref{sec: main proof} to conclude the proof. In section \ref{sec: topo properties} we turn to topological properties of the QW, using an index that leads to transport for specific choices of coin matrices, both with and without randomness (see Theorem \ref{thm: index theorem}). The appendices \ref{sec: first resampling argu} and \ref{sec: second resampling argu} deal with a re-sampling argument, a technical step required in the proof of dynamical localization. Sections \ref{sec: trans-inv case}, \ref{sec: proof of dyn loc total} and \ref{sec: topo properties} are independent from each other and can be read in any order. \vspace{2mm} \\
\textbf{Acknowledgements:} \\
I would like to express my heartfelt gratitude to Simone Warzel and her research group for their warm hospitality and the many insightful discussions during my two months in Munich working on this project. My thanks also go to my PhD supervisor Alain Joye for his support and guidance at every stage of this work. 

\section{The model and main result} \label{sec: model and results}
The hexagonal lattice $\Gamma$ is defined using the integer spans $\Gamma_A$ and $\Gamma_B$
\begin{align*}
    \Gamma_A = \text{span}_\Z \left( \begin{pmatrix}
        \sqrt{3} \\ 0
    \end{pmatrix}, \, \frac{1}{2} \begin{pmatrix}
        \sqrt{3} \\ 3
    \end{pmatrix} \right), \;\; \Gamma_B = \Gamma_A + \begin{pmatrix}
        0 \\ 1
    \end{pmatrix}, \;\; \Gamma = \Gamma_A \cup \Gamma_B.
\end{align*}
Fixing an orthonormal basis (ONB) $f_1, f_2$ of $\C^2$ and using that $l^2(\Gamma) \simeq l^2(\Z^2) \otimes \C^2$, an ONB of $l^2(\Gamma)$ is given by $\ket{j,k} \otimes f_i$. The components $j,k$ are the coefficients with respect to the vectors used to span $\Gamma_A$, while the vectors $f_1$ and $f_2$ refer to lattice sites in $\Gamma_A$ and $\Gamma_B$ respectively. We define the shift operators $S_{1}, S_{2}, S_{3}$ as translation-invariant operators on $l^2(\Gamma)$: 
\begin{align} \label{eq: expre of shifts}
\begin{split}
    S_{1} \left( \ket{j,k} \otimes f_1 \right) = \ket{j,k-1} \otimes f_2, \;\;\;\;\;\;\;\;
    S_{1} \left( \ket{j,k} \otimes f_2 \right) &= \ket{j-1,k+1} \otimes f_1 \\
    S_{2} \left( \ket{j,k} \otimes f_1 \right) = \ket{j+1,k-1} \otimes f_2, \;\;
    S_{2} \left( \ket{j,k} \otimes f_2 \right) &= \ket{j,k} \otimes f_1 \\
    S_{3} \left( \ket{j,k} \otimes f_1 \right) = \ket{j,k} \otimes f_2, \;\;\;\;\;\;\;\;\;\;\;\;\;\;
    S_{3} \left( \ket{j,k} \otimes f_2 \right) &= \ket{j,k+1} \otimes f_1.
\end{split}
\end{align}
Their action is illustrated in Figure \ref{fig:shifts definition}.
\begin{figure} [h!]
    \centering
    \includegraphics[scale=1]{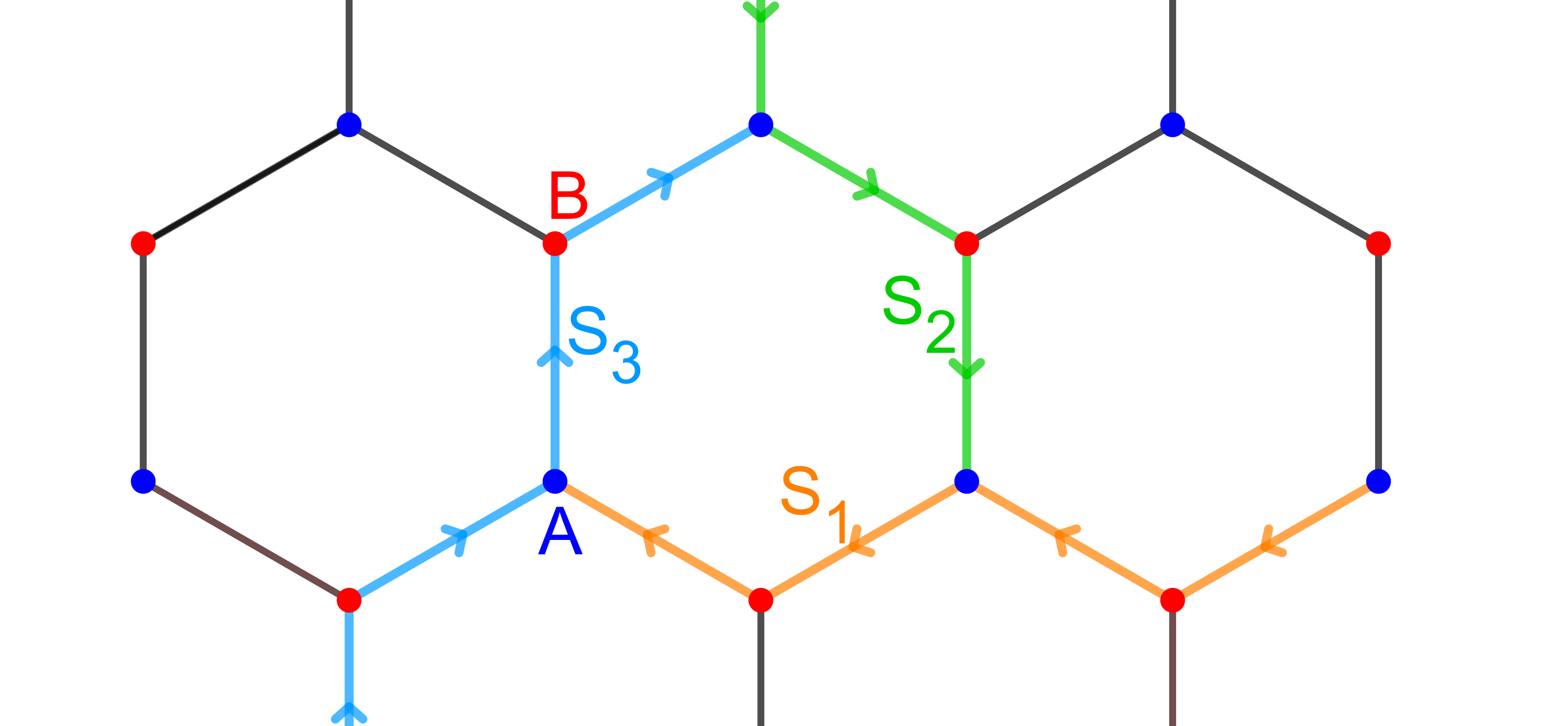}
    \caption{The action of the shift operators $S_j$}
    \label{fig:shifts definition}
\end{figure}

We equip every site $x \in \Gamma$ with a three-dimensional coin space. The Quantum Walk acts on the Hilbert space $\Hp = l^2(\Gamma) \otimes \C^3 \simeq l^2(\Z^2 ; \C^6)$. We fix an arbitrary ONB $e_1, e_2, e_3$ of $\C^3$ and define the shift operator $S \in \mathcal{B}(\Hp)$
\begin{align*}
    S = S_{1} \otimes \ket{e_1} \bra{e_1} + S_{2} \otimes \ket{e_2} \bra{e_2} + S_{3} \otimes \ket{e_3} \bra{e_3}.
\end{align*}
To every lattice site $x \in \Gamma$ we associate a unitary coin matrix $C_x \in U(3)$. For any $\ket{x} = \ket{j,k} \otimes f_i$ and coin state $\tau \in \C^3$ we define the coin operator $\mathcal{C} = \{C_x\}_{x \in \Gamma}$
\begin{align*}
    \mathcal{C} \left( \ket{x} \otimes \tau \right) = \ket{x} \otimes C_x \, \tau.
\end{align*}
The deterministic QW is given by $U = S \mathcal{C}$. In some instances we will assume that there exists a matrix $C \in U(3)$ such that $C_x = C$ for all $x \in \Gamma$. In this case, we will denote the coin operator by $\mathcal{C}(C)$ and the corresponding QW by $U(C)$.

Every $x \in \Gamma$ is equipped with an i.i.d. random variable $\omega_x$ that is uniformly distributed on the torus $\T$. Note that whenever we view the torus as the real interval $[0, 2\pi)$, we denote it by $\T$, while the complex sphere it is homeomorphic to will be called $\Sp^1$. We define the random unitary operator
\begin{align*}
    D_\omega \left( \ket{x} \otimes \tau \right) = e^{i \omega_x} \left( \ket{x} \otimes \tau \right).
\end{align*}
The random QW is given by $U_\omega = D_\omega U$.
\begin{remark} \label{remark: randomness as random coin mat}
    We define the random QW by $D_\omega S \mathcal{C}$, applying the random phase only after the shift operator to align with the existing literature. However, an alternative definition is $W_\omega = S \mathcal{C} D_\omega$. They are unitarily equivalent due to $W_\omega = D_\omega^* U_\omega D_\omega$ and thus share the same transport and spectral properties. Due to $\mathcal{C} D_\omega \ket{x} \otimes \tau = e^{i \omega_x} \ket{x} \otimes C_x \tau$ we can replace the deterministic coin matrices $C_x$ by the random matrices $e^{i \omega_x} C_x$. Using the random coin operator $\mathcal{C}_\omega = \{e^{i \omega_x} C_x \}_{x \in \Gamma}$, we have $W_\omega = S C_\omega$ and can view the randomness as a random phase associated to each coin matrix.
\end{remark}
Let $C_0 = \begin{pmatrix}
    0 & 1 & 0 \\ 0 & 0 & 1 \\ 1 & 0 & 0
\end{pmatrix}$ be the permutation matrix corresponding to the permutation $(1 \; 3 \; 2)$ and denote $x^{(i)} = \ket{x} \otimes e_i$. Our goal is the following theorem:
\begin{theorem} \label{thm:main_thm}
    For any $s \in (0,\frac{1}{3})$ there exists $\varphi, c, g > 0$ such that for any $C \in U(3)$ with $\| C - C_0 \|_\infty \leq \varphi$ and all lattice sites $x, y \in \Gamma$, coin states $i,j \in \{1,2,3\}$, $z \in \C$ with $|z|<1$:
    \begin{align} \label{eq:expo_decay}
        \E \left( | \langle x^{(i)} | \, \big( U_\omega(C)-z \big)^{-1} \, y^{(j)} \rangle |^s \right) \leq c \, e^{-g |x - y|}.
    \end{align}
\end{theorem} 
Here $|x - y|$ denotes the length of the shortest path connecting $x$ and $y$. Thus, $d(x, y) = |x - y|$ is a metric. 
\begin{remark} \label{remark: generalization of main thm}
    The proof of Theorem \ref{thm:main_thm} can be directly generalized to any deterministic family of unitary matrices $(C_x)_{x \in \Gamma}$ by replacing the assumption $\|C - C_0\|_\infty \leq \varphi$ by $\sup_{x \in \Gamma}\| C_x - C_0 \|_\infty \leq \varphi$. Thus, Theorem \ref{thm:main_thm} also applies to non-ergodic models.
\end{remark}
\begin{remark} \label{remark: correlations}
    We emphasize that the elements $\ket{x^{(1)}}$, $\ket{x^{(2)}}$ and $\ket{x^{(3)}}$ of the ONB of $\Hp$ are equipped with the same random variable. When the operator $D_\omega$ is expressed in this ONB, this yields an infinite diagonal matrix, where each triplet of matrix elements shares the same random variable. Unlike the models in the existing works \cite{HJS:09} and \cite{ABJ:2012} where the infinite matrix $D_\omega$ has independent diagonal entries, this structure introduces some correlations to the model.
    
    The fact that exponential decay of the fractional moments implies dynamical localization is proven in Theorem 3.2 in \cite{HJS:09}. However, their proof relies on the independence of the diagonal entries in the matrix expression of $D_\omega$ and can not be directly adapted to the correlated framework presented above. This is because Krein's formula, which is equation (5.5) in \cite{HJS:09}, requires a rank one perturbation and does not easily generalize to other finite-dimensional perturbations. Since $D_\omega$ has correlated diagonal blocks of size $3 \times 3$, we are dealing with rank three perturbations. Nevertheless, these difficulties seem rather technical than fundamental and dynamical localization is still expected in the correlated case (see \cite{Klausen:2023}, Quantum Walks in Random Magnetic Fields, Conjecture 2.5).
\end{remark}
Nevertheless, we can still conclude dynamical localization if we adapt the random operator $D_\omega$: Equipping every site $x \in \Gamma$ with three i.i.d. random variables $\omega_{x,1}$, $\omega_{x,2}$ and $\omega_{x,3}$ we can define the random operator $D_\omega \ket{x^{(j)}}= e^{i \omega_{x,j}} \, \ket{x^{(j)}}$. Its matrix expression in the ONB $(\ket{x^{(i)}})_{x,i}$ has now independent diagonal entries. Since Theorem \ref{thm:main_thm} still holds, we can conclude dynamical localization using Theorem 3.2 in \cite{HJS:09} under the assumptions of Theorem \ref{thm:main_thm}:
\begin{align} \label{eq:def_dyn_loc}
    \E \left( \sup_{n \in \Z} | \langle x^{(i)}\, | \, U_\omega(C)^n \, y^{(j)} \rangle | \right) \leq c e^{-g |x - y|}.
\end{align}

\begin{remark}
    We can adapt the proof of Theorem 3.2 in \cite{HJS:09} to the decorrelated model presented above without any significant changes. In particular, Proposition 5.1 in \cite{HJS:09} can be reformulated as
    \begin{align*}
        \E \left( (1-|z|^2) \, | \bra{x^{(i)} } (U_\omega-z)^{-1} \ket{y^{(j)}} |^2 \right) \leq C(s) \, \sum_{\substack{v^{(k)} \text{ s.t.} \\ |v-x| \leq 2}} \E \left( | \bra{v^{(k)}} (U_\omega -z)^{-1} \ket{y^{(j)}} |^s \right).
    \end{align*}
    We stress again that the key step is equation (5.5) in \cite{HJS:09}, which does not easily generalize to the correlated case. The integral representation for $f(U)$ using the resolvent holds for general unitary operators, while Section 5.3 in \cite{HJS:09} does not involve any model specific properties.
\end{remark}

\begin{wrapfigure}{r}{0.35\textwidth}
    \centering
    \includegraphics[scale=2]{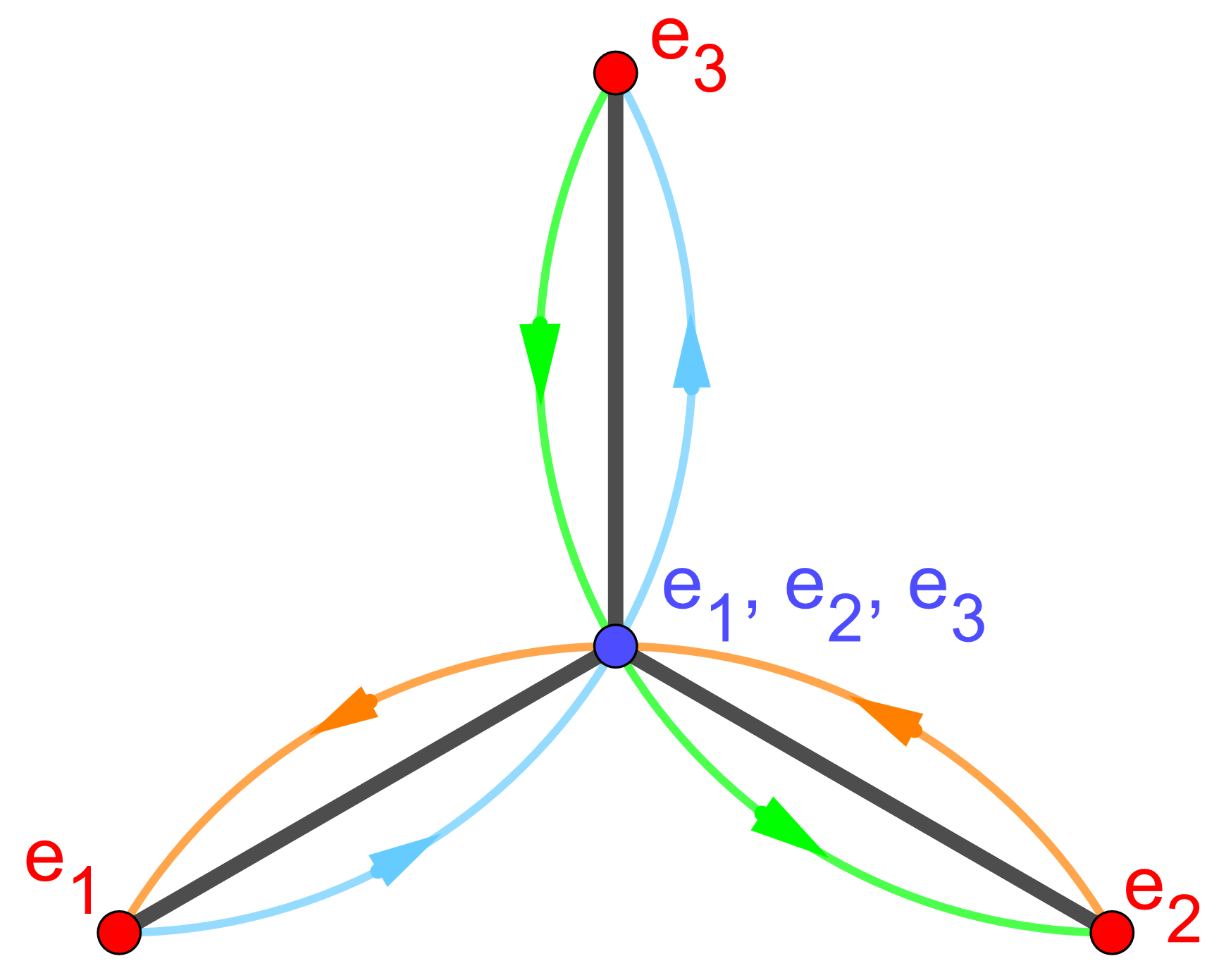}
    \caption{The invariant subspace $\Hp^{j,k}$ of $U_\omega(C_0)$}
    \label{fig:invariant subspace}
\end{wrapfigure}
Dynamical localization means that the probability to move from a lattice site $x \in \Gamma$ to another site $y \in \Gamma$ decreases almost surely exponentially in the distance $|x - y|$, independently of how many steps $n$ the Quantum Walker may take. Using the coin matrix $C_0$, the Quantum Walk is fully localized, see Figure \ref{fig:invariant subspace}. In other words, the subspaces $\Hp^{j,k}$ as defined in \eqref{eq:def_H_jk} are invariant under $U_\omega(C_0)$. Thus, Theorem \ref{thm:main_thm} yields dynamical localization if the chosen coin matrix $C$ is sufficiently close to the fully localized case $C_0$.

\begin{align} \label{eq:def_H_jk}
    \begin{split}
        \Hp^{j,k} = \text{span} \Big\{ \ket{j,k} &\otimes f_1 \otimes e_1, \; \ket{j,k} \otimes f_1 \otimes e_2, \; \ket{j,k} \otimes f_1 \otimes e_3, \\ \ket{j,k} &\otimes f_2 \otimes e_3, \; \ket{j,k-1} \otimes f_2 \otimes e_1, \; \ket{j+1,k-1} \otimes f_2 \otimes e_2 \Big\}.
    \end{split}
\end{align}

\begin{remark}
    The matrix $\tilde{C}_0$ corresponding to the permutation $(1 \; 2 \; 3)$ induces full localization, similarly to $C_0$ which corresponds to the permutation $(1 \; 3 \; 2)$. The invariant subspaces $\Hp^{j,k}$ pictured in Figure \ref{fig:invariant subspace} are then centered at a $\Gamma_B$ site instead of a $\Gamma_A$ site. The proof of Theorem \ref{thm:main_thm} can be adapted without any major changes, thereby yielding a fractional moment estimate for any choice of $C \in U(3)$ such that $\| C- \Tilde{C}_0 \| \leq \varphi_0$. Using the techniques presented in section \ref{sec: trans-inv case}, it can be shown that $C_0$ and $\tilde{C}_0$ are (up to phase factors) the only coin matrices inducing full localization. 
\end{remark}

\begin{remark}
    It was shown in \cite{felix:22} that without random phases the coin matrix $C(\theta) = \begin{pmatrix}
        0 & \cos(\theta) & \sin(\theta) \\ 0 &-\sin(\theta) & \cos(\theta) \\ 1 & 0 & 0
    \end{pmatrix}$ leads to non-empty continuous spectrum for any $\theta \notin \pi \Z$ and therefore no localization. Since $\| C(\theta) - C_0 \|$ can be arbitrarily small, this shows that it is the randomness which drives localization.
\end{remark}

\section{The band structure of translation-invariant QWs} \label{sec: trans-inv case}
We consider a deterministic Quantum Walk $U = S \mathcal{C}$ as defined in section \ref{sec: model and results}. Furthermore, we assume that the coin operator $\mathcal{C}$ is such that there exist two unitary matrices $C_A, C_B \in U(3)$ such that $C_x = C_A$ on all lattice sites $x \in \Gamma_A$ and $C_x = C_B$ on all $x \in \Gamma_B$. The resulting operator $U$ is translation-invariant with respect to shifts that respect the lattice structure, i.e. shifts $T_v: x \mapsto x+v$ where $v \in \Gamma_A$. We consider the Fourier Transform
\begin{align*}
    \F: l^2(\Z^2) \to L^2(\T^2), \;\; \widehat{\varphi}(k) = \frac{1}{2 \pi} \, \sum_{n \in \Z^2} e^{i k \cdot n} \, \varphi(n).
\end{align*}
$\F$ is a unitary operator when the spaces $l^2(\Z^2)$ and $L^2(\T^2)$ are equipped with their standard scalar products without normalization factors. The Fourier Transform on $l^2(\Z^2) \otimes \C^6$ is given by $\F \otimes I$ and the Fourier transformed QW is thus defined as $\widehat{U} = (\F \otimes I) U (\F \otimes I)^{-1}$. It is a unitary operator on $L^2(\T^2) \otimes \C^6$. We define the translation $T_v$ on elements of $l^2(\Z^2)$ by $T_v \varphi (x) = \varphi(x+v)$ and remember that $\big( \F T_v \, \varphi \big) (k) = e^{-i v \cdot k} \widehat{\varphi}(k)$. Since $U$ is translation invariant, an easy calculation shows that $\widehat{U}$ acts as a multiplication by the matrix 
\begin{align*}
    \widehat{U}(k) = \begin{pmatrix}
        0 & S_{BA}(k) C_B \\ S_{AB}(k) C_A & 0
    \end{pmatrix}.
\end{align*}
Here we have ordered the ONB of $\C^6$ by $f_1 \otimes e_1,..., f_1 \otimes e_3, f_2 \otimes e_1,..., f_2 \otimes e_3$. The blocks are given by
\begin{align*}
    S_{BA}(k) = \begin{pmatrix}
        e^{i (k_1-k_2)} & 0 & 0 \\
        0 & 1 & 0 \\
        0 & 0 & e^{-i k_2}
    \end{pmatrix} \;\; \text{and} \;\; S_{AB}(k) = \begin{pmatrix}
        e^{i k_2} & 0 & 0 \\
        0 & e^{i (k_2-k_1)} & 0 \\
        0 & 0 & 1
    \end{pmatrix}.
\end{align*}
The matrix $S_{AB}(k)$ (respectively $S_{BA}(k)$) corresponds to the Walker moving from $\Gamma_A$ sites to $\Gamma_B$ sites (respectively $\Gamma_B$ to $\Gamma_A$), while the coefficients $e^{i v \cdot k}$ stem from the expression of the shifts $S_j$ as defined in \eqref{eq: expre of shifts}. Due to the off-diagonal structure of $\widehat{U}(k)$, we have:
\begin{align*}
    \sigma(U) &= \bigcup_{k \in \T^2} \big\{ \pm \sqrt{\lambda} \, \text{ s.t. } \lambda \in \sigma(V(k)) \big\}\, \text{ where } V(k)= S_{BA}(k) \, C_B \, S_{AB}(k) \, C_A.
\end{align*}
Let $k\mapsto \lambda_j(k)$ denote the three eigenvalues of $V(k)$ and $\mu_j(k) = -i \log(\lambda_j(k))$ their real arguments, which we call the bands of $U$. Since $\det(V(k)) = \det(C_B) \, \det(C_A)$ is independent of $k$, we can normalize and assume without loss of generality that $V(k)$ has determinant $1$. The characteristic polynomial of $V(k)$ is given by
\begin{align*}
    \chi_{V(k)}(\lambda) = \lambda^3 - \tr(V(k)) \, \lambda^2 + \overline{\tr(V(k))} \, \lambda - 1.
\end{align*}
The spectrum of $U$ is thus characterized by the map $k \mapsto \tr(V(k))$. We immediately see that $\sigma_{cont}(U) = \varnothing$ if and only if $k \mapsto \tr(V(k))$ is constant, in which case $\sigma(U)$ consists of at most six infinitely degenerate eigenvalues.

Moreover, we can define the set of trace values $\mathcal{T} = \{\tr(U) | \, U \in SU(3) \}$. It is a closed subset of $\C$ and we can easily show that two bands agree (e.g. $\mu_1(k) = \mu_2(k)$) at some $k \in \T^2$ if and only if $\tr(V(k)) \in \partial \mathcal{T}$, see \cite{felix:22}. From this we can derive a fruitful theory regarding intersections in the band structure. In particular, we can obtain a criterion for conical intersections in the band structure based on the first two derivatives of the trace of $V(k)$. We note that conical intersections, also called Dirac points, are indeed observed for graphene, whose atomic structure is a hexagonal lattice \cite{Fefferman:2012}. Restricting the model to circulant coin matrices offers further results, such as intersections at $k=0$ always being conical. These results will be presented in the PhD manuscript of the author.

\section{Exponential decay of the fractional moments} \label{sec: proof of dyn loc total}
We use sections \ref{sec: frac mom bound} - \ref{sec: main proof} to prove Theorem \ref{thm:main_thm}, following the methods presented in \cite{HJS:09} and \cite{ABJ:2012}. The proof uses a finite volume method and is structured as follows: We start by showing that the fractional moments in \eqref{eq:expo_decay} are bounded by some constant $C(s)$ in section \ref{sec: frac mom bound}. We define the restriction of the Quantum Walk to some finite box $\Lambda_L$ of size $L$ in section \ref{sec: box definition} and show that the operator describing the transition between $\Lambda_L$ and $\Lambda_L^C$ is bounded independently of $L$ in section \ref{sec: trans op}. In section \ref{sec: spec gaps full loc op} we investigate the probability of gaps in the spectrum of the fully localized operator, i.e. when $C=C_0$. Using this, we are able to obtain polynomial decay in $|L|$ of the resolvent restricted to some box $\Lambda_L$ in section \ref{sec: poly dec box res}. The key ingredient for the proof of Theorem \ref{thm:main_thm} is finally obtained in section \ref{sec: iterative step}. We show that the expectation $\E \left( \left| \langle x^{(i)} \, | \, (U_\omega(C) - z)^{-1} \, y^{(j)} \rangle \right|^s \right)$ can be bounded by $q<1$ times $\E \left( \left| \langle v^{(k)} \, | \, (U_\omega(C) - z)^{-1} \, y^{(j)} \rangle \right|^s \right)$, where $v^{(k)}$ lies on the boundary of $\Lambda_L$ (and is thus closer to $y^{(j)}$). Repeated application of this result finally yields Theorem \ref{thm:main_thm} in section \ref{sec: main proof}. The appendices \ref{sec: first resampling argu} and \ref{sec: second resampling argu} are devoted to proving the so-called re-sampling arguments, which are key ingredients in the finite volume method and omitted in section \ref{sec: iterative step}.

\subsection{Fractional moment bound} \label{sec: frac mom bound}
In this section we prove the boundedness of the fractional moments in equation \eqref{eq:expo_decay} under the weaker assumption that the i.i.d random variables $(\omega_x)_{x \in \Gamma}$ are subject to an absolutely continuous distribution $\mu$ with bounded density $\tau$. We prove the following:
\begin{theorem} \label{thm: frac mom bound}
    For every $s \in (0,1)$ there exists $C(s) < \infty$ such that
    \begin{align} \label{eq: fract mom bound}
        \int_\T \int_\T | \langle x^{(i)} \, | (U_\omega - z )^{-1} y^{(j)} \rangle |^s \; d\mu( \omega_{x}) \, d\mu(\omega_{y}) \leq C(s)
    \end{align}
    for all $z \in \C$, $|z| \neq 1$, $x^{(i)}, y^{(j)} \in l^2(\Gamma) \otimes \C^3$, and arbitrary values of all other random variables $\omega_{v}$. Thus, the expectation in \eqref{eq:expo_decay} is bounded by $C(s)$.
\end{theorem}

The proof of Theorem \ref{thm: frac mom bound} mainly follows the steps of \cite{HJS:09} (Theorem 3.1). However, due to our model being less random, the resulting dissipative operator $-i \widehat{F}_z$ acts on a six-dimensional space instead of a two-dimensional one, leading us to extend the integral bound presented in \cite{HJS:09} to any finite-dimensional space.

\begin{proof}
    The bound \eqref{eq: fract mom bound} is trivial when $z$ is away from the unit circle:
    \begin{align*}
        | \langle x^{(i)} \, | (U_\omega - z )^{-1} \, y^{(j)} \rangle | \leq \|(U_\omega - z )^{-1} \| = \frac{1}{\text{dist}(z, \sigma(U_\omega))} \leq \frac{1}{| 1 - |z| \, |}. 
    \end{align*}
    We may thus assume without loss of generality, that $\frac{1}{2} \leq |z| \leq 2$ and $|z| \neq 1$. Furthermore, it holds that
    \begin{align*}
        \frac{1}{2z} \, \Big( (U_\omega + z) (U_\omega - z)^{-1} - 1 \Big) = (U_\omega - z)^{-1}.
    \end{align*}
    To prove \eqref{eq: fract mom bound} it thus suffices to show the similar bound on $(U_\omega + z)(U_\omega-z)^{-1}$:
    \begin{align} \label{eq: bound on U plus z U minus z}
        \int_\T \int_\T | \langle x^{(i)} \, | (U_\omega+z) \, (U_\omega - z )^{-1} \, y^{(j)} \rangle |^s \; d\mu( \omega_{x}) \, d\mu(\omega_{y}) \leq C(s).
    \end{align}
    We assume for now that $x^{(i)} \neq y^{(j)}$, while the simpler case $x^{(i)} = y^{(j)}$ will be dealt with at the end. Following the steps of \cite{HJS:09}, we introduce the change of variables
    \begin{align*}
        \alpha = \frac{1}{2} \big( \omega_{x} + \omega_{y} \big) \;\;\; \text{and} \;\;\; \beta = \frac{1}{2} \big( \omega_{x} - \omega_{y} \big).
    \end{align*}
    Furthermore, we introduce the unitary operators:
    \begin{align*}
        D_\alpha v^{(k)} &= \begin{cases}
            e^{i \alpha} \, v^{(k)} & \text{if } v \in \{ x, y \} \\
            v^{(k)} & \text{else}
        \end{cases}, \;\;
        D_\beta v^{(k)} = \begin{cases}
            e^{i \beta} \, v^{(k)} & \text{if } v = x\\
            e^{-i \beta} \, v^{(k)} & \text{if } v = y \\
            v^{(k)} & \text{else}
        \end{cases} \\ \widehat{D} v^{(k)} &= \begin{cases}
            v^{(k)} & \text{if } v \in \{ x, y \} \\
            e^{i \omega_{v}} \, v^{(k)} & \text{else}
        \end{cases}.
    \end{align*}
    By definition, we have $D_\omega = D_\alpha \, D_\beta \, \widehat{D}$. We define the unitary operator $V_\omega = D_\beta \, \widehat{D} \, U$. Denoting by $P$ the projection onto the subspace spanned by $\{V_\omega^{-1} x^{(k)}, V_\omega^{-1} y^{(k)}, \, k=1,2,3 \}$, we easily verify that $(U_\omega - V_\omega) \, (1 - P) = 0$ and $V_\omega^{-1} \, U_\omega P = e^{i \alpha} P$. We conclude
    \begin{align} \label{eq: fmb estimate 1}
        U_\omega = U_\omega \, P + U_\omega \, (1-P) = e^{i \alpha} \, V_\omega \, P + V_\omega \, (1-P).
    \end{align}
    For $|z| \notin \{0,1\}$ we define the operators
    \begin{align*}
        F_z = P (U_\omega + z) (U_\omega -z)^{-1} P \text{  and  } \widehat{F}_z = P (V_\omega + z) (V_\omega -z)^{-1} P,
    \end{align*}
    viewing them as operators on the range of $P$. Using that $[V_\omega +z, (V_\omega -z)^{-1}] =0$ and $\big( (V_\omega - z)^{-1} \big)^* = (V_\omega^* - \overline{z})^{-1}$, we compute
    \begin{align*}
    \begin{split}
        \widehat{F}_z + \widehat{F}_z^* = P \Big( 2\big( I - |z|^2 \big) (V_\omega -z)^{-1} \, \big( (V_\omega -z)^{-1} \big)^* \Big) P.
    \end{split}
    \end{align*}
    Assuming that $|z| > 1$, we conclude that $\widehat{F}_z + \widehat{F}_z^* < 0 $ is invertible. Thus, $-i \widehat{F}_z$ is a dissipative operator if $|z| > 1$. Note that we call an operator $A$ dissipative if its imaginary part $\frac{1}{2i} (A-A^*)$ is positive. Using $\widehat{F}_z^{-1} + ( \widehat{F}_z^{-1} )^* = \widehat{F}_z^{-1} \big( \widehat{F}_z + \widehat{F}_z^* \big) (\widehat{F}_z^{-1})^*$, we conclude $ \widehat{F}_z^{-1} + (\widehat{F}_z^{-1})^* < 0$ and see that $-i \widehat{F}_z^{-1}$ is also dissipative for $|z| > 1$. Similarly, we obtain that $i \widehat{F}_z$ and $i \widehat{F}_z^{-1}$ are dissipative for $|z| < 1$. Employing the resolvent identity and \eqref{eq: fmb estimate 1}, we compute
    \begin{align*}
        F_z - \widehat{F}_z = P \Big( 2z (V_\omega - z)^{-1} (1- e^{i \alpha}) V_\omega \, P \, (U_\omega -z)^{-1} \Big) P.
    \end{align*}
    We use the definitions of $F_z$ and $\widehat{F}_z$ to obtain on the range of $P$:
    \begin{align*}
        F_z - \widehat{F}_z = \frac{1}{2} \, (e^{i \alpha} -1) \, (1 + \widehat{F}_z) \, (1 - F_z). 
    \end{align*}
    Following the steps from \cite{HJS:09}, we define $m(\alpha) = i \frac{1+e^{i \alpha}}{1-e^{i \alpha}} = - \cot \big(\frac{\alpha}{2}\big)$ for $\alpha \notin \pi \Z$ and obtain from a direct computation:
    \begin{align} \label{eq: fmb estimate 2}
        F_z = -i \Big( -i \widehat{F}_z + m(\alpha) \Big)^{-1} -i \Big( -i \widehat{F}_z^{-1} - \frac{1}{m(\alpha)} \Big)^{-1}.
    \end{align}
    Using $V_\omega^{-1} \, U_\omega P = e^{i \alpha} P$ and unitarity of $V_\omega$, we obtain:
    \begin{align*}
        \langle x^{(i)} \, | \, (U_\omega+z) \, (U_\omega -z)^{-1} \, y^{(j)} \rangle = \langle V_\omega^* \, x^{(i)} \, | \, F_z V_\omega^* \, y^{(j)} \rangle.
    \end{align*}
    This allows us to estimate the integral in \eqref{eq: bound on U plus z U minus z}:
    \begin{align*}
        &\int_{\T^2} | \langle x^{(i)} \, | \, (U_\omega+z) \, (U_\omega - z )^{-1} y^{(j)} \rangle |^s \; d\mu( \omega_{x}) \, d\mu(\omega_{y}) \leq \| \tau \|_\infty^2 \int_{-\pi}^\pi \int_0^{2\pi} \| F_z \|^s \; d\alpha \, d\beta .
    \end{align*}
    We stress that while $F_z$ depends on $\alpha$, $\widehat{F}_z$ does not. Note that for $s < 1$, we have $|a+b|^s \leq |a|^s + |b|^s$ for all $a,b \in \R$. Plugging equation \eqref{eq: fmb estimate 2} into the integral above, we obtain
    \begin{align} \label{eq: fmb estimate 3}
        \int_0^{2\pi} \| F_z \|^s \; d\alpha \leq \int_0^{2\pi} \| \big(-i \widehat{F}_z + m(\alpha) \big)^{-1} \|^s \; d\alpha + \int_0^{2\pi} \| \Big(-i \widehat{F}_z^{-1} - \frac{1}{m(\alpha)} \Big)^{-1} \|^s \; d\alpha .
    \end{align}
    Using that $\widehat{F}_z$ does not depend on $\alpha$, we substitute $x=m(\alpha)$ to bound the first integral on the right hand side of \eqref{eq: fmb estimate 3}:
    \begin{align*}
        \int_0^{2\pi} \| \big(-i \widehat{F}_z + m(\alpha) \big)^{-1} \|^s \; d\alpha &= \lim_{R \to \infty} \int_{-R}^{R} \frac{2}{1+x^2} \, \| \big(-i \widehat{F}_z + x \big)^{-1} \|^s \; dx \\ &\leq 2 \sum_{n \in \Z} \frac{1}{1+(|n|-1)^2} \, \int_n^{n+1} \| (-i \widehat{F}_z + x)^{-1} \|^s \, dx.
    \end{align*}
    A similar expression holds for the second term in \eqref{eq: fmb estimate 3}. The operators $-i \widehat{F}_z$ and $-i \widehat{F}^{-1}_z$ are defined on $\text{Ran}(P)$, which has dimension six and are thus maximally dissipative. Since $\text{Ran}(P)$ is two-dimensional in \cite{HJS:09}, we have to extend their results to any dissipative operator on a finite dimensional space, using Lemma \ref{lemma: dissipative integral bound} from \cite{AENSS:2005}: By mapping an arbitrary orthonormal basis of $\text{Ran}(P)$ to the standard basis of $\C^6$, we view $-i \widehat{F}_z$ and $-i \widehat{F}^{-1}_z$ as operators on $\C^6$. Due to the map $\text{Ran}(P) \to \C^6$ being unitary, the resulting operators on $\C^6$ are maximally dissipative and $\| (-i \widehat{F}_z -x)^{-1} \|$ and $\|(-i \widehat{F}^{-1}_z -x)^{-1}\|$ are invariant. We write $A$ for $-i \widehat{F}_z$ and $-i \widehat{F}^{-1}_z$ and estimate the integrals in \eqref{eq: fmb estimate 3} for $|z| > 1$. Using the layer cake representation, we rewrite the integral of $\|(A+x)^{-1}\|^s$ on some unit interval:
    \begin{align*}
        &\int_n^{n+1} \| (A + x)^{-1} \|^s \, dx = \int_n^{n+1} \int_0^\infty \mathbbm{1}_{ \{ \| (A+x)^{-1} \|^s > t \} } (x) \, dt \, dx \\
        &= \int_n^{n+1} \int_0^\infty s \, t^{s-1} \, \mathbbm{1}_{ \{ \| (A+x)^{-1} \| > t \} } (x) \, dt \, dx = \int_0^\infty s \, t^{s-1} \int_n^{n+1} \mathbbm{1}_{ \{ \| (A+x)^{-1} \| > t \} } (x) \, dx \, dt \\
        &= \int_0^1 s \, t^{s-1} \int_n^{n+1} \mathbbm{1}_{ \{ \| (A+x)^{-1} \| > t \} } (x) \, dx \, dt + \int_1^\infty s \, t^{s-1} \int_n^{n+1} \mathbbm{1}_{ \{ \| (A+x)^{-1} \| > t \} } (x) \, dx \, dt .
    \end{align*}
    In the first term, we bound the inner integral by $1$, and use that $\int_0^1 s \, t^{s-1} \, dt < \infty$ for $s > 0$. For the second term, we apply Lemma \ref{lemma: dissipative integral bound} below with $M_1 = I = M_2$:
    \begin{align*}
        \int_n^{n+1} \mathbbm{1}_{ \{ \| (A+x)^{-1} \| > t \} } (x) \, dx = | \left\{ x \in [n, n+1] \, \text{ s.t. } \|(A+x)^{-1} \| > t \right\} | \leq \frac{c}{t},
    \end{align*}
    where the constant $c$ is independent of $A$ and $n$ (but does depend on the dimension of $\text{Ran}(P)$). We use that $\int_1^\infty t^{s-2} \, dt < \infty$ for $s < 1$ to bound the second term. Thus, there exists a constant $C(s)$ such that
    \begin{align*}
        \int_n^{n+1} \| (A + x)^{-1} \|^s \, dx \leq C(s)
    \end{align*}
    for all maximally dissipative $A$, where $C(s)$ is independent of $n$ and $\text{Ran}(P)$. Using that $\sum_{n \in \Z} \frac{1}{1 + (|n|-1)^2} < \infty$, we can therefore bound the right hand side of \eqref{eq: fmb estimate 3} for all $|z| > 1$. If $|z| < 1$, we drop a minus sign in both norms in \eqref{eq: fmb estimate 3} and use that $i \widehat{F}_z$ and $i \widehat{F}_z^{-1}$ are dissipative. Repeating the arguments from above yields the bound for all $|z| < 1$.

    The case $x^{(i)} = y^{(j)}$ is significantly easier: We can directly take $\alpha = \omega_{x}$ and do not need $\beta$. The operators $F_z$ and $\widehat{F}_z$ act on $\C^3$, while all other estimates still hold. The integrals in \eqref{eq: fmb estimate 3} can be similarly bounded by Lemma \ref{lemma: dissipative integral bound}.
\end{proof}

We state a simplified version of Lemma 3.1 from \cite{AENSS:2005}, which we use to bound the integrals in the proof of Theorem \ref{thm: frac mom bound}.
\begin{lemma} \label{lemma: dissipative integral bound}
    Let $\Hp$ be a separable Hilbert space, $A$ be a maximally dissipative operator with strictly positive imaginary part and $M_1, M_2: \Hp \to \Hp$ be Hilbert-Schmidt operators. Then there exists a constant $c$ independent of $A$, $M_1$ and $M_2$ such that for any $t > 0$:
    \begin{align*}
        | \left\{ x \in \R \, \text{s.t. } \|M_1 \, (A+x)^{-1} \, M_2 \| > t \right\} | \leq c \, \|M_1\|_{HS} \, \|M_2\|_{HS} \, \frac{1}{t},
    \end{align*}
    where $| \cdot |$ denotes the Lebesgue measure.
\end{lemma}

\subsection{The boundary of a box} \label{sec: box definition}
We want to define a "box" $\Lambda_L$ for any size $L \in \N^2$, whose sides have lengths $L_1$ and $L_2$, such that the Quantum Walker is unable to cross the boundary of $\Lambda_L$. Restricting the Walker to some box $\Lambda_L$ is achieved by changing the coin matrix at specific lattice sites on the boundary of $\Lambda_L$. In other words, we want to obtain unitary operators $U_\omega^{(L)} = U_\omega^{\Lambda_L} \oplus U_\omega^{\Lambda_L^C}$ and subspaces $\Hp_L \oplus \Hp_L^C = \Hp$ such that $\Hp_L$, respectively $\Hp_L^C$, are invariant under $U_\omega^{\Lambda_L}$, respectively $U_\omega^{\Lambda_L^C}$. Note that we call a subspace $\Hp' \subset \Hp$ invariant under $U$ if $U \Hp' \subset \Hp'$. Recalling the definition of $\Hp^{j,k}$ \eqref{eq:def_H_jk}, we use that $C_0$ induces a fully localized Quantum Walk, see section \ref{sec: model and results}, and define the invariant subspaces:
\begin{align*}
    \Hp_L = \bigoplus_{\substack{-L_1  \leq j \leq L_1 -1 \\ -L_2  \leq k \leq L_2 - 1 \\ j+k > -L_1-L_2}} \Hp^{j,k}, \;\; \Hp_L^C = \Hp \setminus \Hp_L.
\end{align*}
The choice $j+k>-L_1-L_2$ is not necessary, but simplifies the structure of $\Lambda_L$, see Figure \ref{fig:boundary coins with origin}. We call the number of $\Gamma_A$-vertices in $\Lambda_L$ the volume of $\Lambda_L$, that is:
\begin{align} \label{eq:def vol LambdaL}
    \text{vol}\left(\Lambda_L\right) = 4L_1 L_2 -1\;\; \text{ and } \;\; |L| = \sqrt{L_1^2 + L_2^2}.
\end{align}
To obtain a Quantum Walk such that these two subspaces are invariant, we need to change the coin matrix at specific lattice sites from $C$ to $C_0$. In particular, we use the coin matrix $C_0$ at all $\Gamma_B$ sites in 
\begin{align*}
    \Gamma_{C_0}^{(L)} = \Big\{ \ket{j,k} &\otimes \begin{pmatrix}
        0 \\ 1
    \end{pmatrix} \text{ s.t. } \Big( -L_1 \leq j \leq L_1-1, k = L_2-1 \Big) \text{ or } \\ &\Big( -L_1+1  \leq j \leq L_1, k = -L_2 -1 \Big) \text{ or } \Big( j = L_1, -L_2  \leq k \leq L_2 -2 \Big) \\ &\text{ or } \Big( j = -L_1, -L_2  \leq k \leq L_2 -2 \Big) \Big\}.
\end{align*}
We stress that lattice sites in $\Gamma_{C_0}^{(L)}$ will now be split between $\Hp_L$ and $\Hp_L^C$ depending on their coin state. 
\begin{figure} [h!]
    \centering
    \includegraphics[scale=0.4]{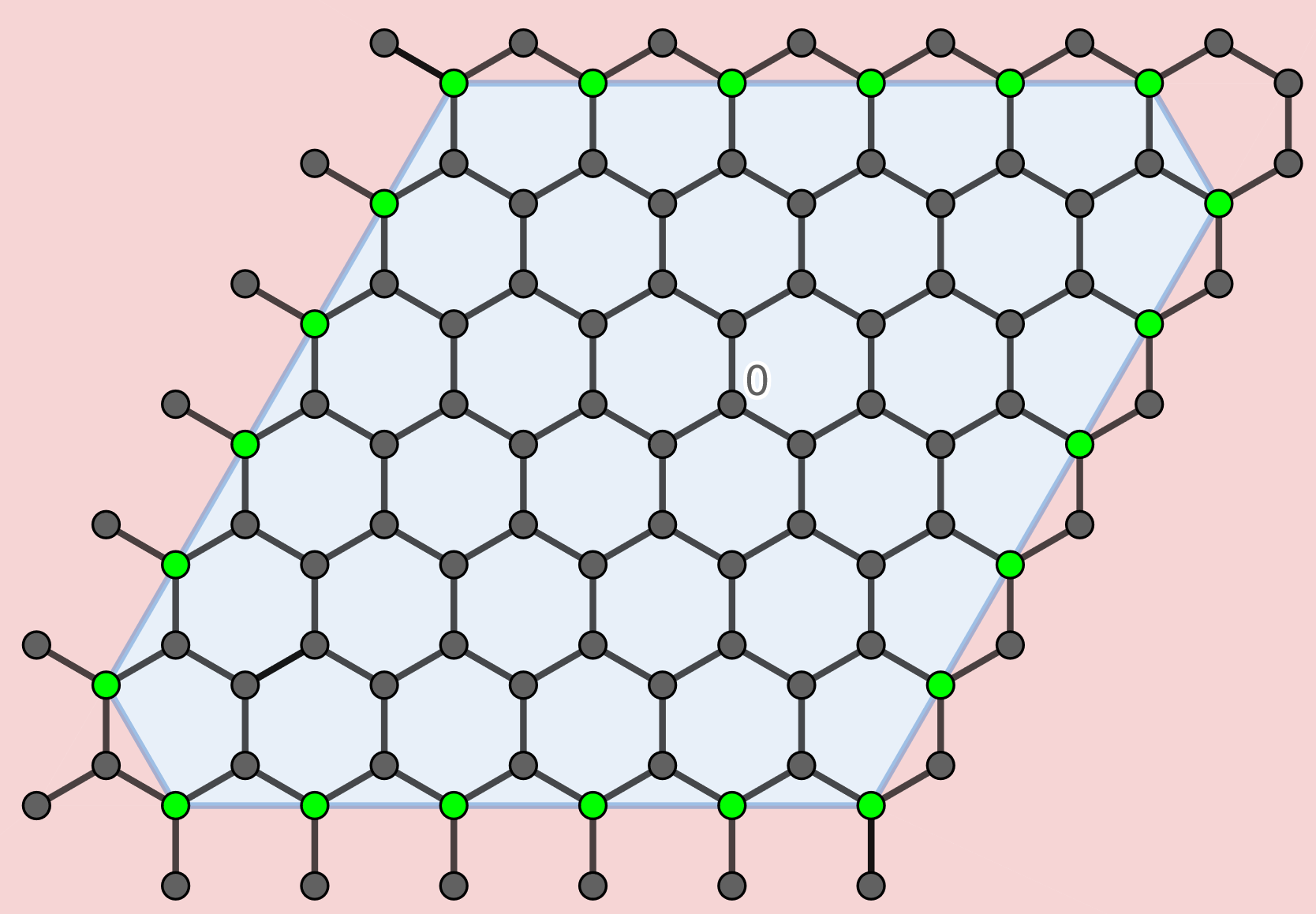}
    \caption{The coin matrices are changed from $C$ to $C_0$ on all green lattice sites. The subspace $\Hp_L$ for $L=(3,3)$ is illustrated as the inner blue area, whereas $\Hp_L^C$ is the exterior red area. Note that sites in $\Gamma_{C_0}^{(L)}$ are split depending on the coin state.}
    \label{fig:boundary coins with origin}
\end{figure}

\noindent The new coin operator is defined as
\begin{align*}
    \mathcal{C}^{\Lambda_L} (\ket{x} \otimes e_j) = \begin{cases} 
        \ket{x} \otimes C_0 \, e_j & \text{if } x \in \Gamma_{C_0}^{(L)} \\
        \ket{x} \otimes C e_j & \text{else}
    \end{cases}.
\end{align*}
We define $U_\omega^{(L)} = D_\omega S \mathcal{C}^{\Lambda_L}$. Note that to pass from $\Lambda_L$ to $\Lambda_L^C$, the Walker needs to cross a site in $\Gamma_{C_0}^{(L)}$. However, due to the choice of coin matrix, the Walker is reflected back at these sites and thus prevented from crossing the boundary of $\Lambda_L$. Denoting by $U_\omega^{\Lambda_L}$ and $U_\omega^{\Lambda_L^C}$ the restrictions of $U_\omega^{(L)}$ to $\Hp_L$ and $\Hp_L^C$, we conclude that $U_\omega^{(L)} = U_\omega^{\Lambda_L} \oplus U_\omega^{\Lambda_L^C}$ and $\Hp_L$, $\Hp_L^C$ are invariant under $U_\omega^{\Lambda_L}$, respectively $U_\omega^{\Lambda_L^C}$.

We note that our construction can be shifted by any lattice shift $v \in \Gamma_A$. We can therefore also consider the operator $U$ restricted to some box $\Lambda_L + v$. Furthermore, we stress that the operator $U_\omega^{(L)}$ is unitary. If we had simply projected onto a subset of $l^2(\Gamma) \otimes \C^3$ (i.e. $U^{\Lambda_L} = PUP + (1-P)U(1-P)$) we would have lost unitarity.

\subsection{The transition operator} \label{sec: trans op}
We define the transition operator:
\begin{align*}
    T_\omega^{(L)}(C) = U_\omega(C) - U_\omega^{(L)}(C) \in \mathcal{B}(\Hp).
\end{align*}
We note that $T_\omega^{(L)}(C)$ applied to some $x^{(i)}$ can only be non-zero if the coin matrix is changed at $x \in \Lambda$. We will show that its norm is bounded by the difference $\|C - C_0 \|_\infty$ independently of $L$. For this, we will make use of the following Lemma, see e.g. \cite{Kato:66}:
\begin{lemma} \label{lem:SchurTest}
    Let $A$ be an infinite matrix such that $\xi(A) = \sup_{i \in \Z} \sum_{j \in \Z} |A_{i,j}| < \infty$ and $\eta(A) = \sup_{j \in \Z} \sum_{i \in \Z} |A_{i,j}| < \infty$, then $A$ is a bounded operator on $l^2(\Z)$ with $\|A\|_{l^2}^2 \leq \xi(A) \, \nu(A)$.
\end{lemma} \noindent
This allows us to prove the following Proposition:
\begin{proposition} \label{propo:norm trans op}
    The norm of $T_\omega^{(L)}(C)$ is bounded independently of $L$, that is
    \begin{align*}
        \|T_\omega^{(L)}(C) \| \leq 3 \|C- C_0 \|_\infty .
    \end{align*}
\end{proposition}

\begin{proof}
    Since $D_\omega$ and $S$ are unitary, we have:
    \begin{align*}
        \|T_\omega^{(L)} \| = \| D_\omega \, S \, \big( \mathcal{C} - \mathcal{C}^{\Lambda_L} \big) \| = \|\mathcal{C} - \mathcal{C}^{\Lambda_L} \|.
    \end{align*}
    As explained before, the transition operator is only non-zero on sites in $\Gamma_{C_0}^{(L)}$. Thus, we compute for any $x^{(i)} \in \Hp$ with $x \in \Gamma_{C_0}^{(L)}$:
    \begin{align*}
        \sum_{y^{(j)} \in \Hp} |\langle x^{(i)} \, | \, \mathcal{C} - \mathcal{C}^{\Lambda_L} \, y^{(j)} \rangle | = \sum_{j=1}^3 |\bra{e_i} C-C_0 \, e_j \rangle | \leq 3 \| C-C_0 \|_\infty.
    \end{align*}
    Therefore, using the notation of Lemma \ref{lem:SchurTest}:
    \begin{align*}
        \xi \big( \mathcal{C} - \mathcal{C}^{\Lambda_L} \big) \leq 3 \, \|C-C_0\|_\infty \;\; \text{ and } \;\; \eta \big( \mathcal{C} - \mathcal{C}^{\Lambda_L} \big) \leq 3 \, \|C-C_0\|_\infty.
    \end{align*}
    Applying Lemma \ref{lem:SchurTest}, this yields $\| T_\omega^{(L)}(C) \| \leq 3 \|C- C_0\|_\infty$.
\end{proof}

\subsection{Spectral gaps of the fully localized operator} \label{sec: spec gaps full loc op}
The subspaces $\Hp^{j,k}$ are invariant under the action of $U_\omega(C_0)$, as explained in section \ref{sec: model and results}. We fix some $\Hp^{j,k}$ and label its 4 random variables $\omega_x$ as shown in Figure \ref{fig:inv subsp with rv}:
\begin{figure} [H]
    \centering
    \includegraphics[scale=1.6]{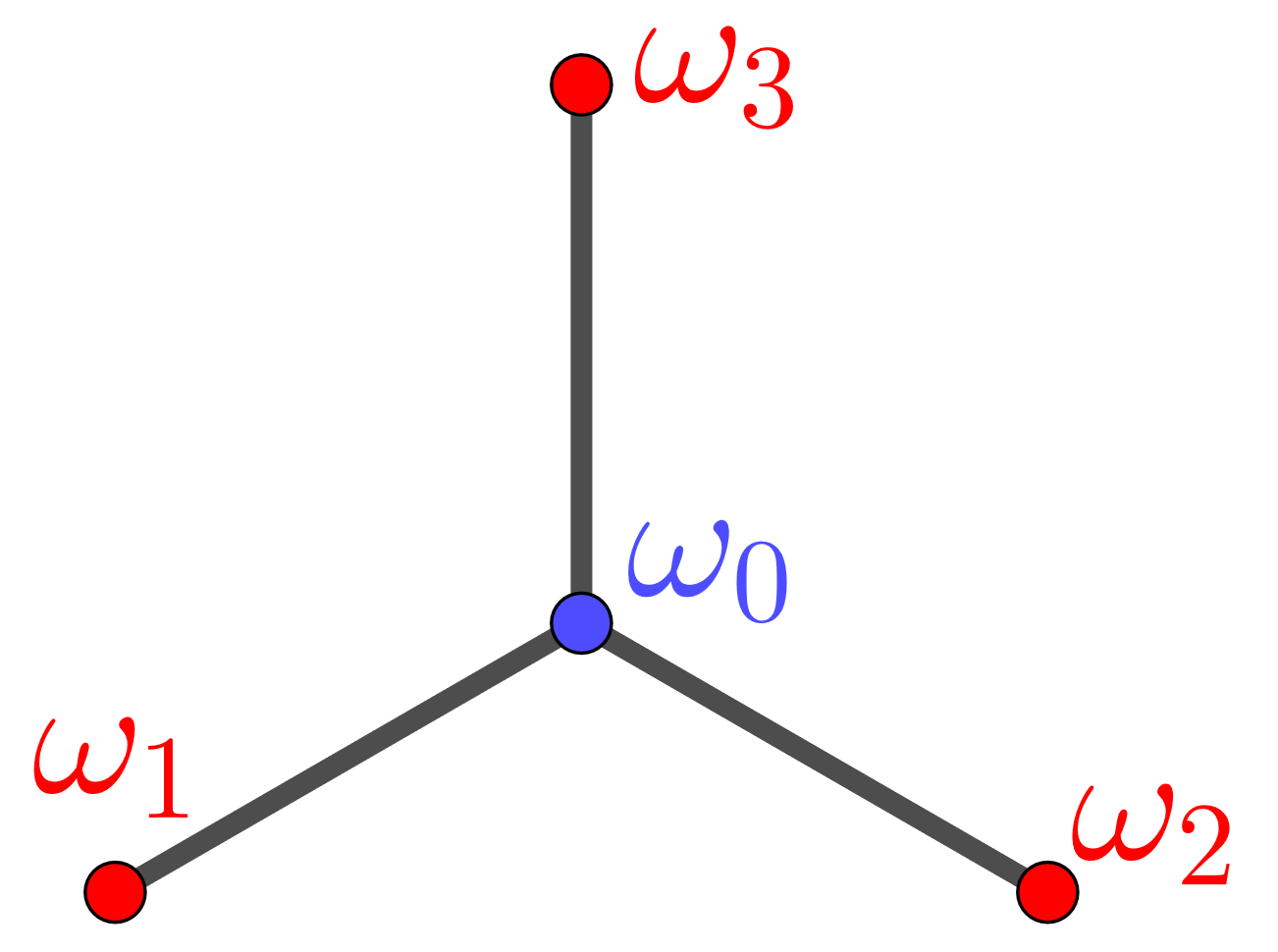}
    \caption{Labeling of the random variables $\omega_x$ for a fixed subspace $\Hp^{j,k}$. We note that if the random variable $\omega_0$ is at site $\ket{j,k} \otimes f_1$, then $\omega_1$ is at $\ket{j,k-1}\otimes f_2$, $\omega_2$ is at $\ket{j+1,k-1} \otimes f_2$ and $\omega_3$ is at $\ket{j,k} \otimes f_2$.}
    \label{fig:inv subsp with rv}
\end{figure} \noindent
We let $U_\omega^{j,k}(C_0)$ denote the restriction of $U_\omega(C_0)$ to the subspace $\Hp^{j,k}$. Using the labeling of Figure \ref{fig:inv subsp with rv} we have $U_\omega^{j,k}(C_0)^6 = e^{i(3 \omega_0 + \omega_1 + \omega_2 + \omega_3)} \, I_{\Hp^{j,k}}$. Moreover, we can express $U_\omega^{j,k}(C_0)$ in the natural ONB of $\Hp^{j,k}$, see \eqref{eq:def_H_jk}. This allows us to directly compute its characteristic polynomial: $\chi_{U_\omega^{j,k}(C_0)}(\lambda) = \lambda^6 - e^{i (3 \omega_0 + \omega_1 + \omega_2 + \omega_3)}$. The eigenvalues of $U_\omega^{j,k}(C_0)$ are the sixth roots of unity of $e^{i (3 \omega_0 + \omega_1 + \omega_2 + \omega_3)} = \Theta_\omega^{j,k}$. Additionally, we can show that the random variables $\left(\Theta_\omega^{j,k} \right)_{j,k \in \Z}$ are independent and uniformly distributed random variables on $\Sp^1 \subset \C$. This claim is obtained by using Lemma 4.1 in \cite{ABJ:2010} and relies on the fact that the RVs are uniformly distributed and the RVs on each $\Gamma_A$ vertex contribute to exactly one $\Theta_\omega^{j,k}$ random variable.

Our next goal is the following Proposition, which estimates the probability of spectral gaps for the fully localized operator:

\begin{proposition} \label{propo: spect gaps fully loc}
    There exists $c>0$ such that for any $L \in \N^2$, $\eta > 0$, $v \in \Gamma_A$:
    \begin{align*}
        \Pp \Big( \text{dist} \big(z, \sigma \left( U_\omega^{\Lambda_L + v} (C_0) \right) \big) \leq \eta \Big) \leq c \, \eta \, \text{vol} \left( \Lambda_L \right)
    \end{align*}
    for $\eta \, \text{vol} \left( \lambda_L \right)$ small enough, uniformly in $z \in \C \setminus \Sp^1$.
\end{proposition}

The proof of Proposition \ref{propo: spect gaps fully loc} follows the steps of \cite{ABJ:2012} (Proposition 2.2) and does not require any significant model-specific changes.

\begin{proof}
    We can assume without loss of generality that $v=0$. Let $l$ denote the natural spherical measure on $\Sp^1 \subset \C$, normalized such that $l(\Sp^1) = 1$ and consider any arc $A \subset \Sp^1$ with length $l(A) < \frac{1}{6}$. Since the eigenvalues of $U_\omega^{j,k}(C_0)$ are the sixth roots of unity of the uniformly distributed random variable $\Theta_\omega^{j,k}$, the probability for the arc $A$ to lie outside of the spectrum of $U_\omega^{j,k}(C_0)$ is $1-6l(A)$. $\Hp_L$ contains $\text{vol}(\Lambda_L) = 4 L_1 L_2 -1$ blocks $\Hp^{j,k}$, which are invariant under $U_\omega^{\Lambda_L}(C_0)$. Using that the random variables $\Theta_\omega^{j,k}$ are independent, it follows for the spectrum of $U_\omega^{\Lambda_L}(C_0)$:
    \begin{align} \label{eq: spect intersect arc}
        \Pp \Big( \sigma \big(U_\omega^{\Lambda_L}(C_0) \big) \cap A = \varnothing \Big) = \big( 1 - 6l(A) \big)^{\text{vol}(\Lambda_L)}.
    \end{align}
    We easily obtain the following geometric property from the law of cosines: Let $B_\eta(z) \subset \C$ be the closed ball centered at $z \in \C$ of radius $\eta \in (0,1)$ , then
    \begin{align*}
        l \big( B_\eta(z) \cap \Sp^1 \big) < \frac{\eta}{2}.
    \end{align*}
    Using this, together with \eqref{eq: spect intersect arc}, it follows that
    \begin{align*}
        \Pp \Big( \text{dist} \big( z, \sigma \left(U_\omega^{\Lambda_L}(C_0) \right) \big) > \eta \Big) = \Big( 1 - 6 \, l \big( B_\eta(z) \cap \Sp^1 \big) \Big)^{\text{vol}(\Lambda_L)} \geq \big(1- 3 \, \eta \big)^{\text{vol}(\Lambda_L)}.
    \end{align*}
    Note that in the last step we used that $\eta$ is small. Finally, this gives us
    \begin{align*}
        \Pp \Big( \text{dist} \big(z, \sigma \left( U_\omega^{\Lambda_L} (C_0) \right) \big) \leq \eta \Big) \leq 1 - \big(1- 3 \, \eta \big)^{\text{vol}(\Lambda_L)}.
    \end{align*} 
    To conclude, we use the following estimate (see \cite{ABJ:2012}):
    \begin{align*}
        \sup_{y \in \left(0, \frac{1}{2}\right)} \left| \big(1-xy \big)^\frac{1}{y} -1 + x \right| \leq \frac{x^2}{2} \;\;\; \forall \, x \in [0,1].
    \end{align*}
    Using $x = 3 \, \eta \, \text{vol}(\Lambda_L)$ and $y = \frac{1}{\text{vol}(\Lambda_L)}$, this finishes the proof.
\end{proof}

\subsection{Polynomial decay of the boxed resolvent} \label{sec: poly dec box res}
We denote by $R_\omega^{\#}(C,z)$ the resolvent of $U_\omega^{\#}(C)$ for any coin $C \in U(3)$ and any $z$ in its resolvent set. Unless it is unclear, we will omit the dependence on $z$. Furthermore, we define the following equivalence relation on the elements of the orthonormal basis $\left( x^{(i)} \right)_{x \in \Gamma, i=1,2,3}$ of $\Hp$:
\begin{align*}
    x^{(i)} \sim y^{(l)} \iff \exists \, j,k \in \Z \; \text{ s.t. } \; x^{(i)}, y^{(l)} \in \Hp^{j,k}.
\end{align*}
Our goal is the following Proposition:
\begin{proposition} \label{propo: poly decay resolvent}
    For any $s \in (0,1)$, $p > \frac{1}{1-s}$ and $a \geq 0$ there exists $c > 0$ such that:
    \begin{align*}
        \E \Big( \left| \langle x^{(i)} \, | \, R_\omega^{\Lambda_L+v}(C,z) \, y^{(j)} \rangle \right|^s \Big) \leq \frac{c}{|L|^a} 
    \end{align*}
    uniformly in $L \in \N^2$ with $\min(L_1, L_2) \geq 3$, $z \in \C \setminus \Sp^1$, all coin matrices $C \in U(3)$ such that $\|C-C_0\|_\infty \leq \frac{1}{|L|^{2ap + 4 + \frac{a}{s}}}$, all translations $v \in \Gamma_A$ and all $x,y \in \Lambda_L + v$ with $x^{(i)} \nsim y^{(j)}$.
\end{proposition}

We prove Proposition \ref{propo: poly decay resolvent} by following the steps of \cite{ABJ:2012} (Proposition 2.3) without any significant changes.

\begin{proof}
    It suffices to consider only the case $v=0$. The case $a=0$, i.e. the expectation is bounded by some constant, follows from Theorem \ref{thm: frac mom bound}, noting that its proof works as well for $U_\omega^{\Lambda_L}$. Note that $\Hp^{j,k}$ is invariant under $U_\omega^{\Lambda_L}(C_0)$, which implies
    \begin{align} \label{eq:res local on Hjk}
        \langle x^{(i)} \, | \, R_\omega^{\Lambda_L}(C_0) \, y^{(j)} \rangle = 0 \;\; \text{ if } \; x^{(i)} \nsim y^{(j)}.
    \end{align}
    It is due to the resolvent identity that
    \begin{align} \label{eq:res id C - C0}
        R_\omega^{\Lambda_L}(C) = R_\omega^{\Lambda_L}(C_0) + R_\omega^{\Lambda_L}(C) \Big( U_\omega^{\Lambda_L}(C_0) - U_\omega^{\Lambda_L}(C) \Big) R_\omega^{\Lambda_L}(C_0).
    \end{align}
    Since the shift operator $S$ only permits nearest neighbor hopping, we have:
    \begin{align} \label{eq:U nearest nbh}
        \langle x^{(i)} \, | \, U_\omega^{\Lambda_L}(C) \, y^{(j)} \rangle = 0 \;\; \text{ if } \; |x-y| \neq 1.
    \end{align}
    Moreover, due to the unitarity of $D_\omega$ and $S$, we can use Lemma \ref{lem:SchurTest} to estimate the norm of $\mathcal{C}^{\Lambda_L}(C) - \mathcal{C}^{\Lambda_L}(C_0)$ to obtain:
    \begin{align} \label{eq:norm Uc - Uc0}
        \| U_\omega^{\Lambda_L}(C) - U_\omega^{\Lambda_L}(C_0) \| \leq 3 \|C- C_0\|_\infty.
    \end{align}
    Equations \eqref{eq:res local on Hjk} - \eqref{eq:U nearest nbh} now yield for $x^{(i)}, y^{(j)} \in \Hp_L$ with $x^{(i)} \nsim y^{(j)}$ and every $z \in \C \setminus \Sp^1$:
    \begin{align*}
        &\big| \langle x^{(i)} \, | \, R_\omega^{\Lambda_L}(C) \, y^{(j)} \rangle \big| \\
        &\overset{\eqref{eq:res id C - C0}, \eqref{eq:res local on Hjk}}{= } \big| \langle x^{(i)} \, | \, R_\omega^{\Lambda_L}(C) \Big( U_\omega^{\Lambda_L}(C_0) - U_\omega^{\Lambda_L}(C) \Big) R_\omega^{\Lambda_L}(C_0) \, y^{(j)} \rangle \big| \\ 
        &\overset{\eqref{eq:res local on Hjk}, \eqref{eq:U nearest nbh}}{\leq} \sum_{\substack{v^{(k)}, \tilde{v}^{(l)} \in \Hp_L \text{ s.t} \\ \tilde{v}^{(l)} \sim y^{(j)}, |v-\tilde{v}|=1}} \big| \langle x^{(i)} \, | \, R_\omega^{\Lambda_L}(C) \, v^{(k)} \rangle \, \langle v^{(k)} \, | \, \Big( U_\omega^{\Lambda_L}(C_0) - U_\omega^{\Lambda_L}(C) \Big) \, \tilde{v}^{(l)} \rangle \\ 
        &\hphantom{{} \overset{\eqref{eq:res local on Hjk}, \eqref{eq:U nearest nbh}}{=} \sum_{\substack{v^{(k)}, \tilde{v}^{(l)} \in \Hp_L \text{ s.t} \\ \tilde{v}^{(l)} \sim y^{(j)}, |v-\tilde{v}|=1}} \big| } \langle \tilde{v}^{(l)} \, | \, R_\omega^{\Lambda_L}(C_0) \, y^{(j)} \rangle \big|.
    \end{align*}
    Using \eqref{eq:norm Uc - Uc0} we can bound the second scalar product by $3 \|C-C_0\|_\infty$. Since $U_\omega^{\Lambda_L}$ is a normal operator, we can bound the third scalar product by $\frac{1}{\text{dist} \big( z, \sigma(U_\omega^{\Lambda_L}(C_0)) \big) }$. Finally, we note that the number of terms in the sum above is finite and independent of $L$. Thus, there exists a constant $c > 0$ such that
    \begin{align} \label{eq:resolvent bound}
        \big| \langle x^{(i)} \, | \, R_\omega^{\Lambda_L}(C) \, y^{(j)} \rangle \big| &\leq c \, \frac{\|C - C_0 \|_\infty}{\text{dist} \big( z, \sigma(U_\omega^{\Lambda_L}(C_0)) \big)} \, \sup_{\substack{v^{(k)}, \tilde{v}^{(l)} \in \Hp_L \\ \text{ s.t. } \tilde{v}^{(l)} \sim y^{(j)}, \\ |v-\tilde{v}|=1}} \big| \langle x^{(i)} \, | \, R_\omega^{\Lambda_L}(C) \, v^{(k)} \rangle \big|.
    \end{align}
    For any $\eta > 0$ we define the set of "good" events:
    \begin{align*}
        G_\eta(z) = \left\{ \omega \in \T^{\Gamma} \text{ s.t. } \text{dist}\Big(z, \sigma \big( U_\omega^{\Lambda_L}(C_0) \big) \Big) > \eta \right\}.
    \end{align*}
    We let $c>0$ denote some constant that is independent of $L$ and may change from line to line. By Proposition \ref{propo: spect gaps fully loc} we have $\Pp \left( G_\eta(z)^C \right) \leq c \, \eta \, \text{vol}\left(\Lambda_L \right)$. We take $p > 1$ as in the statement of Proposition \ref{propo: poly decay resolvent} and $q > 1$ such that $\frac{1}{p} + \frac{1}{q} = 1$ and apply Hölder's inequality:
    \begin{align} \label{eq: first part of expectation}
    \begin{split}
        \E &\left( \chi_{G_\eta(z)^C} \, \left| \langle x^{(i)} \, | \, R_\omega^{\Lambda_L}(C) \, y^{(j)} \rangle \right|^s \right) \\
        &\leq \underbrace{\E \left( \left| \chi_{G_\eta(z)^C} \right|^p \right)^\frac{1}{p}}_{=\Pp\left(G_\eta(z)^C \right)^\frac{1}{p} } \, \underbrace{\E \left( \left| \langle x^{(i)} \, | \, R_\omega^{\Lambda_L}(C) \, y^{(j)} \rangle \right|^{sq} \right)^\frac{1}{q}}_{\overset{\eqref{eq: fract mom bound}}{\leq} \text{const}} \leq c \big(\eta \, \text{vol} (\Lambda_L) \big)^\frac{1}{p}.
    \end{split}
    \end{align}
    Note that to use estimate \eqref{eq: fract mom bound}, we needed that $sq < 1$, which is equivalent to the assumption $p > \frac{1}{1-s}$ from the statement of Proposition \ref{propo: poly decay resolvent}. Setting $\eta = \frac{1}{\text{vol}\left(\Lambda_L\right) |L|^{ap}}$ proves the claim on the set $G_\eta(z)^C$.

    In order to get a bound on the set $G_\eta(z)$, we use the following result from perturbation theory:
    \begin{lemma} \label{lem: dist spect A-B}
        Let $A,B$ be normal operators on some Hilbert space and $z \in \C$ be in both of their resolvent sets. For any $\eta > 0$, if $\text{dist} \big(z, \sigma(A) \big) > \eta$ and $\|A-B\| \leq \frac{\eta}{2}$, then $\text{dist} \big(z, \sigma(B) \big) > \frac{\eta}{2}$.
    \end{lemma} \noindent
    Thus, if the coin $C$ is chosen such that $3 \|C-C_0\|_\infty \leq \frac{\eta}{2}$, then 
    \begin{align} \label{eq:spectral implication}
        \text{dist} \Big(z, \sigma \big( U_\omega^{\Lambda_L}(C_0) \big) \Big) > \eta \Longrightarrow \text{dist} \Big(z, \sigma \big( U_\omega^{\Lambda_L}(C) \big) \Big) > \frac{\eta}{2}.
    \end{align}
    Now we can estimate for all events in $G_\eta(z)$:
    \begin{align*}
        \chi_{G_\eta(z)}(\omega) \, &\big| \langle x^{(i)} \, | \, R_\omega^{\Lambda_L}(C,z) \, y^{(j)} \rangle \big|^s \\
        &\overset{\eqref{eq:resolvent bound}}{\leq} \chi_{G_\eta(z)}(\omega) \, \Big| c \, \frac{\|C - C_0 \|_\infty}{\text{dist} \big( z, \sigma(U_\omega^{\Lambda_L}(C_0)) \big)} \, \sup_{\substack{v^{(k)}, \tilde{v}^{(l)} \in \Hp_L \\ \text{ s.t. } \tilde{v}^{(l)} \sim y^{(j)}, \\ |v-\tilde{v}|=1}} \underbrace{\big| \langle x^{(i)} \, | \, R_\omega^{\Lambda_L}(C) \, v^{(k)} \rangle \big|}_{\leq \frac{1}{\text{dist}(z, \sigma(U_\omega^{\Lambda_L}(C)))}} \Big|^s.
    \end{align*}
    If $\omega \in G_\eta(z)$, we have $\text{dist} \big( z, \sigma(U_\omega^{\Lambda_L}(C_0)) \big) > \eta$ and therefore by \eqref{eq:spectral implication}:
    \begin{align} \label{eq: almost final bound}
        \chi_{G_\eta(z)}(\omega) \, \left| \langle x^{(i)} \, | \, R_\omega^{\Lambda_L}(C) \, y^{(j)} \rangle \right|^s \leq \chi_{G_\eta(z)}(\omega) \, c \, \bigg( \frac{\|C-C_0\|_\infty}{\eta^2} \bigg)^s.
    \end{align}
    Now, if we choose $C$ such that $\|C-C_0\|_\infty \leq \frac{1}{|L|^{2ap+4+\frac{a}{s}}}$, then together with our choice for $\eta$ we obtain
    \begin{align*}
        \frac{\|C-C_0\|_\infty}{\eta^2} \leq \frac{\text{vol}\left(\Lambda_L\right)^2 |L|^{2ap}}{|L|^{2ap+4+\frac{a}{s}}} = \frac{\text{vol}\left(\Lambda_L\right)^2}{|L|^{4+\frac{a}{s}}} \leq \frac{4}{|L|^\frac{a}{s}}.
    \end{align*}
    In the last step we have used the bound $\text{vol}(\Lambda_L) \leq 2 |L|^2$, which can easily be obtained from \eqref{eq:def vol LambdaL}. We insert this into \eqref{eq: almost final bound}:
    \begin{align} \label{eq: second part of expectation}
        \chi_{G_\eta(z)}(\omega) \, \left| \langle x^{(i)} \, | \, R_\omega^{\Lambda_L}(C) \, y^{(j)} \rangle \right|^s \leq \chi_{G_\eta(z)}(\omega) \, c \, \frac{1}{|L|^a}.
    \end{align}
    The estimates \eqref{eq: first part of expectation} and \eqref{eq: second part of expectation} finish the proof.
\end{proof} \noindent
We remark that in the proof we have used the following two assumptions on $C$:
\begin{align} \label{eq:assumptions on C}
    \|C-C_0\|_\infty \leq \frac{1}{|L|^{2ap+4+\frac{a}{s}}} \;\; \text{and} \;\; \|C-C_0\|_\infty \leq \frac{\eta}{6}.
\end{align}
Using the definition of $\eta$ and the relation $\text{vol}(\Lambda_L) \leq 2 |L|^2$ we can show that the first implies the second whenever $|L|^2 \geq 12$. This is always satisfied due to our choice of $\min(L_1, L_2) \geq 3$. 

\subsection{The iterative step} \label{sec: iterative step}
We start by defining the boundary of both $\Lambda_L$ and its compliment:
\begin{align*}
    \partial \big( \Lambda_L \big) &= \text{span} \left( x^{(i)} \in \Hp_L \text{ s.t. } x \in \Gamma_{C_0}^{(L)} \text{ or } x \text{ is adjacent to some } y \in \Gamma_{C_0}^{(L)} \right), \\
    \partial \big( \Lambda_L^C \big) &= \text{span} \left( x^{(i)} \in \Hp_L^C \text{ s.t. } x \in \Gamma_{C_0}^{(L)} \text{ or } x \text{ is adjacent to some } y \in \Gamma_{C_0}^{(L)} \right).
\end{align*}
The sets are illustrated in Figure \ref{fig:interior and boundary}.
\begin{figure} [h!]
    \centering
    \includegraphics[scale=0.4]{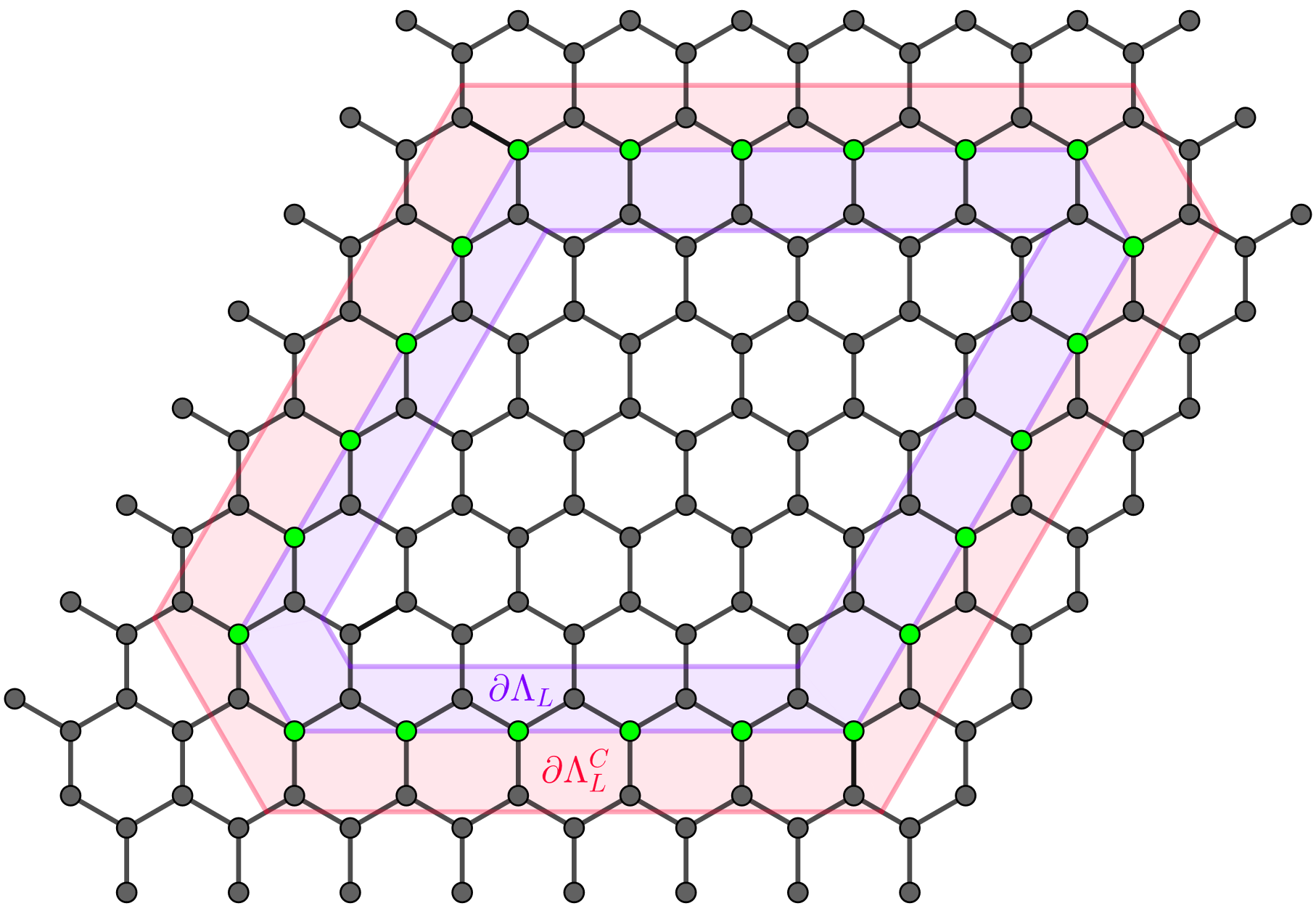}
    \caption{The boundaries of $\Lambda_L$ and $\Lambda_L^C$.}
    \label{fig:interior and boundary}
\end{figure}\\
For any subset $S \subset \Hp$, we define its closure to be:
\begin{align} \label{eq: def N(S)}
    \overline{S} = \left\{x^{(i)} \in \Hp \, | \, \exists \, j \text{ s.t. } x^{(j)} \in S \right\}.
\end{align}
For any $x^{(i)} \in \Hp$ let $[x^{(i)}]$ denote the unique element $y^{(i)} \in \Hp$ such that $y \in \Gamma_A$ and $x^{(i)} \sim y^{(i)}$. In other words, if $x \in \Gamma_A$, then $[x^{(i)}] = x^{(i)}$, and if $x \in \Gamma_B$, then $[x^{(i)}]$ is the unique $\Gamma_A$ vertex that is in the same invariant subspace $\Hp^{j,k}$ as $x^{(i)}$ with the same coin state. This construction allows us to assume without loss of generality that the box $\Lambda_L$ is centered at the origin. We will prove the following:
\begin{proposition} \label{propo: max on boundary}
    There exist $L_0 > 0$ and $q < 1$ such that for all $s \in \left(0, \frac{1}{3} \right)$, $L \in \N^2$ with $min(L_1, L_2) \geq 3$ and $|L| \geq L_0$ there exists $\varphi(L,s) > 0$ such that
    \begin{align*}
        \E \left( \left| \langle x^{(i)} \, | \, R(C,z) \, y^{(j)} \rangle \right|^s \right) \leq q \, \max_{v^{(k)} \in \overline{\partial ( \Lambda_{L+\iota}^C)} + [x^{(i)}] } \E \left( \left| \langle v^{(k)} \, | \, R(C,z) \, y^{(j)} \rangle \right|^s \right)
    \end{align*}
    for all $z \in \C$ with $|z|<1$, $\|C - C_0 \|_\infty \leq \varphi$, $x^{(i)} \in \Hp$ and $y^{(j)} \in \Hp_{L+ 2 \iota}^C + [x^{(i)}]$, where $\iota=(2,2)$ is used to expand the box $\Lambda_L$ from size $L$ to size $L + \iota$.
\end{proposition}
The proof of Proposition \ref{propo: max on boundary} follows mainly the steps of \cite{HJS:09} (Proposition 3.1) with some model-specific changes. It uses two re-sampling arguments that require some adaptations and are detailed in the appendices \ref{sec: first resampling argu} and \ref{sec: second resampling argu}.
\begin{remark} \label{remark: relabeling propo}
    We can relabel $L + \iota$ to $L$ in the statement of Proposition \ref{propo: max on boundary}. We then obtain for $\min(L_1, L_2) \geq 5$, $|L + \iota| \geq L_0$ and $y^{(j)} \in \Hp_{L+\iota}^C$:
    \begin{align*}
        \E \left( \left| \langle x^{(i)} \, | \, R(C,z) \, y^{(j)} \rangle \right|^s \right) \leq q \, \max_{v^{(k)} \in \overline{\partial ( \Lambda_{L}^C)} + [x^{(i)}]} \E \left( \left| \langle v^{(k)} \, | \, R(C,z) \, y^{(j)} \rangle \right|^s \right).
    \end{align*}
\end{remark}
\begin{proof}
    To simplify the notation we assume without loss of generality that $x^{(i)} \in \Hp^{0,0}$. From now on we will omit the dependencies on $z, C$ and $\omega$. We use the resolvent identity twice and obtain:
    \begin{align*}
        R &= R^{(L)} - R^{(L)} \, T^{(L)} \, R = R^{(L)} - R^{(L)} \, T^{(L)} \, \big( R^{(L+ \iota)} - R \, T^{(L+\iota)} \, R^{(L+ \iota)} \big),
    \end{align*}
    where $T^{(L)}$ denotes the transition operator defined in section \ref{sec: trans op}. Using the resolvent identity above, we have:
    \begin{align*}
        \langle x^{(i)} \, | \, R \, y^{(j)} \rangle 
        &= \langle x^{(i)} \, | \, R^{(L)} \, y^{(j)} \rangle - \langle x^{(i)} \, | \, R^{(L)} \, T^{(L)} \, R^{(L+\iota)} \, y^{(j)} \rangle \\
        &+ \langle x^{(i)} \, | \, R^{(L)} \, T^{(L)} \, R \, T^{(L+\iota)} \, R^{(L+\iota)} \, y^{(j)} \rangle.
    \end{align*}
    The first of the three terms is zero, since $x^{(i)} \in \Hp_L$ and $y^{(j)} \in \Hp_L^C$ and both subspaces are invariant under $R^{(L)}$. Similarly, since $y^{(j)} \in \Hp_{L+\iota}^C$ and using that $T^{(L)}$ has range $1$, we conclude that the second scalar product is also zero. Inserting the identity operator between each of the remaining terms, we obtain:
    \begin{align} \label{eq: much before first resampling}
    \begin{split}
        \langle x^{(i)} \, | \, R \, y^{(j)} \rangle 
        = \sum_{\substack{v^{(k)}, \tilde{v}^{(l)},\\ u^{(n)}, \tilde{u}^{(m)} \in \Hp}} &\langle x^{(i)} \, | \, R^{(L)} \, v^{(k)} \rangle \, \langle v^{(k)} \, | \, T^{(L)} \, \tilde{v}^{(l)} \rangle \, \langle \tilde{v}^{(l)} \, | \, R \, u^{(n)} \rangle \\ &\langle u^{(n)} \, | \, T^{(L+\iota)} \, \tilde{u}^{(m)} \rangle \, \langle \tilde{u}^{(m)} \, | \, R^{(L+\iota)} \, y^{(j)} \rangle.
    \end{split}
    \end{align}
    We define the set of vertex pairs crossing the boundary of $\Lambda_L$:
    \begin{align*}
        B_L = \{ (v^{(k)}, \tilde{v}^{(l)}) \in \Hp \times \Hp \text{ s.t. } \langle v^{(k)} \, | \, T^{(L)} \, \tilde{v}^{(l)} \rangle \neq 0 \}.
    \end{align*}
    Since $T^{(L)}$ is only non-zero on sites where the coin matrix is changed from $C$ to $C_0$, it follows that for any pair $(v^{(k)}, \tilde{v}^{(l)}) \in B_L$ we have $v, \tilde{v} \in \partial(\Lambda_L) \cup \partial(\Lambda_L^C)$ and the coin matrix is changed at $\tilde{v}$. Since the norm of the transition operator $T^{(L)}$ is bounded independently of $L$ (see Proposition \ref{propo:norm trans op}), there exists a constant $c_1 > 0$ that depends only on the coin matrix $C$ such that:
    \begin{align} \label{eq: before first resampling}
    \begin{split}
        &\E \big( \big| \langle x^{(i)} \, | \, R \, y^{(j)} \rangle \big|^s \big) \\ &\leq c_1 \, \sum_{\substack{(v^{(k)}, \tilde{v}^{(l)}) \in B_L, \\ (u^{(n)}, \tilde{u}^{(m)}) \in B_{L+\iota}, \\ v^{(k)} \in \Hp_L, \tilde{u}^{(m)} \in \Hp_{L+\iota}^C}} \E \left( \left| \langle x^{(i)} \, | \, R^{(L)} \, v^{(k)} \rangle \right|^s \left| \langle \tilde{v}^{(l)} \, | \, R \, u^{(n)} \rangle \right|^s \left| \langle \tilde{u}^{(m)} \, | \, R^{(L+\iota)} \, y^{(j)} \rangle \right|^s \right).
    \end{split}
    \end{align}
    Note that even though the first and third terms are independent, they are coupled by the second term. In order to decouple the terms, we adapt the so called re-sampling argument from \cite{HJS:09}, section 13, which is based on a strategy developed in \cite{AENSS:2005} and \cite{DNSS:2005} for continuum Anderson-type models. We show in the appendix \ref{sec: first resampling argu} that there exists a constant $c_2 >0$ independent of $L$ such that:
    \begin{align} \label{eq: first resampling}
    \begin{split}
        \E \big( \big| &\langle x^{(i)} \, | \, R \, y^{(j)} \rangle \big|^s \big) \\ &\leq c_2 \, \sum_{u^{(n)} \in \partial (\Lambda_L)} \E \left( \left| \langle x^{(i)} \, | \, R^{(L)} u^{(n)} \rangle \right|^s \right) \, \sum_{\tilde{u}^{(m)} \partial(\Lambda_{L+\iota}^C)} \E \left( \left| \langle \tilde{u}^{(m)} \, | \, R^{(L+\iota)} \, y^{(j)} \rangle \right|^s \right).
    \end{split}
    \end{align}
    Using the resolvent identity $R^{L+\iota} = R + R^{(L+\iota)} \, T^{(L+\iota)} \, R$ on the second term in \eqref{eq: first resampling} and inserting identities between each of the operators, we obtain for fixed $\tilde{u}^{(m)}$:
    \begin{align*}
        \E &\Big(\left| \langle \tilde{u}^{(m)} \, | \, R^{(L+\iota)} \, y^{(j)} \rangle \right|^s \Big) \leq \E \left(\left| \langle \tilde{u}^{(m)} \, | \, R \, y^{(j)} \rangle \right|^s \right) \\
        &+ \E \bigg( \sum_{v^{(k)}, \tilde{v}^{(l)} \in \Hp} \left| \langle \tilde{u}^{(m)} \, | \, R^{(L+\iota)} v^{(k)} \rangle \right|^s \, \left| \langle v^{(k)} \, | \, T^{(L+\iota)} \tilde{v}^{(l)} \rangle \right|^s \, \left| \langle \tilde{v}^{(l)} \, | \, R \, y^{(j)} \rangle \right|^s \bigg).
    \end{align*}
    Since $\tilde{u}^{(m)} \in \Hp_{L+\iota}^C$, we obtain $(v^{(k)}, \tilde{v}^{(l)}) \in B_{L+\iota}$ with $v^{(k)} \in \Hp_{L+\iota}^C$. Bounding $T^{(L+\iota)}$, there exists a constant $c_3>0$ that only depends on $C$ such that
    \begin{align*}
        \E \Big(\big| \langle \tilde{u}^{(m)} \, | \, &R^{(L+\iota)} \, y^{(j)} \rangle \big|^s \Big) \leq \E \left(\left| \langle \tilde{u}^{(m)} \, | \, R \, y^{(j)} \rangle \right|^s \right) \\ 
        &+ c_3 \sum_{\substack{(v^{(k)}, \tilde{v}^{(l)}) \in B_{L+\iota}, \\ v^{(k)} \in \Hp_{L+\iota}^C}} \E \bigg( \left| \langle \tilde{u}^{(m)} \, | \, R^{(L+\iota)} \, v^{(k)} \rangle \right|^s \, \left| \langle \tilde{v}^{(l)} \, | \, R \, y^{(j)} \rangle \right|^s \bigg).
    \end{align*}
    Plugging this into \eqref{eq: first resampling} yields:
    \begin{align} \label{eq: before second resampling}
    \begin{split}
        \E \big( \big| &\langle x^{(i)} \, | \, R \, y^{(j)} \rangle \big|^s \big) \leq c_2 \, \bigg( \sum_{u^{(n)} \in \partial (\Lambda_L)} \E \left( \left| \langle x^{(i)} \, | \, R^{(L)} \, u^{(n)} \rangle \right|^s \right) \bigg) \\ &\times \sum_{\tilde{u}^{(m)} \in \partial (\Lambda_{L+\iota}^C)} \Bigg( \E \left(\left| \langle \tilde{u}^{(m)} \, | \, R \, y^{(j)} \rangle \right|^s \right) \\ &+ c_3 \sum_{\substack{(v^{(k)}, \tilde{v}^{(l)}) \in B_{L+\iota}, \\ v^{(k)} \in \Hp_{L+\iota}^C}} \E \bigg( \left|\langle \tilde{u}^{(m)} \, | \, R^{(L+\iota)} \, v^{(k)} \rangle \right|^s \, \left| \langle \tilde{v}^{(l)} \, | \, R \, y^{(j)} \rangle \right|^s \bigg) \Bigg).
    \end{split}
    \end{align}
    A second re-sampling argument allows us to estimate \eqref{eq: before second resampling}: There exists another constant $c_4 > 0$ independent of $L$ such that
    \begin{align} \label{eq: second resampling}
    \begin{split}
        \E \big( &\big| \langle x^{(i)} \, | \, R \, y^{(j)} \rangle \big|^s \big) \\ &\leq c_4 \, |L| \, \sum_{u^{(n)} \in \partial (\Lambda_L)} \E \left( \left| \langle x^{(i)} \, | \, R^{(L)} \, u^{(n)} \rangle \right|^s \right) \, \sum_{\tilde{v}^{(l)} \in \overline{\partial ( \Lambda_{L+\iota}^C)}} \E \left( \big| \langle \tilde{v}^{(l)} \, | \, R \, y^{(j)} \rangle \big|^s \right).
    \end{split}
    \end{align}
    We prove \eqref{eq: second resampling} in appendix \ref{sec: second resampling argu}. Note that the cardinalities of $\partial ( \Lambda_L )$ and $\overline{ \partial ( \Lambda_{L+\iota}^C )}$ are both $\mathcal{O}(|L|)$. Using \eqref{propo: poly decay resolvent}, this gives for the first sum in \eqref{eq: second resampling} for any $a \geq 0$:
    \begin{align*}
        \sum_{u^{(n)} \in \partial (\Lambda_L)} \E \left( \left| \langle x^{(i)} \, | \, R^{(L)} \, u^{(n)} \rangle \right|^s \right) \leq \frac{c}{|L|^{a-1}}.
    \end{align*}
    Note that in order to apply Proposition \ref{propo: poly decay resolvent} we need that $\min(L_1, L_2) \geq 3$ and $x^{(i)} \nsim u^{(n)}$. Since $x^{(i)} \in \Hp^{0,0}$ and $u^{(n)} \in \partial ( \Lambda_L )$, the first assumption implies the second one. Additionally, Proposition \ref{propo: poly decay resolvent} requires $\|C-C_0\|_\infty \leq \frac{1}{|L|^{2ap + 4 + \frac{a}{s}}}$ for some $p > \frac{1}{1-s}$. Using that the cardinality of $\overline{\partial( \Lambda_{L+\iota}^C )}$ is $\mathcal{O}(|L|)$ yields
    \begin{align*}
        \sum_{\tilde{v}^{(l)} \in \overline{\partial ( \Lambda_{L+\iota}^C)}} \E \left( \big| \langle \tilde{v}^{(l)} \, | \, R \, y^{(j)} \rangle \big|^s \right) \leq c \, |L| \, \max_{\tilde{v}^{(l)} \in \overline{\partial ( \Lambda_{L+\iota}^C)}} \E \left( \big| \langle \tilde{v}^{(l)} \, | \, R \, y^{(j)} \rangle \big|^s \right).
    \end{align*}
    We plug the last two estimates into \eqref{eq: second resampling} and relabel $\tilde{v}^{(l)}$ to $v^{(k)}$ to obtain
    \begin{align*}
        \E \big( &\big| \langle x^{(i)} \, | \, R \, y^{(j)} \rangle \big|^s \big) \leq c \, |L|^{3-a} \, \max_{v^{(k)} \in \overline{\partial ( \Lambda_{L+\iota}^C)}} \E \left( \big| \langle v^{(k)} \, | \, R \, y^{(j)} \rangle \big|^s \right).
    \end{align*}
    for some constant $c > 0$ and $\|C-C_0\|_\infty \leq \frac{1}{|L|^{2ap + 4 + \frac{a}{s}}} =: \varphi$. Choosing $a$ and $L_0$ big enough such that $q = c \, L_0^{3-a} < 1$ finishes the proof.
\end{proof}

\subsection{Finite volume method} \label{sec: main proof}
Repeated application of Proposition \ref{propo: max on boundary} allows us to prove Theorem \ref{thm:main_thm}:
\begin{proof}
    We consider $L_0 \geq 5$ from Remark \ref{remark: relabeling propo} and set $L=(L_0,L_0)$. The left hand side of \eqref{eq:expo_decay} is bounded due to Theorem \ref{thm: frac mom bound}, which is why it suffices to consider only $x^{(i)}$ and $y^{(j)}$ such that $|x-y| > 4L_0 + 8$ . This condition guarantees that $y^{(j)} \in \Hp_{L+\iota}^C + [x^{(i)}]$, as can be verified in Figure \ref{fig: box and bigger box}. Applying Proposition \ref{propo: max on boundary} and Remark \ref{remark: relabeling propo} gives:
    \begin{align*}
        \E \left( \left| \langle x^{(i)} \, | \, R(C,z) \, y^{(j)} \rangle \right|^s \right) \leq q \, \max_{v^{(k)} \in \overline{\partial ( \Lambda_{L}^C)} + [x^{(i)}] } \E \left( \left| \langle v^{(k)} \, | \, R(C,z) \, y^{(j)} \rangle \right|^s \right).
    \end{align*}
    for some $q<1$ and all $C$ such that $\|C-C_0 \|_\infty \leq \varphi$. Applying Proposition \ref{propo: max on boundary} to the right hand side now yields:
    \begin{align*}
        \E \left( \left| \langle x^{(i)} \, | \, R(C,z) \, y^{(j)} \rangle \right|^s \right) \leq q^2 \, \max_{\tilde{v}^{(l)} \in \overline{\partial ( \Lambda_{L}^C)} + [v^{(k)}]} \E \left( \left| \langle \tilde{v}^{(l)} \, | \, R(C,z) \, y^{(j)} \rangle \right|^s \right).
    \end{align*}
    We iterate this procedure and apply Proposition \ref{propo: max on boundary} at least $n= \lfloor \frac{|x-y|}{4L_0+1} \rfloor$ times, until the point we obtained in the previous step is contained in a box of size $L_0$ centered at $y^{(j)}$. We obtain the constant $4L_0+1$ by identifying the point $v^{(k)} \in \partial(\Lambda_L^C)$ in Figure \ref{fig: box and bigger box} that is the furthest away from the center of the box, since this minimizes the steps $n$ we have to take to move from $x$ to $y$. After $n$ iterations of Proposition \ref{propo: max on boundary} we apply Theorem \ref{thm: frac mom bound} to obtain:
    \begin{align*}
        \E \left( \left| \langle x^{(i)} \, | \, R(C,z) \, y^{(j)} \rangle \right|^s \right) \leq c \, q^n.
    \end{align*}
    Setting $g = \frac{|\log(q)|}{4L_0+1}$ finishes the proof.
\end{proof}
\begin{figure} [h!]
    \centering
    \includegraphics[scale=0.35]{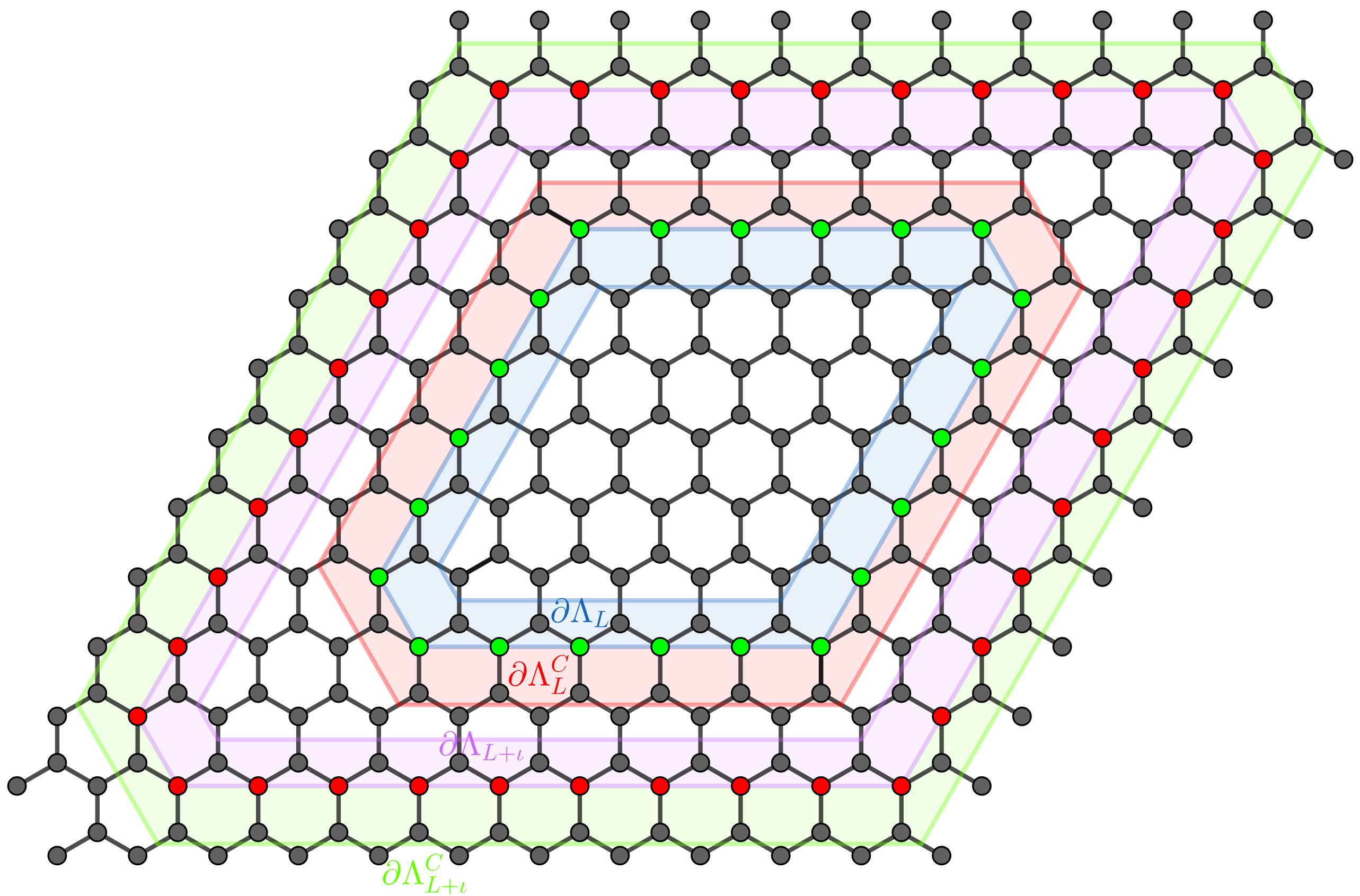}
    \caption{The boundaries of the boxes $\Lambda_L$ and $\Lambda_{L+\iota}$.}
    \label{fig: box and bigger box}
\end{figure}

\section{Transport using a topological index} \label{sec: topo properties}
\subsection{The scattering Quantum Walk}
In order to address transport properties of the QW using topological methods, we find it convenient to view the coined QW (CQW) as a scattering QW (SQW). Contrary to a CQW (see section \ref{sec: model and results}), where the Walker lives on the vertices of the graph, the Walker lives on its edges when considering a SQW (see eg. \cite{joye:2024}). For this, we identify the subspace of incoming edges to any lattice site with its coin space. The Walker moves from edge to edge. Notice that every CQW corresponds to a SQW, and vice versa.

Since the scattering approach requires directed edges, we split each edge into two directed edges and label them by their centers. Let $\mathcal{R}$ denote the set of edge centers. We connect two edge centers $u,v \in \mathcal{R}$ by a directed edge $(u,v)$ if and only if there exists a lattice site $x \in \Gamma$ such that $u$ is an incoming and $v$ an outgoing edge of $x$. The pair $(u,v)$ is called a link. We denote by $\mathcal{E}$ the set of links. The resulting graph $G = (\mathcal{R}, \mathcal{E})$ is called the scattering graph and illustrated in Figure \ref{fig: scattering graph}. 

\begin{wrapfigure}{r}{0.4\textwidth}
    \centering
    \includegraphics[scale=1.4]{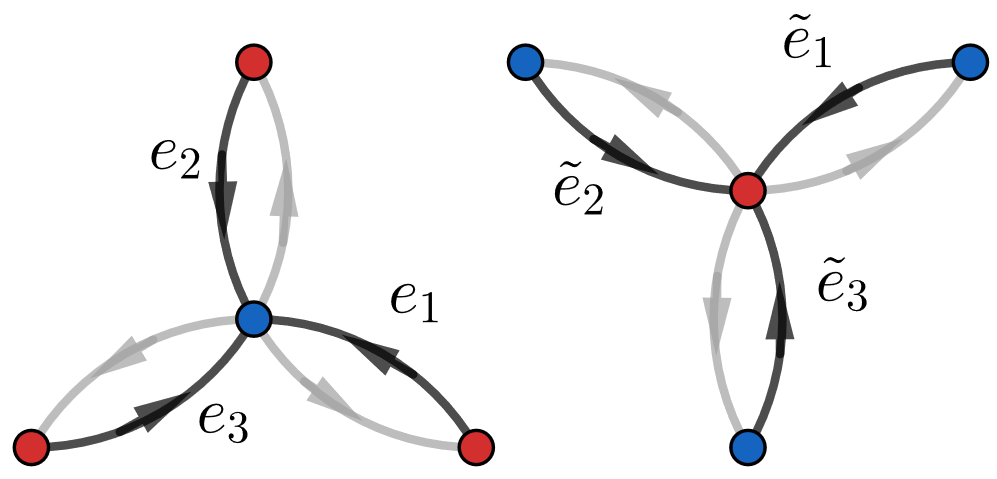}
    \caption{The relation between incoming subspace and coin space at $\Gamma_A$ and $\Gamma_B$ sites.}
    \label{fig: correspondence edges coin state}
\end{wrapfigure}
In the following we present the unitary map $\mathcal{V}$ mapping a coined QW to its scattering version, following the steps of \cite{joye:2024}: The operator describing the scattering QW is a unitary operator on $l^2(\mathcal{R})$. It moves the Walker along the links in $\mathcal{E}$. For any  $x \in \Gamma$, let $\mathcal{I}_x$ denote the three-dimensional subspace of $l^2(\mathcal{R})$ spanned by the incoming edges of $x$. Similarly, let $\mathcal{O}_x$ denote the subspace spanned by the outgoing edges. For any lattice site $x \in \Gamma$ we identify the subspace of incoming edges $\mathcal{I}_x$ with its coin space as shown in Figure \ref{fig: correspondence edges coin state}. We stress that we can also view the $\tilde{e}_j$ as an ONB of $\mathcal{O}_x$ for any $x \in \Gamma_A$, and similarly for $e_j$ and $x \in \Gamma_B$. Note that $l^2(\mathcal{R}) = \bigoplus_{x \in \Gamma} \mathcal{I}_x = \bigoplus_{x \in \Gamma} \mathcal{O}_x$ and $l^2(\Gamma) \otimes \C^3 = \bigoplus_{x \in \Gamma} \C^3$. Let $\mathcal{V}_x: \mathcal{I}_x \to \C^3$ denote the unitary operator mapping the orthonormal basis of $\mathcal{I}_x$ to the standard basis $(e_i)_i$ of the coin space $\C^3$, as detailed in Figure \ref{fig: correspondence edges coin state}. Therefore, the operator
\begin{align*}
    \mathcal{V} = \bigoplus_{x \in \Gamma} \mathcal{V}_x \, : l^2(\mathcal{R}) \to l^2(\Gamma) \otimes \C^3
\end{align*}
is unitary. The scattering version of the coined QW $U$ is then given by $\mathcal{V}^{-1} U \mathcal{V}$. From now on, let $\mathcal{U}$ denote the scattering version of the QW on the hexagonal lattice.
\begin{figure} [h!]
    \centering
    \includegraphics[scale=1.5]{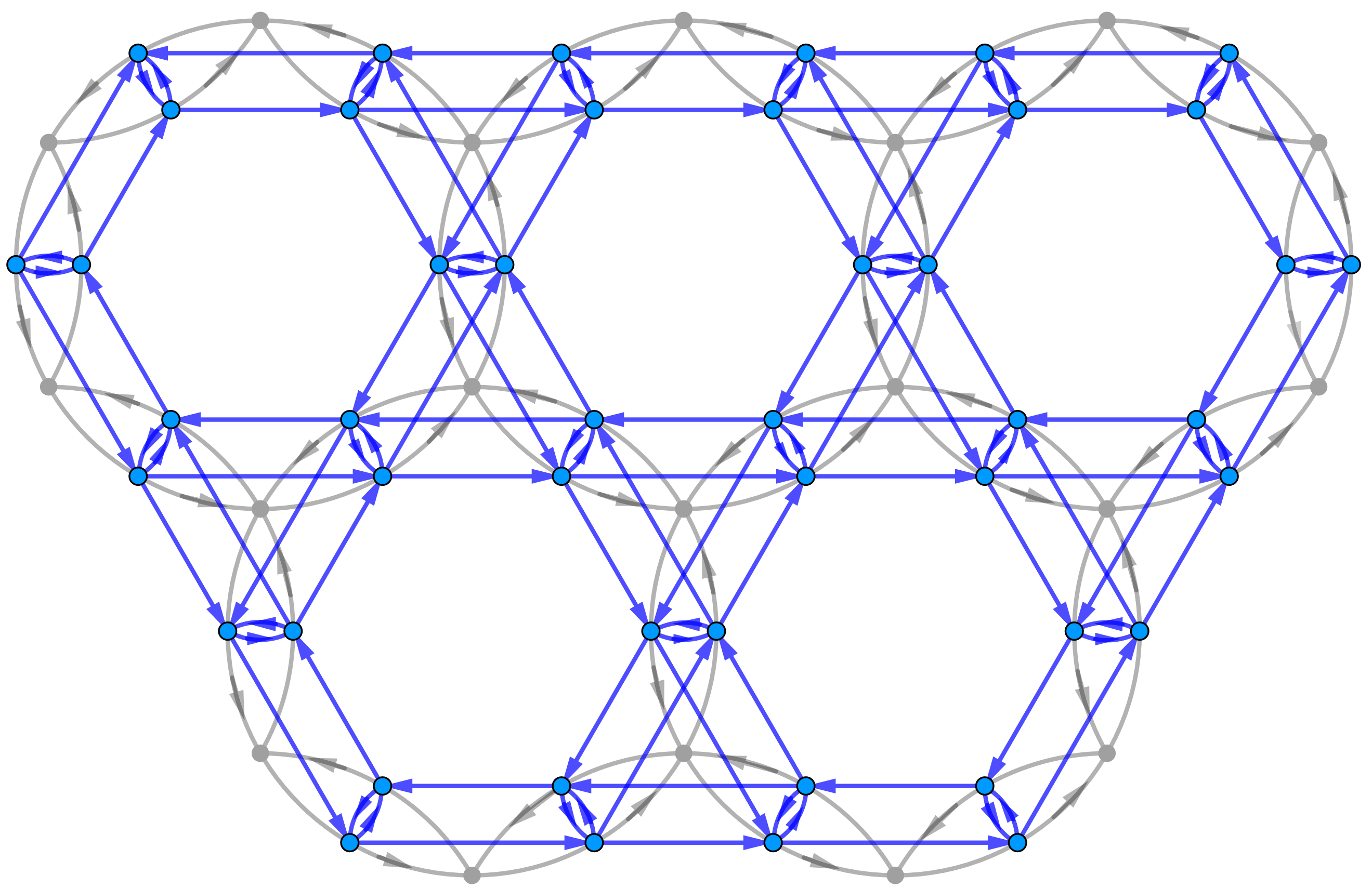}
    \caption{The scattering graph (blue) is obtained from the directed hexagonal lattice (grey). Note that the scattering graph is not planar.}
    \label{fig: scattering graph}
\end{figure}

\subsection{The index}
The index we consider is the index of a pair of projections, as defined in \cite{ASS:94}. Under the right conditions, it coincides with the trace of the flux operator of a projection on a suitable subset of edge centers in $\mathcal{R}$. The following theorem was proven in \cite{ABJ:2020}:
\begin{theorem} \label{thm: index properties}
    Let $\mathcal{U}$ be a unitary operator on a Hilbert space and $P$ be an orthogonal projection. Define the self-adjoint operator 
    \begin{align*}
        \Phi = \mathcal{U}^* P\mathcal{U} - P = \mathcal{U}^* [P,\mathcal{U}].
    \end{align*}
    If $1$ is an isolated eigenvalue of finite multiplicity of $\Phi^2$, we define the index
    \begin{align*}
        \text{ind}(\Phi) := \text{dim} \big( \text{ker} (\Phi - 1) \big) - \text{dim} \big( \text{ker} (\Phi + 1) \big).
    \end{align*}
    Furthermore, if $\text{ind}(\Phi) \neq 0$, then
    \begin{enumerate}
	   \item if $\Phi$ is compact, then the spectrum covers the whole unit circle: $\sigma(\mathcal{U}) = \Sp^1$
	   \item if $\Phi$ is trace class, then the absolutely continuous spectrum covers the unit circle: $\sigma_{ac}(\mathcal{U}) = \Sp^1$
    \end{enumerate}
\end{theorem}
For the proof we refer to \cite{ABJ:2020}. It was shown in \cite{ASS:94} that for two orthogonal projections $P$ and $P'$ we have:
\begin{align} \label{eq: index invariance compact perturbations}
    P-P' \text{ is compact} \Rightarrow \text{ind}(\Phi) = \text{ind}(\Phi').
\end{align}
Our goal is to define a projection onto a suitable "half-space" of $\mathcal{R}$ that yields a non-zero index, and establish conditions under which the flux operator $\Phi$ is compact or trace-class. From this we can deduce the existence of continuous spectrum, which is an indicator for transport. Let $Q_x$, $\widehat{Q}_x$ be the projections onto $\mathcal{I}_x , \mathcal{O}_x \subset l^2(\mathcal{R})$. We can show the following:
\begin{proposition} \label{propo: counting dimensions}
    For any subset $\mathcal{M} \subset \mathcal{R}$, let $P = \sum_{u \in \mathcal{M}} \ket{u} \bra{u}$ be the projection onto $\mathcal{M}$ in $l^2(\mathcal{R})$. Let $\Phi = \mathcal{U}^* P \mathcal{U} - P$, then it holds for all $x \in \Gamma$:
    \begin{enumerate}
        \item $[\Phi, Q_x] = 0$
        \item $\Phi Q_x = \mathcal{U}^* P \widehat{Q}_x \mathcal{U} - P Q_x$
        \item $\text{ind}(\Phi)$ is well-defined if and only if $\exists \, c < 1$ and $R > 0$ s.t. $\sup_{|x| > R} \| \Phi Q_x \| \leq c$. In this case, we have
        \begin{align} \label{eq: sum dim difference}
            \text{ind}(\Phi) = \sum_{|x| \leq R} \text{dim} \big( \text{Ran} (P \widehat{Q}_x) \big) - \text{dim} \big( \text{Ran} (P Q_x) \big).
        \end{align}
    \end{enumerate}
\end{proposition}
Proposition \ref{propo: counting dimensions} was already shown for different models in \cite{ABJ:2020} (Proposition 4.2) and \cite{AschMouneime:2019} (Proposition 2.2). Their proof can be applied to the presented SQW on the hexagonal lattice without any significant changes and is based on the identities $\mathcal{U} Q_x = \widehat{Q}_x \mathcal{U}$ and $[P, Q_x] = 0$.

Whenever \eqref{eq: sum dim difference} holds we can obtain the index by counting the amount of links that enter or leave $\mathcal{M}$ in some finite neighborhood of the origin. To easily count these links, we will define a path $\gamma$ that parametrizes the boundary of $\mathcal{M}$. Since the scattering graph $G=(\mathcal{R}, \mathcal{E})$ is not planar (see Figure \ref{fig: scattering graph}), we can not use its dual graph. Nevertheless, we can define $\gamma$ as some curve in $\R^2$ that connects the "faces" of $G$:

\begin{wrapfigure}{r}{0.35\textwidth}
    \centering
    \includegraphics[scale=1.3]{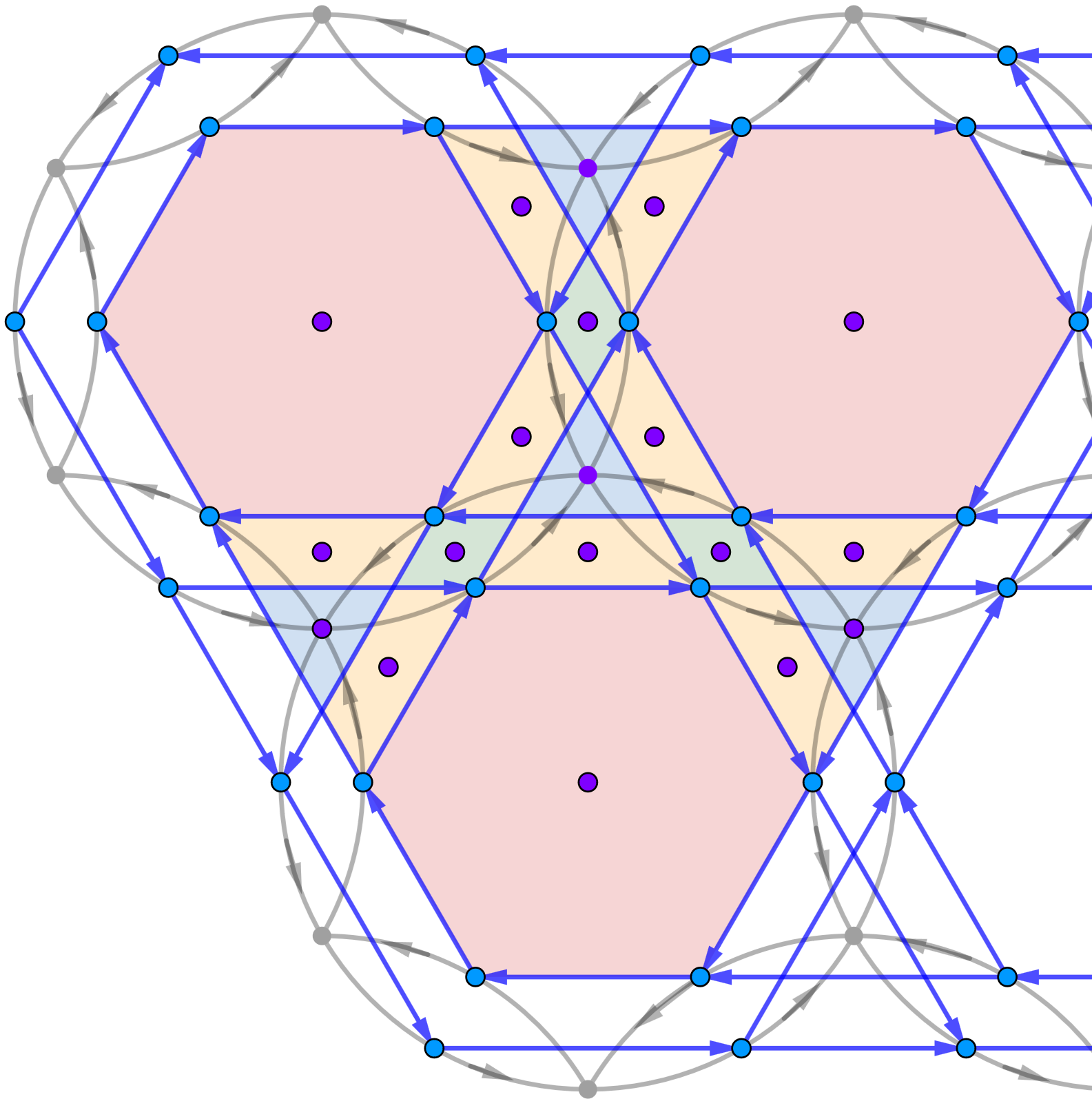}
    \caption{The "faces" of $G$ and their centers $\mathcal{R}^*$.}
    \label{fig: scattering graph faces}
\end{wrapfigure}
We omit the loops that arise from splitting the edges of the hexagonal lattice in two directions, since they do not affect the definition of $\mathcal{M}$, and identify the "faces" of $G$ as shown in Figure \ref{fig: scattering graph faces}. Each site of $\Gamma$ is the center of a triangle, which in turn borders on trapezoids. There is a parallelogram between two adjacent lattice sites, while we find a hexagon in each plaquette of the lattice. We stress that these are not actual faces in the sense of graph theory. However, it helps us clearly define the path $\gamma$ that parametrizes the boundary of $\mathcal{M}$.

Let $\mathcal{R}^*$ denote the set of centers of faces, then let $\gamma$ be a continuous path of piecewise straight segments connecting adjacent faces such that $\gamma(t) \in \mathcal{R}^*$ for all $t \in \Z$. A loop in $\gamma$ corresponds to a finite set of edges in $\mathcal{M}$, which leads to a finite-rank perturbation of $P$. Since the index is invariant under compact perturbations \eqref{eq: index invariance compact perturbations}, we can assume without loss of generality that $\gamma$ does not contain any loops. Note that crossing the corners of faces is not allowed, i.e. from a parallelogram we can only enter a trapezoid. Then $\gamma$ partitions the set $\mathcal{R}$ into the set $\mathcal{M}$ of edges that are on the left of $\gamma$ and the set $\mathcal{M}^C$ of edges on the right of $\gamma$, see Figure \ref{fig: possible path}. \\
\begin{figure} [h!]
    \centering
    \includegraphics[scale=1.2]{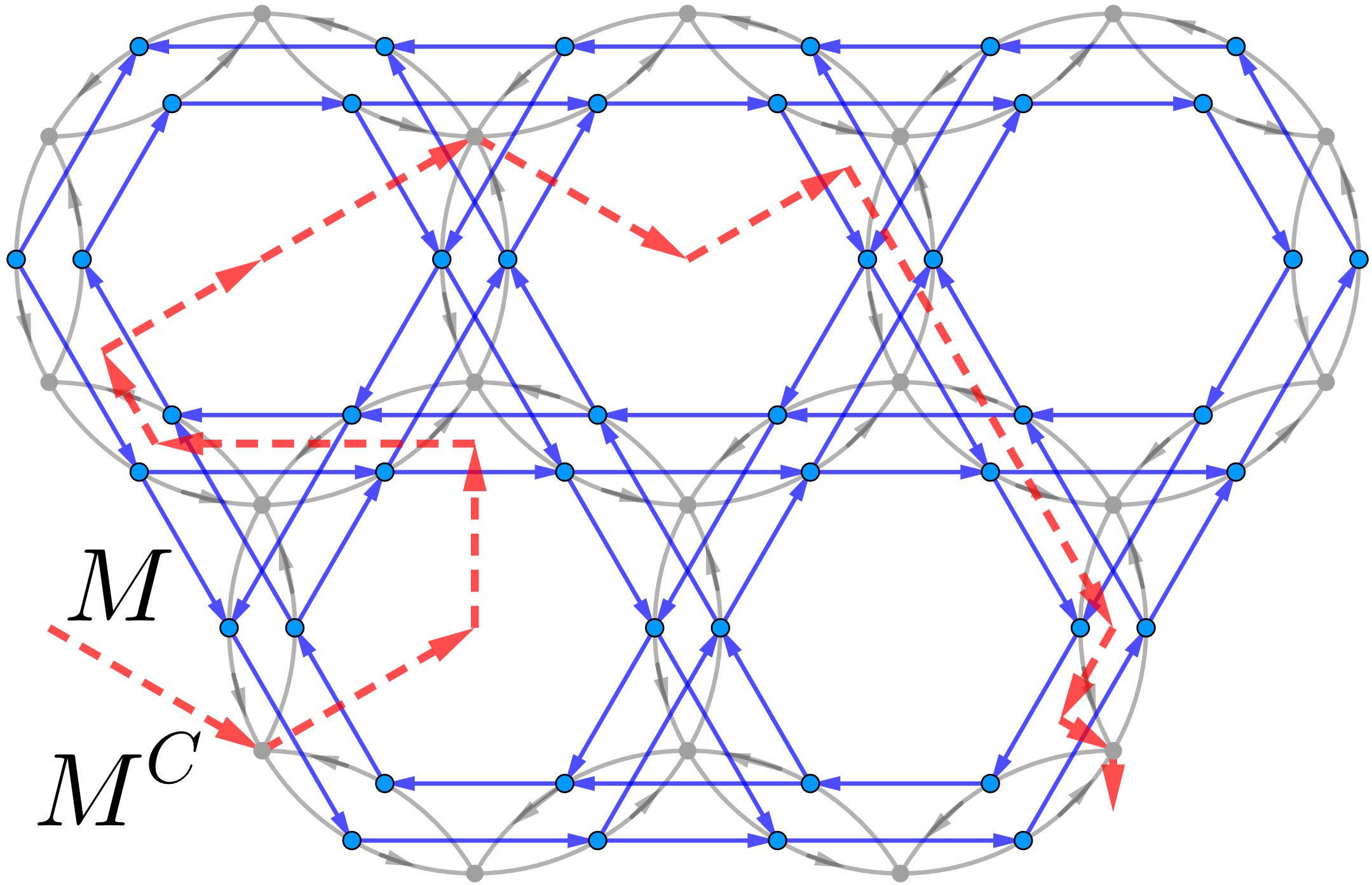}
    \caption{A path $\gamma$ partitioning $\mathcal{R}$ into $\mathcal{M}$ and $\mathcal{M}^C$.}
    \label{fig: possible path}
\end{figure}\\
Due to Proposition \ref{propo: counting dimensions} we consider all vertices $x \in \Gamma$ such that there exists a link between an incoming and an outgoing edge of $x$ that is bisected by $\gamma$:
\begin{align*}
    \Gamma_\gamma = \{ x \in \Gamma \, | \, P^\perp \mathcal{U} P Q_x \neq 0 \; \text{or} \; P \mathcal{U} P^\perp Q_x \neq 0 \}.
\end{align*}
Here $P^\perp = 1 - P$ denotes the projection onto $\mathcal{M}^C$. We note that $\gamma$ can be partitioned into a countable set of "steps" (see Figure \ref{fig: steps passing z}) in which $\gamma$ passes a vertex $x \in \Gamma_\gamma$. These steps must start and end in either a hexagon or a parallelogram. A path starting in a hexagon and finishing in a parallelogram is called an $hp$-step. Similarly, we define $hh$-, $pp$- and $ph$-steps. Since the direction of $\gamma$ corresponds to a sign change in \eqref{eq: sum dim difference}, a $ph$-step is equivalent to an $hp$-step. Note that there may be multiple paths from one face to another while passing a site $x \in \Gamma$. However, since they can be continuously deformed into another without crossing any of the adjacent edge centers, the resulting split of the adjacent edges of $x$ into $\mathcal{M}$ and $\mathcal{M}^C$ is the same. Thus, up to rotation and reflection, any path $\gamma$ can without loss of generality be partitioned into a series of 4 steps illustrated in Figure \ref{fig: steps passing z}.
\begin{figure}[h!]
    \centering
    \begin{minipage}{.25\textwidth}
        \centering
        \includegraphics[width=.9\linewidth]{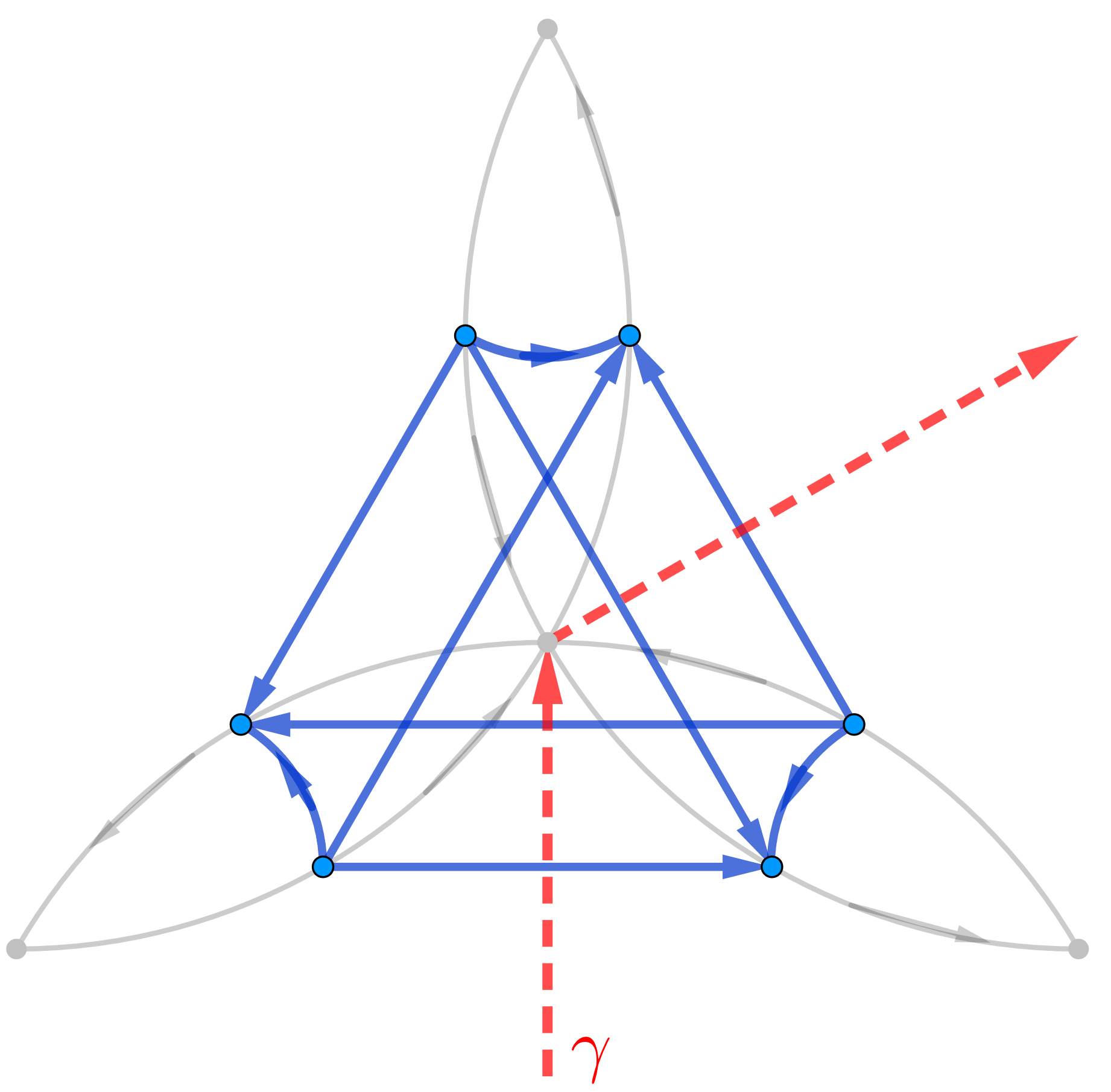}
    \end{minipage}%
    \begin{minipage}{.25\textwidth}
        \centering
        \includegraphics[width=.9\linewidth]{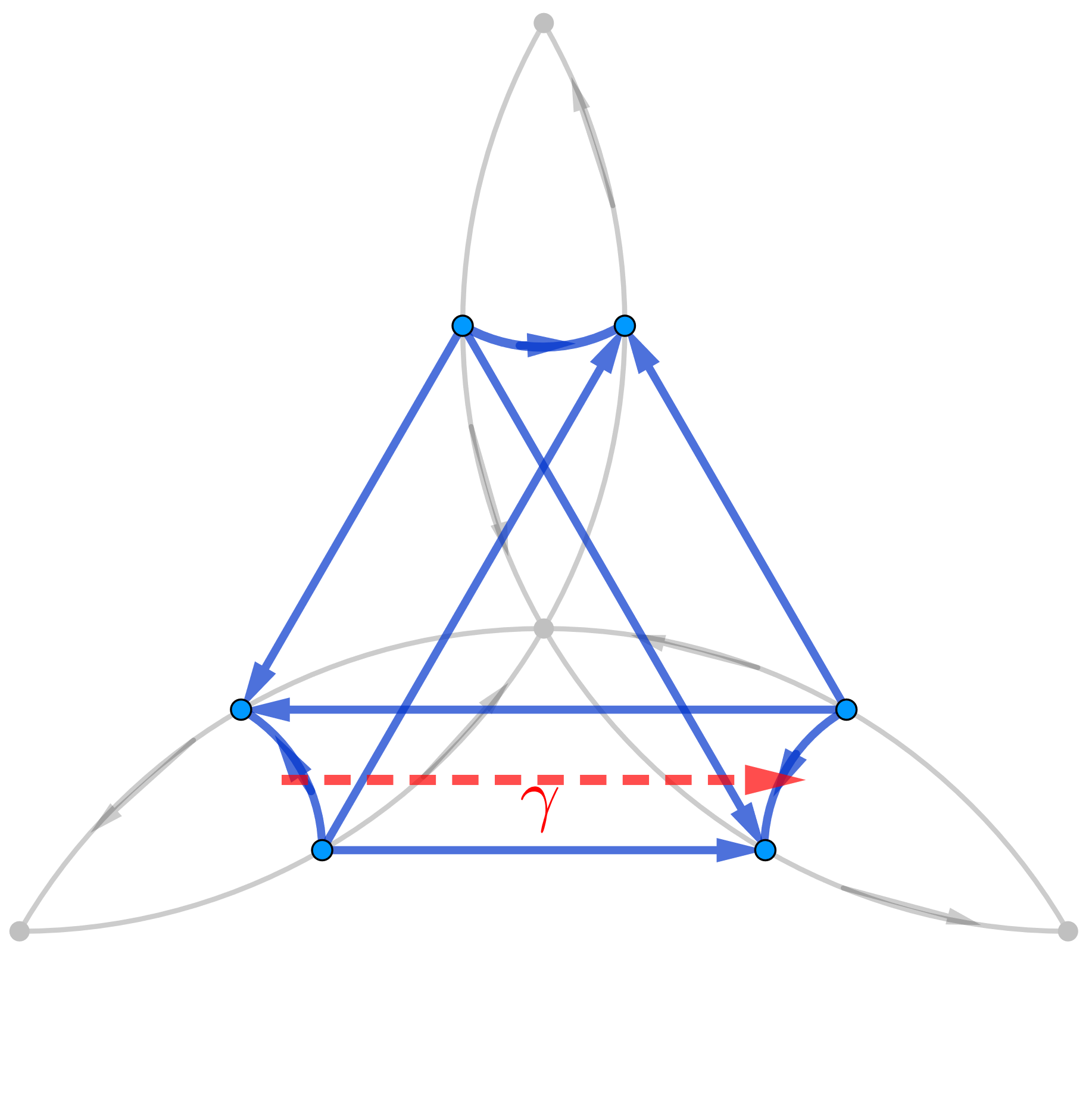}
    \end{minipage}%
    \begin{minipage}{.25\textwidth}
        \centering
        \includegraphics[width=.9\linewidth]{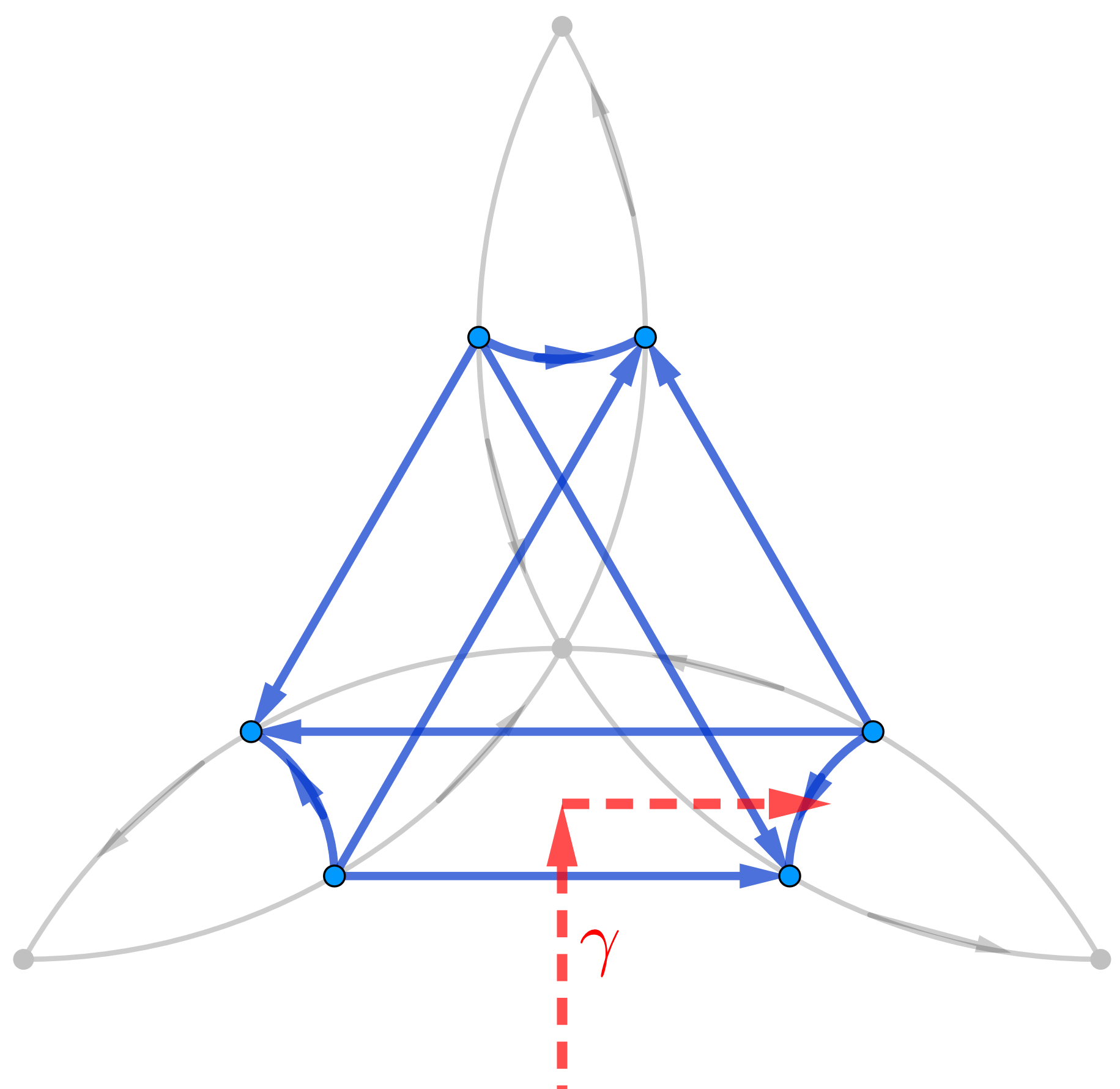}
    \end{minipage}%
    \begin{minipage}{.25\textwidth}
        \centering
        \includegraphics[width=.9\linewidth]{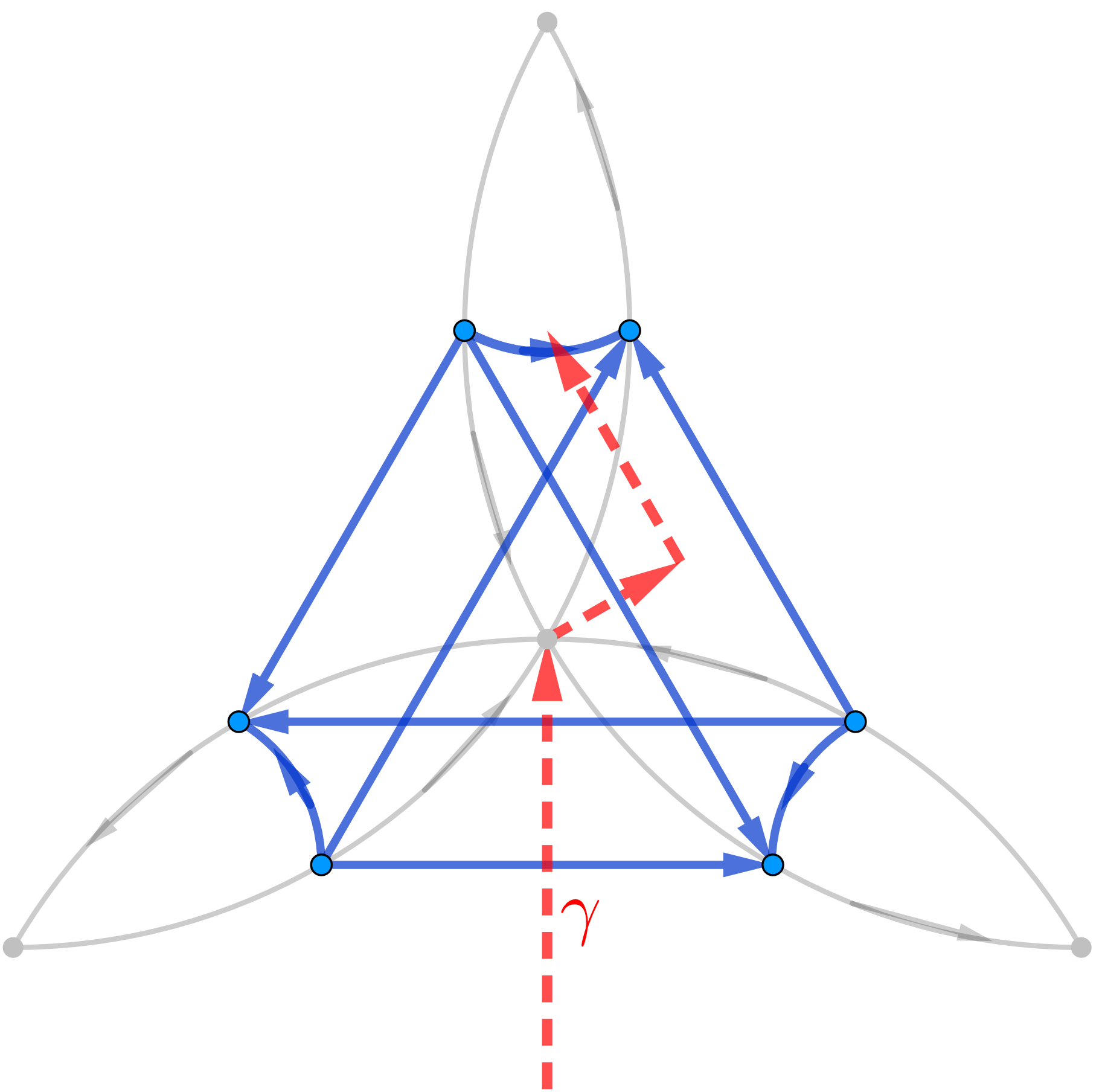}
    \end{minipage}
    \caption{All possible steps passing a lattice site $x \in \Gamma$, from left to right: $hh$, $pp$, and two $hp$ steps.}
    \label{fig: steps passing z}
\end{figure}

We can count the dimensions of the respective incoming and outgoing subspaces in all 3 cases: If $\gamma$ performs an $hh$- or $pp$-step, then $\text{dim} \big( \text{Ran} (P \widehat{Q}_x) \big)$ and $\text{dim} \big( \text{Ran} (P Q_x) \big)$ agree with values $1$ or $2$, depending on the direction of $\gamma$. In any case, the resulting difference is zero. In the case of an $hp$-step, we have $\text{dim} \big( \text{Ran} (P \widehat{Q}_x) \big) - \text{dim} \big( \text{Ran} (P Q_x) \big) = \pm 1$, depending on the direction of $\gamma$. We note that in both cases this corresponds to the difference of bisected links entering and leaving $\mathcal{M}$ divided by $3$. Thus, only an $hp$-step contributes to a non-zero index.

\begin{definition} \label{def: relevant path}
    A path $\gamma$ connecting adjacent faces of the scattering graph (see Figure \ref{fig: scattering graph faces}) with piecewise straight segments is called relevant if it is continuous, without loops and $\exists \, N \geq 0$ such that $\gamma$ performs only $hh$-steps for all $t \leq -N$ and only $pp$-steps for all $t \geq N$.
\end{definition}

We can prove the following:

\begin{lemma}
    Let $\gamma$ be a relevant path and $\mathcal{M}$ be the subset of $\mathcal{R}$ on the left of $\gamma$. Define $\Phi = \mathcal{U}^* P \mathcal{U} - P$, where $P$ is the projection onto $\mathcal{M}$ and suppose $\exists \, R > 0$ such that $\sup_{x \in \Gamma_\gamma, |x| > R} \|\Phi Q_x \| < 1$. Then $\text{ind}(\Phi) = \pm 1$. 
\end{lemma}

\begin{proof}
    Suppose that $N=1$, i.e. $\gamma$ only performs $hh$-steps on $(-\infty, 1]$ and only $pp$-steps on $[1, \infty)$. The index of $\Phi$ is well-defined by Proposition \ref{propo: counting dimensions} and can be obtained by counting the dimensions of the respective incoming and outgoing subspaces for all $x \in \Gamma_\gamma$ in a finite neighborhood of the origin. Since $\gamma$ performs only a single $hp$-step, we have $\text{ind}(\Phi)= \pm 1$. 

    If $N > 1$, let $N_- \leq -N$ denote the biggest $t \leq -N$ such that $\gamma(t)$ is the center of a hexagon, and similarly let $N_+$ be the smallest $t \geq N$ such that $\gamma(t)$ is the center of a parallelogram. We connect $\gamma(N_-)$ and $\gamma(N_+)$ by a path $\tilde{\gamma}$ that performs only $hh$-steps until $\tilde{\gamma}(N_+ - 1)$, which is in the center of a hexagon. The obtained path $\Tilde{\gamma}$ thus performs only $hh$-steps for $t \leq N_+ -1$ and only $pp$-steps for $t \geq N_+$. Using the translation $t \mapsto t - N_+$, we obtain a relevant path with $N = 1$. Let $\tilde{P}$ be the projection on the subset of $\mathcal{R}$ on the left of $\tilde{\gamma}$. Since the difference $P-\tilde{P}$ is of finite rank, we have $\text{ind}(\Phi) = \text{ind}(\tilde{\Phi}) = \pm 1$ due to \eqref{eq: index invariance compact perturbations}.  
\end{proof}

\subsection{Relation to the coefficients of the coin matrix}
A direct calculation shows
\begin{align} \label{eq: Phi squared}
    \Phi^2 = P \mathcal{U}^* P^\perp \mathcal{U} P + P^\perp \mathcal{U}^* P \mathcal{U} P^\perp.
\end{align}
Using the correspondence between the incoming subspace and the coin space as illustrated in Figure \ref{fig: correspondence edges coin state}, we can assign the correct weights to the links in the scattering graph, see Figure \ref{fig: PhiQ H1}.
\begin{figure} [h!]
    \centering
    \includegraphics[scale=4]{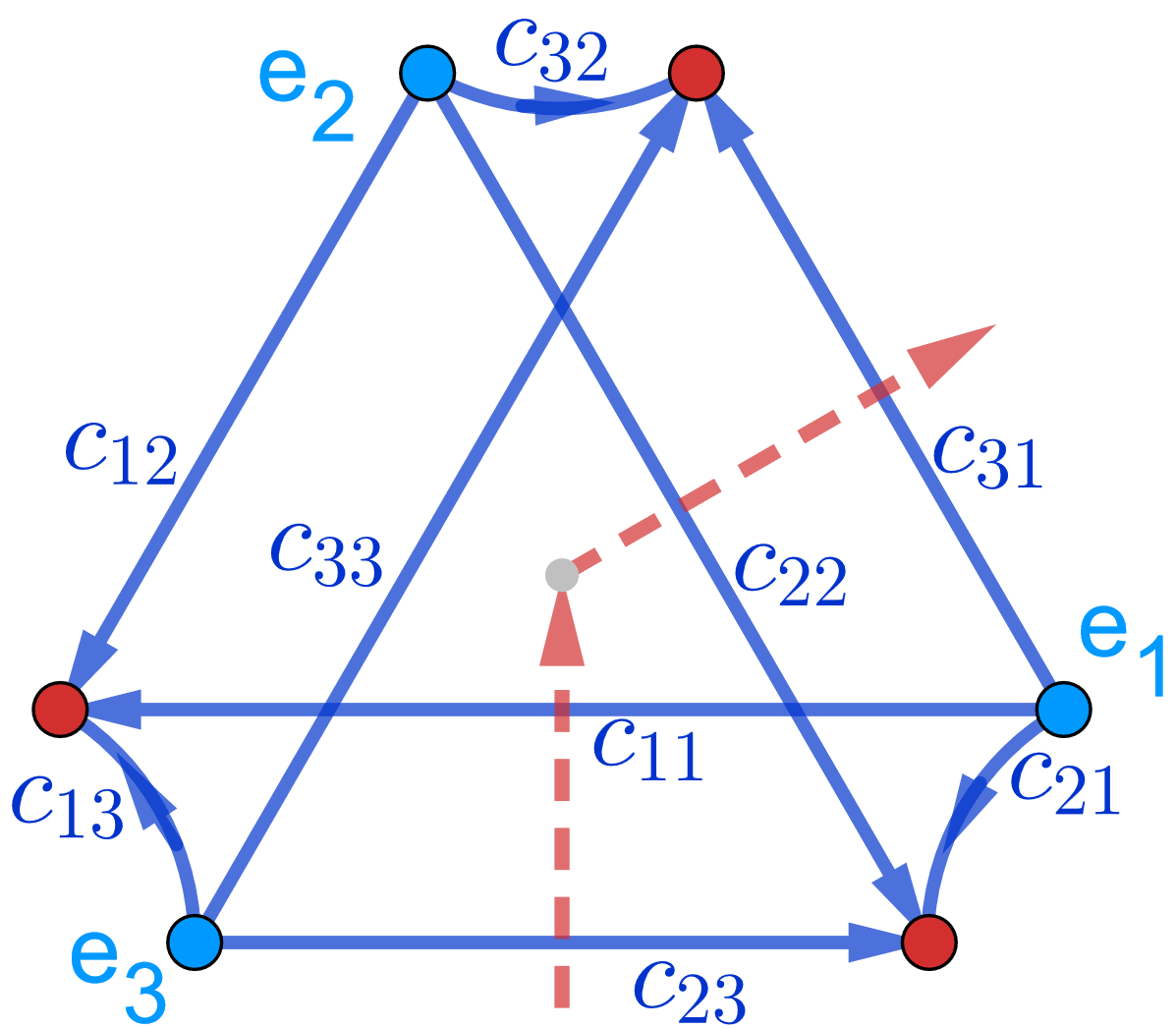}
    \caption{The weighted scattering graph, together with a path $\gamma$ passing $x \in \Gamma_A$ in a $hh$-step. The weights $c_{ij}$ are the coefficients of the coin matrix $C_x$ at site $x$.}
    \label{fig: PhiQ H1}
\end{figure}

Applying \eqref{eq: Phi squared}, we can express $\Phi^2 Q_x$ as a $3 \times 3$ matrix. If the orthonormal basis of $\mathcal{I}_x$ is labeled and $\gamma$ passes $x \in \Gamma_A$ as illustrated in Figure \ref{fig: PhiQ H1}, we can compute the matrix elements of $\Phi^2 Q_x$:
\begin{align*}
    \bra{e_3} \Phi^2 Q_x e_2 \rangle = \bra{e_1} P \mathcal{U}^* P^\perp \mathcal{U} P e_2 \rangle = \bra{e_3} \mathcal{U}^* c_{22} \tilde{e}_2 \rangle = \bra{c_{23} \tilde{e}_2 } c_{22} \tilde{e}_2 \rangle = \overline{c_{23}} c_{22}.
\end{align*}
Here, we let $\tilde{e}_j$ denote an element of the ONB of the outgoing subspace $\mathcal{O}_x$, see Figure \ref{fig: correspondence edges coin state}. Similar computations yield:
\begin{align*}
    \Phi^2 Q_x = \begin{pmatrix}
        |c_{11}|^2 + |c_{31}|^2 & 0 & 0 \\
        0 & |c_{22}|^2 & c_{23} \overline{c_{22}} \\
        0 & c_{22} \overline{c_{23}} & |c_{23}|^2
    \end{pmatrix}.
\end{align*}
The eigenvalues of this matrix are $|c_{11}|^2 + |c_{31}|^2$, $|c_{22}|^2 + |c_{23}|^2$ and $0$. Since the coin matrix is unitary, the first two eigenvalues both simplify to $1-|c_{21}|^2$. We have
\begin{align*}
    \| \Phi Q_x \|^2 = \sup_{\substack{ \varphi \in \text{Ran}(Q_x) \\ \| \varphi\| = 1}} \bra{\varphi} \Phi^* \Phi \, \varphi \rangle = \max \big( | \lambda | \; \text{s.t. } \, \lambda \in \sigma(\Phi^2 \, Q_x) \big).
\end{align*}
Thus, we conclude that $\| \Phi Q_x \| = \sqrt{1- |c_{21}|^2}$ for the case illustrated in Figure \ref{fig: PhiQ H1}. We can argue similarly for the other possible $hh$-steps: Ordering the steps as shown in Figure \ref{fig: hh steps passing z}, we obtain for $\| \Phi Q_x \|$ the values $\sqrt{1- |c_{13}|^2}$, $\sqrt{1- |c_{32}|^2}$, $\sqrt{1- |c_{23}|^2}$, $\sqrt{1- |c_{31}|^2}$ and $\sqrt{1- |c_{12}|^2}$.\\
\begin{figure}[h!]
    \centering
    \begin{minipage}{.2\textwidth}
        \centering
        \includegraphics[width=.9\linewidth]{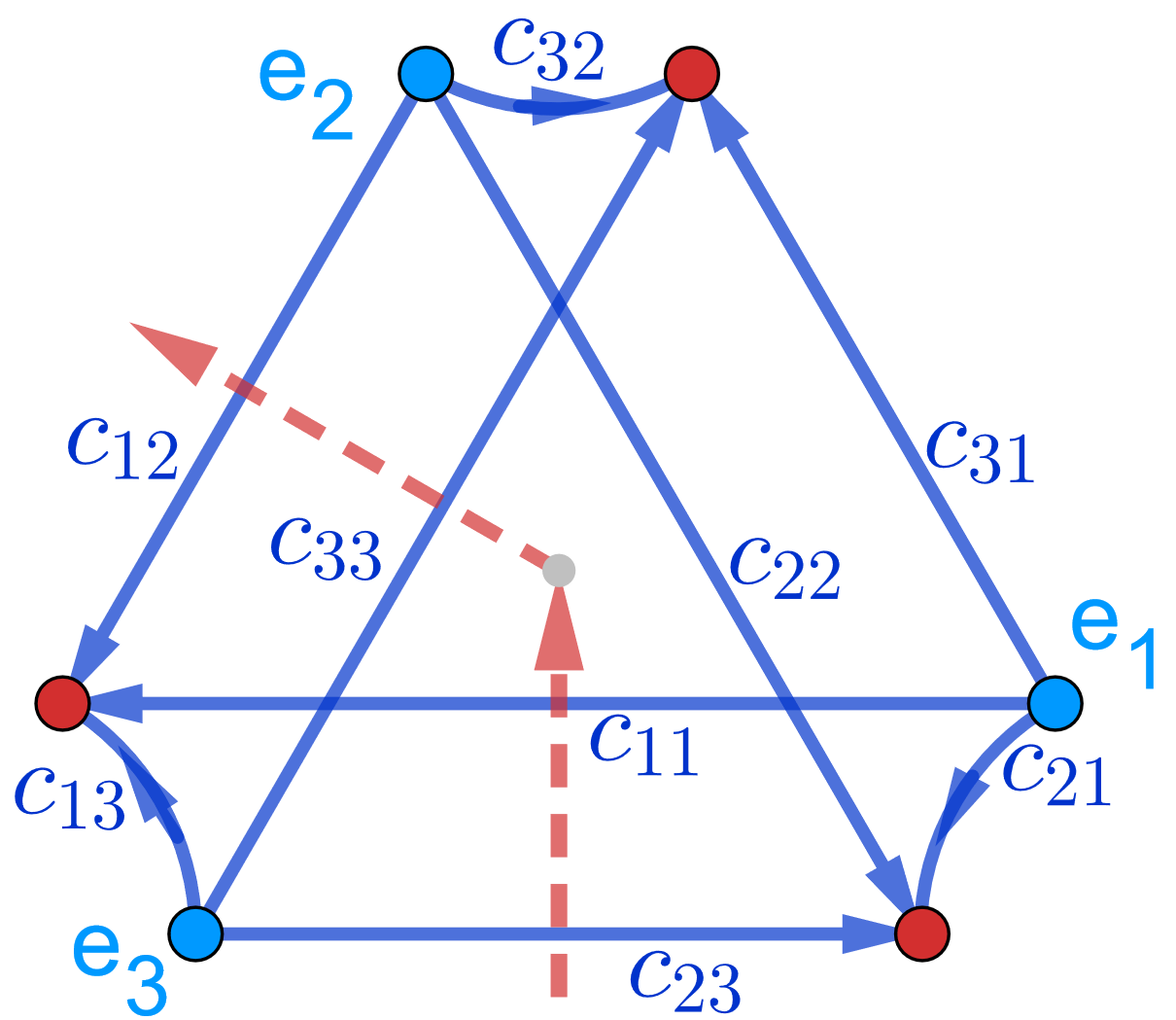}
    \end{minipage}%
    \begin{minipage}{.2\textwidth}
        \centering
        \includegraphics[width=.9\linewidth]{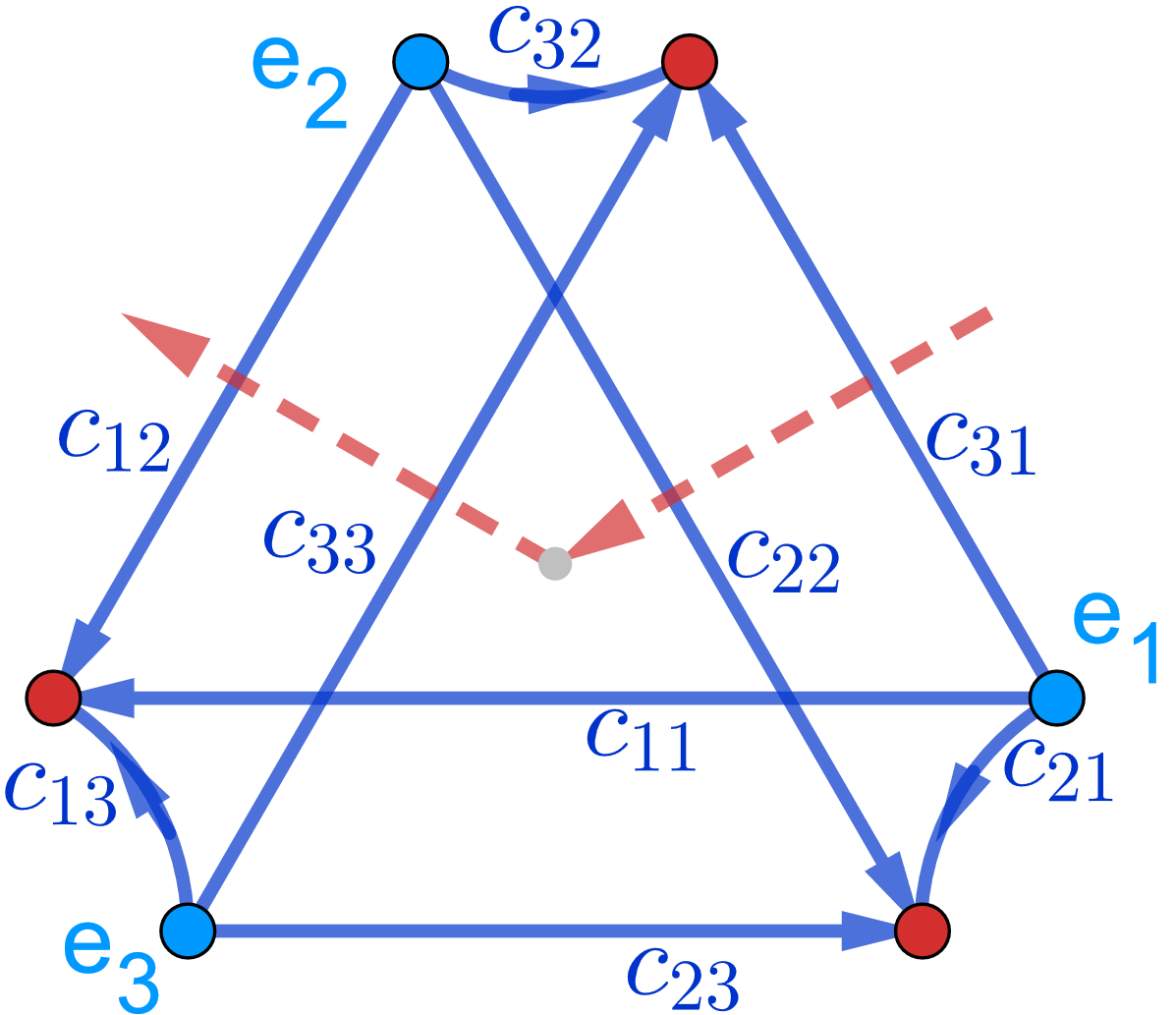}
    \end{minipage}%
    \begin{minipage}{.2\textwidth}
        \centering
        \includegraphics[width=.9\linewidth]{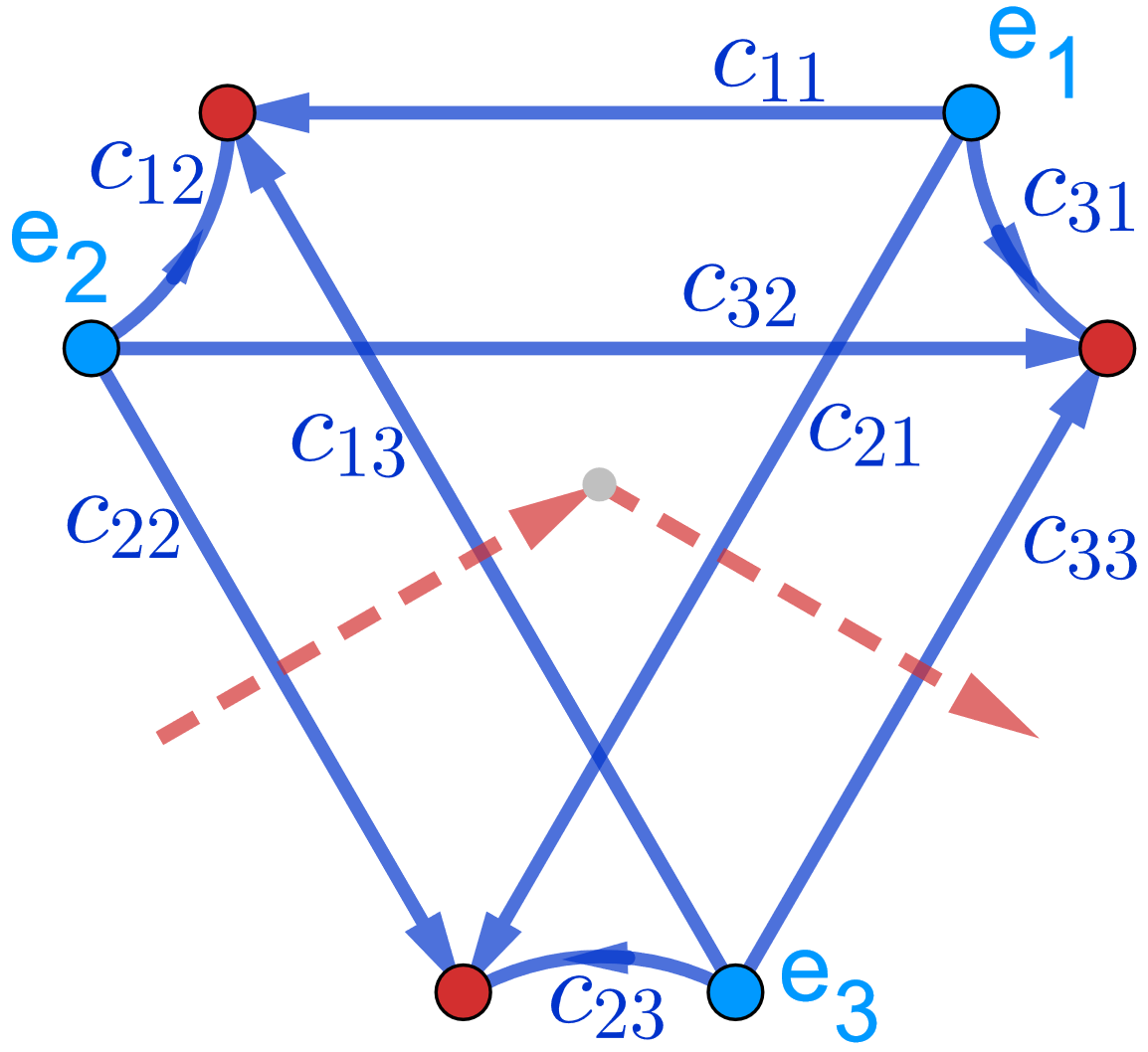}
    \end{minipage}%
    \begin{minipage}{.2\textwidth}
        \centering
        \includegraphics[width=.9\linewidth]{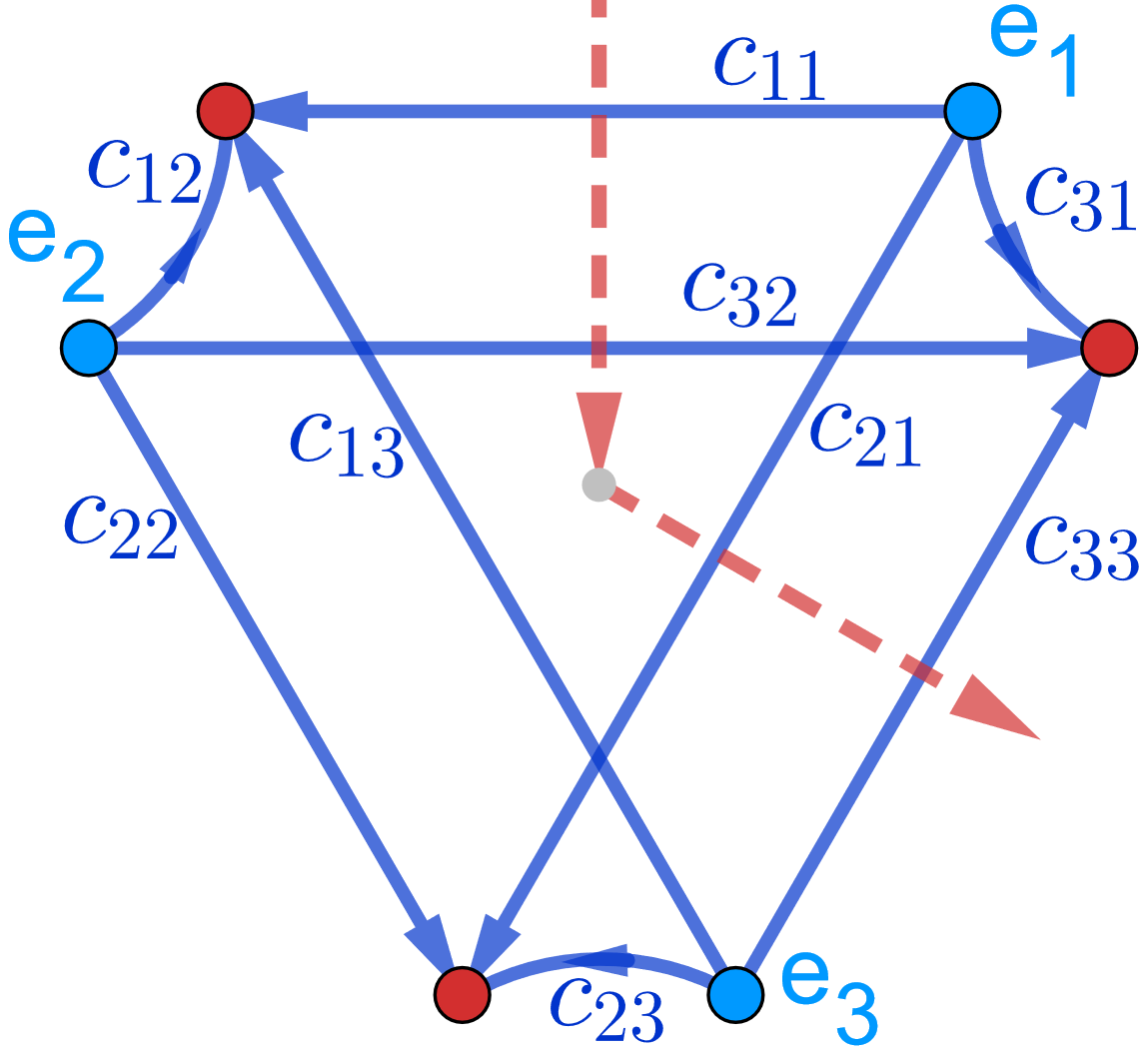}
    \end{minipage}%
    \begin{minipage}{.2\textwidth}
        \centering
        \includegraphics[width=.9\linewidth]{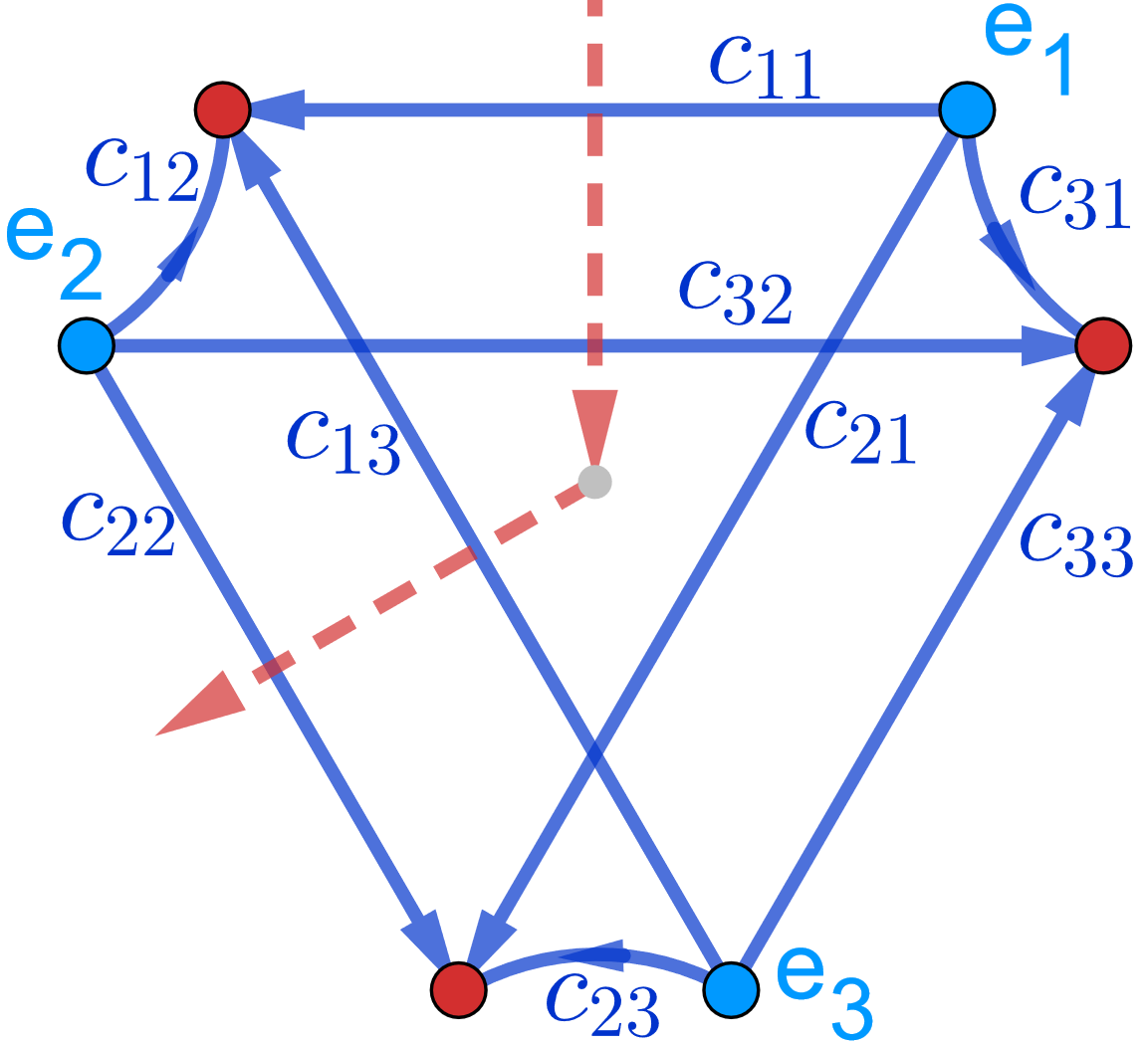}
    \end{minipage}
    \caption{Possible $hh$-steps passing a vertex $x \in \Gamma$.}
    \label{fig: hh steps passing z}
\end{figure}\\
By labeling all possible $pp$-steps from $\gamma_1$ to $\gamma_6$ as shown in Figure \ref{fig: pp steps passing z}, we obtain for the norm of $\Phi Q_x$ the values $\sqrt{1 - |c_{23}|^2}$, $\sqrt{1 - |c_{31}|^2}$, $\sqrt{1 - |c_{21}|^2}$, $\sqrt{1 - |c_{22}|^2}$, $\sqrt{1 - |c_{33}|^2}$ and $\sqrt{1 - |c_{11}|^2}$.

\begin{figure}[h!]
    \centering
    \begin{minipage}{.3\textwidth}
        \centering
        \includegraphics[width=.9\linewidth]{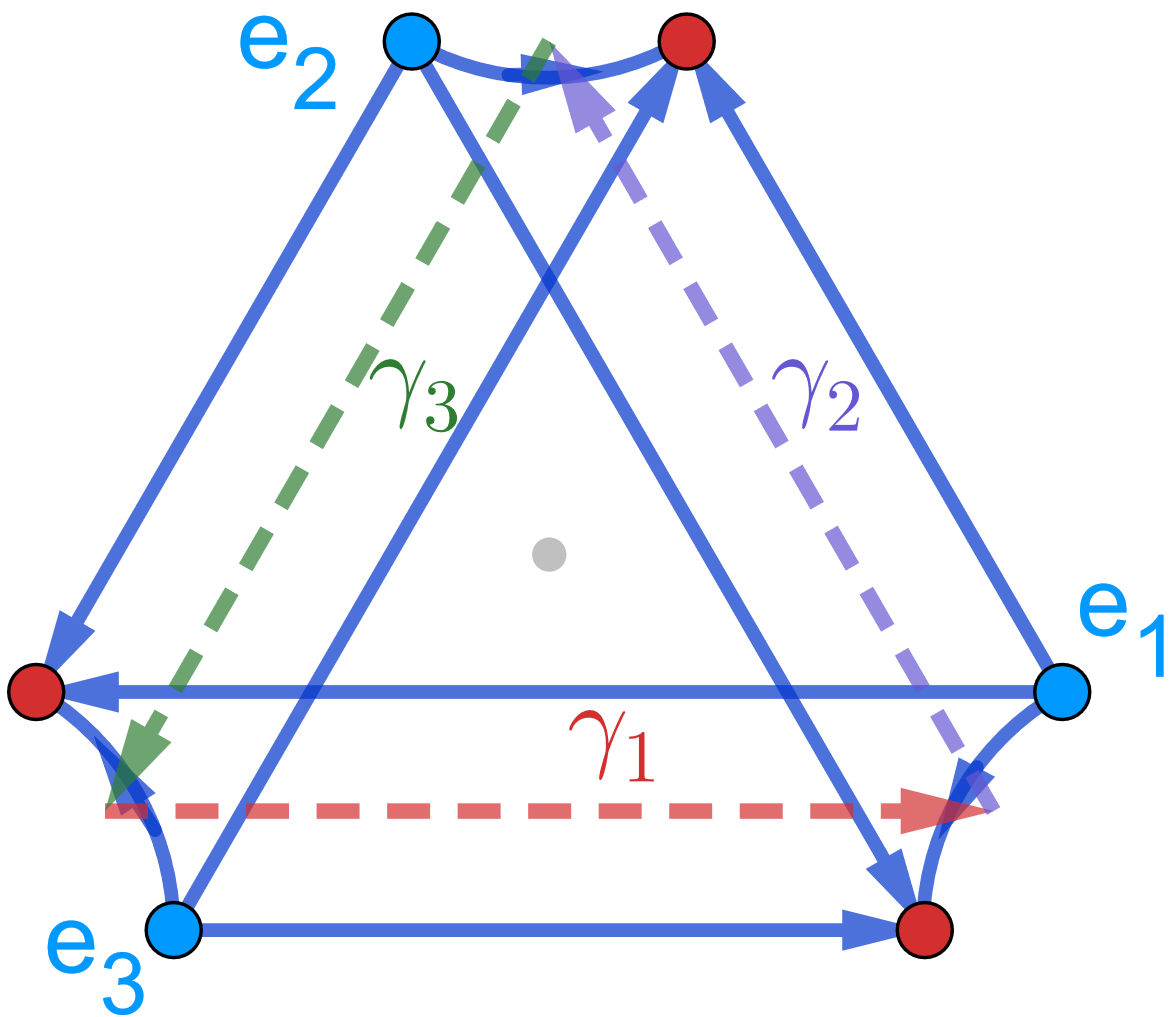}
    \end{minipage}%
    \begin{minipage}{.3\textwidth}
        \centering
        \includegraphics[width=.9\linewidth]{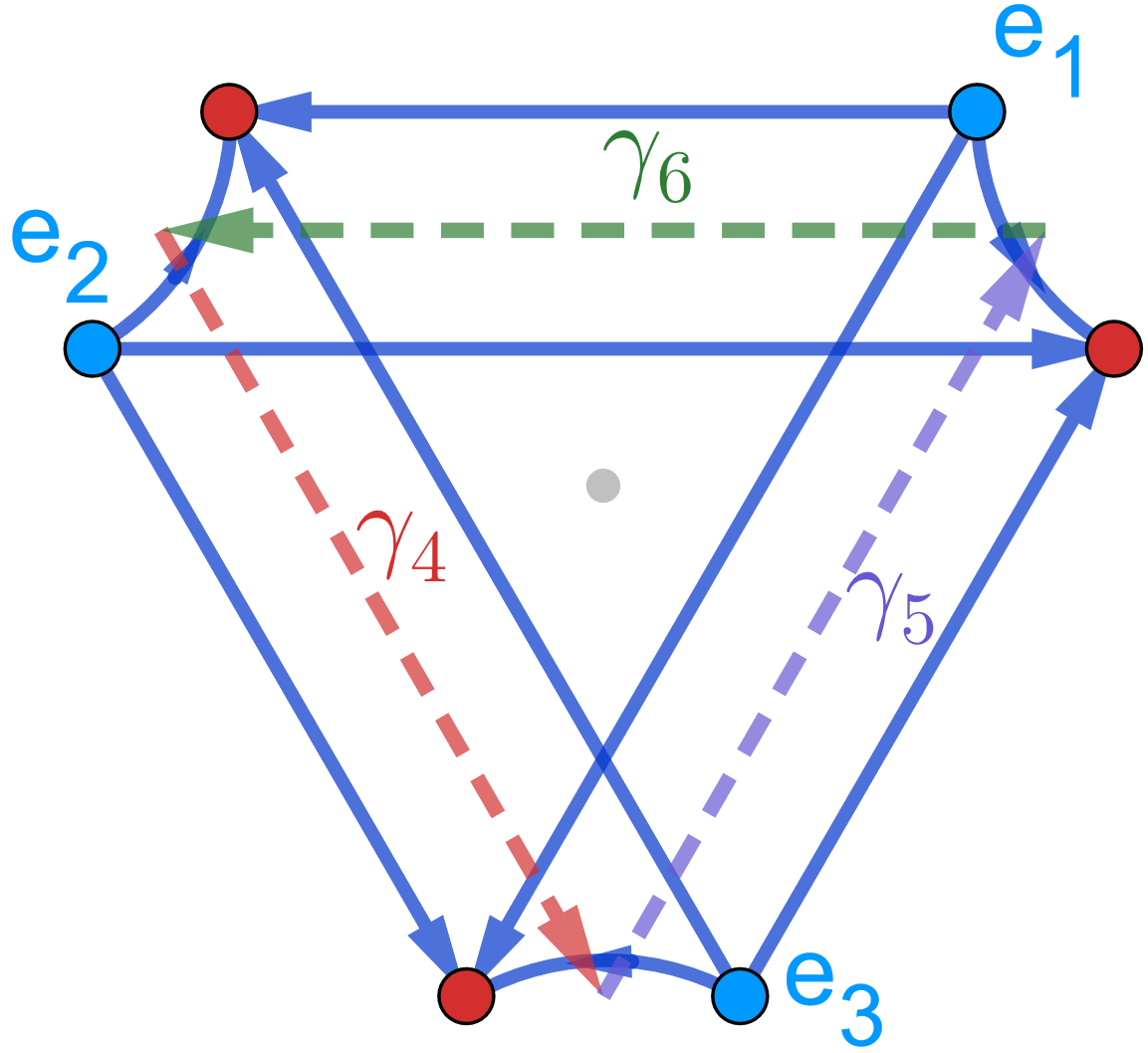}
    \end{minipage}
    \caption{All possible $pp$-steps, labeled from $\gamma_1$ to $\gamma_6$.}
    \label{fig: pp steps passing z}
\end{figure}
For any relevant path $\gamma$ we let $c^\gamma(x)$ denote the coefficient $(C_x)_{i,j}$ of the coin matrix at $x \in \Gamma_\gamma$ that corresponds to the $hh$ or $pp$ step that $\gamma$ performs at $x$. We stress that the index pair $(i,j)$ may change from lattice site to lattice site, depending on the step $\gamma$ performs. Using Proposition \ref{propo: counting dimensions}, we conclude that a relevant path $\gamma$ leads to a well-defined index if and only if there exists some $\varepsilon > 0$ such that for all $x \in \Gamma_\gamma$ sufficiently large we have $|c^\gamma(x)| > \varepsilon$.

The self-adjoint operator $\Phi$ is compact if and only if there exists an orthogonal set $(f_n)_{n \in \N}$ such that $\Phi = \sum_{n=1}^\infty \lambda_n \ket{f_n} \bra{f_n}$ with eigenvalues $\lambda_n \to 0$. Note that $\Phi = \bigoplus_{x \in \Gamma} \Phi \, Q_x$ where $\Phi\,  Q_x = 0$ if $x \notin \Gamma_\gamma$. The spectral decomposition of $\Phi$ thus reduces to the spectral decomposition of the blocks $\Phi \, Q_x$ and $\Phi$ is compact if and only if $|c^\gamma(x)| \to 1$ for $x \in \Gamma_\gamma$, $|x| \to \infty$. Similarly, we obtain for the trace of $|\Phi|$:
\begin{align} \label{eq: Phi trace class}
    \tr(|\Phi|) = \tr \Big( \bigoplus_{x \in \Gamma} |\Phi| \, Q_x \Big) = \sum_{x \in \Gamma_\gamma} \tr(|\Phi| \, Q_x) = \sum_{x \in \Gamma_\gamma} \sqrt{1 - |c^\gamma(x)|^2}.
\end{align}
The operator $\Phi$ is thus trace-class if and only if the sum \eqref{eq: Phi trace class} converges. To summarize, while compactness of $\Phi$ requires $|c^\gamma(x)| \to 1$ for $|x| \to \infty$ along $\gamma$, $\Phi$ is trace-class if this convergence is fast enough. Using Theorem \ref{thm: index properties}, we have proved the following:
\begin{theorem} \label{thm: index theorem}
    Let $\mathcal{M} \subset \mathcal{R}$ be a subset of edges, $P$ be the projection onto $\mathcal{M}$ in $l^2(\mathcal{R})$ and $\Phi = \mathcal{U}^* P \mathcal{U} - P$. Suppose that a path $\gamma$ describing the boundary of $\mathcal{M}$ is relevant (see Definition \ref{def: relevant path}). For any $x \in \Gamma_\gamma$, let $c^\gamma(x)$ denote the coefficient of the coin matrix $C_x$ that corresponds to the step $\gamma$ performs at $x$, i.e. $\| \Phi \, Q_x \| = \sqrt{1 - |c^\gamma(x)|^2}$. Then
    \begin{align*}
        \text{ind}(\Phi) \text{ is well-defined} \iff \exists \, R>0 \text{ s.t.} \inf_{x \in \Gamma_\gamma, |x| > R} |c^\gamma(x)| > 0.
    \end{align*}
    In this case, we have $\text{ind}(\Phi) = \pm 1$ and
    \begin{enumerate}
        \item $\Phi$ is compact $\iff$ $|c^\gamma(x)| \to 1$ for $|x| \to \infty$. In this case, we have $\sigma(\mathcal{U}) = \Sp^1$.
        \item $\Phi$ is trace-class $\iff$ $\sum_{x \in \Gamma_\gamma} \sqrt{1 - |c^\gamma(x)|^2} < \infty$. Then, we have $\sigma_{ac}(\mathcal{U}) = \Sp^1$.
    \end{enumerate}
\end{theorem}

\begin{remark}
    We note that the criteria above depend only on the absolute value of the coefficients of $C_x$. Thus, the same results hold when introducing random phases to the model. 
\end{remark}

\begin{remark}
    All criteria in Theorem \ref{thm: index theorem} concern only the coin matrices on sites $x \in \Gamma_\gamma$, i.e. sites close to the path $\gamma$. The results are thus independent of the coin matrices on all other sites.
\end{remark}

\begin{example}
    Let $C_x = C_0$ (see section \ref{sec: model and results}) for all $x \in \Gamma$. We know that this choice of coin matrix induces a fully localized behavior with spectrum $\sigma(\mathcal{U}) = \{1, e^{i \frac{2}{3}\pi}, e^{i \frac{4}{3}\pi} \}$. However, this is no contradiction to Theorem \ref{thm: index theorem}, since at least one non-zero diagonal entry is required to define a relevant path $\gamma$ that leads to a non-zero index. Thus, Theorem \ref{thm: index theorem} does not apply.
\end{example}

\begin{example}
    Let $C_x = I$ for all $x \in \Gamma$. In this case, we have $\sigma_{ac}(\mathcal{U}) = \Sp^1$. However, in order to define a relevant path $\gamma$ that leads to a non-zero index, we require at least one non-zero off-diagonal entry in the coin matrix. Thus, no path $\gamma$ yields a non-zero index and we can not apply Theorem \ref{thm: index theorem}.
\end{example}

\begin{example} \label{exmpl: things work}
    Let $C_x = \begin{pmatrix}
        0 & 1 & 0 \\ 1 & 0 & 0 \\ 0 & 0 & 1
    \end{pmatrix}$ for all $x \in \Gamma$. We can construct a path $\gamma$ that performs the $hh$-step illustrated in Figure \ref{fig: PhiQ H1} for all $t \leq -1$ and the $pp$-steps $\gamma_3$ and $\gamma_5$ for all $t \geq 1$, see Figure \ref{fig: pp steps passing z}. These steps correspond to the coefficients $c_{21}$ and $c_{33}$ and Theorem \ref{thm: index theorem} trivially yields that such a path $\gamma$ leads to a non-zero index with $\Phi$ being trace-class. By a direct computation using the Fourier Transform we verify that $\sigma_{ac}(\mathcal{U}) = \Sp^1$.
\end{example}

\begin{example}
    Let $C_x = \begin{pmatrix}
        * & a_1(x) & * \\ a_2(x) & * & * \\ * & * & a_3(x)
    \end{pmatrix}$. Then the same path $\gamma$ as in Example \ref{exmpl: things work} leads to a non-zero index if and only if $|a_j(x)| > \varepsilon > 0$ for all $x$ sufficiently large (in fact, only for all large $x$ that are close to $\gamma$). Furthermore, if $|a_j(x)| \to 1$, then $\sigma(\mathcal{U}) = \Sp^1$. If this convergence is fast enough such that $\sum_{x \in \Gamma_\gamma} \sqrt{1 - |a_j(x)|^2} < \infty$ for all $j$, then $\sigma_{ac}(\mathcal{U}) = \Sp^1$.
\end{example}

\appendix

\section{First re-sampling argument} \label{sec: first resampling argu}
We derive equation \eqref{eq: first resampling} from \eqref{eq: before first resampling}, following the steps of \cite{HJS:09} (Proposition 13.1). We will let $c$ denote a constant, which is independent of $L$, but may change from line to line.
\subsection{Splitting terms}
We fix $(v^{(k)}, \tilde{v}^{(l)}) \in B_L, (u^{(n)}, \tilde{u}^{(m)}) \in B_{L+\iota}$ such that $v^{(k)} \in \Hp_L$ and $\tilde{u}^{(m)} \in \Hp_{L+\iota}^C$ and define $J = \{v, \tilde{v}, u, \tilde{u} \}$. For all sites $r \in J$ we choose i.i.d. random variables $(\widehat{\omega}_r)_{r\in J}$ that are uniformly distributed on $\T$ and independent from $(\omega_r)_{r \in \Gamma}$. Writing $P_r = \sum_{i=1}^3 \ket{r^{(i)}} \bra{r^{(i)}}$, we define
\begin{align*}
    D_{\omega, \widehat{\omega}} = \sum_{r \in J} e^{i \widehat{\omega}_r} \, P_r + \sum_{r \notin J} e^{i \omega_r} \, P_r.
\end{align*}
$D_{\omega, \widehat{\omega}}$ uses the phase $\widehat{\omega}_r$ on all $r \in J$ and $\omega_r$ on all other sites. We set
\begin{align*}
    U_{\omega, \widehat{\omega}}^{(L)} = D_{\omega, \widehat{\omega}} \, S \, \mathcal{C}^{\Lambda_L}(C) \;\; \text{and} \;\; \widehat{R}^{(L)} = \big( U_{\omega, \widehat{\omega}}^{(L)} - z \big)^{-1}.
\end{align*}
To simplify the notation we define $\widehat{D} = \sum_{r \in J} \big( e^{i \omega_r} - e^{i \widehat{\omega}_r} \big) \, P_r$. Applying the resolvent identity to $R^{(L)}$ and $R^{(L+ \iota)}$ yields
\begin{align*}
    R^{(L)} &= \widehat{R}^{(L)} - \widehat{R}^{(L)} \widehat{D} \, U^{(L)} R^{(L)} \;\; \text{ and } \;\; R^{(L+ \iota)} = \widehat{R}^{(L+ \iota)} - R^{(L+ \iota)} \widehat{D} \, U^{(L+\iota)} \widehat{R}^{(L+ \iota)}.
\end{align*}
We use this to bound each term in the right hand side of \eqref{eq: before first resampling}:
\begin{align} \label{eq: splitting first res}
\begin{split}
    &\E \Big( \big| \langle x^{(i)} \, | \, R^{(L)} \, v^{(k)} \rangle \big|^s \, \big| \langle \tilde{v}^{(l)} \, | \, R \, u^{(n)} \rangle \big|^s \, \big| \langle \tilde{u}^{(m)} \, | \, R^{(L+\iota)} \, y^{(j)} \rangle \big|^s \Big) \\
    &\overset{s<1}{\leq} \widehat{\E} \E \bigg( \Big( \big| \langle x^{(i)} \, | \, \widehat{R}^{(L)} \, v^{(k)} \rangle \big|^s + \big| \langle x^{(i)} \, | \, \widehat{R}^{(L)} \widehat{D} \, U^{(L)} R^{(L)} \, v^{(k)} \rangle \big|^s \Big) \, \big| \langle \tilde{v}^{(l)} \, | \, R \, u^{(n)} \rangle \big|^s \\
    &\hphantom{{} \leq \widehat{\E} \E \bigg(} \Big( \big| \langle \tilde{u}^{(m)} \, | \, \widehat{R}^{(L+ \iota)} \, y^{(j)} \rangle \big|^s + \big| \langle \tilde{u}^{(m)} \, | \, R^{(L+ \iota)} \widehat{D} \, U^{(L+\iota)} \widehat{R}^{(L+ \iota)} \, y^{(j)} \rangle \big|^s \Big) \bigg) \\
    & = A_1 + A_2 + A_3 + A_4,
\end{split}
\end{align}
where
\begin{align*}
    A_1 &= \widehat{\E}\E \bigg( \Big| \langle x^{(i)} \, | \, \widehat{R}^{(L)} \, v^{(k)} \rangle \Big|^s  \; \Big| \langle \tilde{v}^{(l)} \, | \, R \, u^{(n)} \rangle \Big|^s \;  \Big| \langle \tilde{u}^{(m)} \, | \, \widehat{R}^{(L+ \iota)} \, y^{(j)} \rangle \Big|^s  \bigg)\\
    A_4 &= \widehat{\E}\E \bigg( \Big| \langle x^{(i)} \, | \, \widehat{R}^{(L)} \widehat{D} \, U^{(L)} R^{(L)} \, v^{(k)} \rangle \Big|^s \; \Big| \langle \tilde{v}^{(l)} \, | \, R \, u^{(n)} \rangle \Big|^s \\ & \hphantom{{} = \widehat{\E}\E \bigg(} \Big| \langle \tilde{u}^{(m)} \, | \, R^{(L+ \iota)} \widehat{D} \, U^{(L+\iota)} \widehat{R}^{(L+ \iota)} \, y^{(j)} \rangle \Big|^s \bigg).
\end{align*}
The intermediate terms $A_2$ and $A_3$ are defined analogously. Here, $\E$ denotes the expectation over $(\omega_r)_{r \in \Gamma}$ and $\widehat{\E}$ the one over $(\widehat{\omega}_r)_{r \in J}$. 

\subsection{Estimate for $A_1$} \label{sec: est for A1}
Let $\E \big( ... |J \big)$ be the conditional expectation with respect to the $\sigma$-field generated by $(\omega_r)_{r \notin J}$. By the law of total expectation we have for any random variable $X$:
\begin{align} \label{eq: total expectation}
    \E (X) = \E \big( \E (X | J ) \big).
\end{align}
This gives us:
\begin{align*}
    A_1 &= \widehat{\E} \bigg( \E \bigg( \big| \langle x^{(i)} | \, \widehat{R}^{(L)} v^{(k)} \rangle \big|^s  \; \E \Big( \big| \langle \tilde{v}^{(l)} | \, R u^{(n)} \rangle \big|^s \, \Big| \, J \Big) \, \big| \langle \tilde{u}^{(m)} | \, \widehat{R}^{(L+ \iota)} y^{(j)} \rangle \big|^s \bigg) \bigg).
\end{align*}
We have used that the operator $U_{\omega, \widehat{\omega}}^{(L)}$ depends only on the random variables $(\omega_r)_{r \notin J}$ and $(\widehat{\omega}_r)_{r \in J}$. Thus, when taking the expectation $\E$, which is with respect to the random variables $(\omega_r)_{r \in \Gamma}$, the resolvents $\widehat{R}^{(L)}$ and $\widehat{R}^{(L+\iota)}$ are measurable with respect to the $\sigma$-algebra generated by $(\omega_r)_{r \notin J}$ and can be pulled out of the conditional expectation. Similar to Theorem \ref{thm: frac mom bound} we can bound the conditional expectation:
\begin{align*}
    \E \Big( \big| \langle \tilde{v}^{(l)} \, | \, R \, u^{(n)} \rangle \big|^s \, \Big| \, J \Big) \leq C(s).
\end{align*}
Note that this estimate requires $\tilde{v}, u \in J$. We obtain:
\begin{align*}
    A_1 \leq c \; \widehat{\E} \bigg( \E \Big( \big| \langle x^{(i)} \, | \, \widehat{R}^{(L)} \, v^{(k)} \rangle \big|^s  \; \big| \langle \tilde{u}^{(m)} \, | \, \widehat{R}^{(L+ \iota)} \, y^{(j)} \rangle \big|^s \Big) \bigg).
\end{align*}
Notice that the equation above does not contain $(\omega_r)_{r \in J}$. Since $(\omega_r)_{r \in J}$ and $(\widehat{\omega}_r)_{r \in J}$ are identically distributed, we can replace the latter by the former:
\begin{align*}
    A_1 \leq c \; \E \Big( \big| \langle x^{(i)} \, | \, R^{(L)} v^{(k)} \rangle \big|^s  \; \big| \langle \tilde{u}^{(m)} \, | \, R^{(L+ \iota)} \, y^{(j)} \rangle \big|^s \Big).
\end{align*}
We stress that $v^{(k)} \in \Hp_L$, $\tilde{u}^{(m)} \in \Hp_{L+\iota}^C$ and $\Hp_L$, $\Hp_{L+\iota}^C$ are invariant under $R^{(L)}$ and $R^{(L+\iota)}$ respectively. Thus, the two scalar products above are independent random variables and we obtain:
\begin{align} \label{eq: A1 estimate}
    A_1 \leq c \; \E \Big( \big| \langle x^{(i)} \, | \, R^{(L)} v^{(k)} \rangle \big|^s \Big) \; \E \Big( \big| \langle \tilde{u}^{(m)} \, | \, R^{(L+ \iota)} \, y^{(j)} \rangle \big|^s \Big).
\end{align}
This estimate is of the required form.

\subsection{Estimate for $A_4$} \label{subsec: A4}
We apply \eqref{eq: total expectation} to $A_4$ and use Hölder's inequality twice:
\begin{align}
\begin{split} \label{eq: A4 terms}
    A_4 &\leq \widehat{\E} \bigg( \E \bigg( \E \Big( \big| \langle x^{(i)} \, | \, \widehat{R}^{(L)} \widehat{D} \, U^{(L)} R^{(L)} \, v^{(k)} \rangle \big|^{3s} \Big| J \Big)^\frac{1}{3} \, \E \Big( \big| \langle \tilde{v}^{(l)} \, | \, R \, u^{(n)} \rangle \big|^{3s} \Big| J \Big)^\frac{1}{3} \\
    & \hphantom{{} \leq \widehat{\E} \bigg( \E \bigg( } \E \Big( \big| \langle \tilde{u}^{(m)} \, | \, R^{(L+ \iota)} \widehat{D} \, U^{(L+\iota)} \widehat{R}^{(L+ \iota)} \, y^{(j)} \rangle \big|^{3s} \Big| J \Big)^\frac{1}{3} \bigg) \bigg).
\end{split}
\end{align}
Due to $s < \frac{1}{3}$ we can bound the second expectation in \eqref{eq: A4 terms} by Theorem \ref{thm: frac mom bound}. We continue with the first term. Using the definition of $\widehat{D}$ we obtain:
\begin{align*}
    \big| &\langle x^{(i)} \, | \, \widehat{R}^{(L)} \widehat{D} \, U^{(L)} R^{(L)} \, v^{(k)} \rangle \big|^{3s} \\ &\overset{s < \frac{1}{3}}{\leq} \sum_{r \in J} \sum_{l=1}^3 \big| e^{i \omega_r} - e^{i \widehat{\omega}_r} \big|^{3s} \, \big| \langle x^{(i)} \, | \, \widehat{R}^{(L)} r^{(l)} \rangle \big|^{3s} \, \big| \langle r^{(l)} \, | \, U^{(L)} R^{(L)} \, v^{(k)} \rangle \big|^{3s}.
\end{align*}
Note that $x^{(i)} \in \Hp_L$, which is invariant under $\widehat{R}^{(L)}$. Thus, $r^{(l)}$ is an element of $\Hp_L$ as well. Since the number of terms in the sum above is bounded independently of $L$, we obtain
\begin{align*}
    \big| &\langle x^{(i)} | \widehat{R}^{(L)} \widehat{D} \, U^{(L)} R^{(L)} v^{(k)} \rangle \big|^{3s} \leq c \sum_{\substack{r^{(l)} \in \Hp_L \\ \text{s.t. } r \in J}} \big| \langle x^{(i)} \, | \, \widehat{R}^{(L)} r^{(l)} \rangle |^{3s} \, | \langle r^{(l)} \, | \, U^{(L)} R^{(L)}  v^{(k)} \rangle \big|^{3s}.
\end{align*}
We remark that $\widehat{R}^{(L)}$ depends only on $(\widehat{\omega}_r)_{r \in J}$ and $(\omega_r)_{r \notin J}$. Since the expectation $\E$ is with respect to $(\omega_r)_{r \in \Gamma}$, the first scalar product is measurable with respect to the $\sigma$-algebra generated by $(\omega_r)_{r \notin J}$. Therefore:
\begin{align}
\begin{split} \label{eq: A4 first term intermediate step}
    \E &\Big( \big| \langle x^{(i)} \, | \, \widehat{R}^{(L)} \widehat{D} \, U^{(L)} R^{(L)} \, v^{(k)} \rangle \big|^{3s} \Big| J \Big) \\
    &\leq c \sum_{\substack{r^{(l)} \in \Hp_L \\ \text{s.t. } r \in J}} \big| \langle x^{(i)} \, | \, \widehat{R}^{(L)} r^{(l)} \rangle \big|^{3s} \, \E \Big( \big| \langle r^{(l)} \, | \, U^{(L)} R^{(L)} \, v^{(k)} \rangle \big|^{3s} \, \Big| \, J \Big).
\end{split}
\end{align}
Since $D_\omega$ is unitary, we have:
\begin{align*}
    \big| \langle r^{(l)} \, | \, U^{(L)} R^{(L)} \, v^{(k)} \rangle \big| = \big| \langle e^{i \omega_r} \, r^{(l)} \, | \, I + z \, R^{(L)} \, v^{(k)} \rangle \big|.
\end{align*}
Using this, Theorem \ref{thm: frac mom bound} and $|z|<1$, we can bound the expectation in \eqref{eq: A4 first term intermediate step} and obtain for the first term in \eqref{eq: A4 terms}:
\begin{align*}
    \E &\Big( \big| \langle x^{(i)} \, | \, \widehat{R}^{(L)} \widehat{D} \, U^{(L)} R^{(L)} \, v^{(k)} \rangle \big|^{3s} \Big| J \Big) \leq c \sum_{\substack{r^{(l)} \in \Hp_L \\ \text{s.t. } r \in J}} \big| \langle x^{(i)} \, | \, \widehat{R}^{(L)} \, r^{(l)} \rangle \big|^{3s}.
\end{align*}
We can obtain a similar estimate for the third term in \eqref{eq: A4 terms}: As above, we use the definition of $\widehat{D}$, apply the expectation, utilize that $\widehat{R}^{(L+\iota)}$ is measurable with respect to $(\omega_r)_{r \notin J}$ and apply Theorem \ref{thm: frac mom bound} to obtain:
\begin{align*}
    \E \Big( \big| \langle \tilde{u}^{(m)} | R^{(L+ \iota)} \widehat{D} U^{(L+\iota)} \widehat{R}^{(L+ \iota)} y^{(j)} \rangle \big|^{3s} \Big| J \Big) \leq c \sum_{\substack{\tilde{r}^{(k)} \in \Hp_{L+\iota}^C \\ \text{s.t. } \tilde{r} \in J}} \big| \langle \tilde{r}^{(k)} | U^{(L+\iota)} \widehat{R}^{(L+ \iota)} y^{(j)} \rangle \big|^{3s}.
\end{align*}
Due to our assumption $y^{(j)} \in \Hp_{L+2\iota}^C$ we have $y \notin J$, as can be verified in Figure \ref{fig: box and bigger box}: The site of $J$ closest to $y$ is $u$, which is at distance $1$ to $\tilde{u}$. The site $\tilde{u}$ is colored red in Figure \ref{fig: box and bigger box} and since $y$ is outside a box of size $L+2 \iota$, the distance between $y$ and $u$ is at least $3$. We stress that here the assumption $y^{(j)} \in \Hp_{L+\iota}^C $ is not enough. Since $D_{\omega, \widehat{\omega}}$ is unitary, it follows:
\begin{align*}
    \big| \langle \tilde{r}^{(k)} | U^{(L+\iota)} \widehat{R}^{(L+ \iota)} y^{(j)} \rangle \big| &= \big| \langle \tilde{r}^{(k)} | I+ z \, \widehat{R}^{(L+ \iota)} y^{(j)} \rangle \big| = |z| \, \big| \langle \tilde{r}^{(k)} | \widehat{R}^{(L+ \iota)} y^{(j)} \rangle \big|.
\end{align*}
Using $|z| < 1$, this yields for the third term in \eqref{eq: A4 terms}:
\begin{align*}
    \E \Big( \big| \langle \tilde{u}^{(m)} \, | \, R^{(L+ \iota)} \widehat{D} U^{(L+\iota)} \widehat{R}^{(L+ \iota)} y^{(j)} \rangle \big|^{3s} \Big| J \Big) \leq c \sum_{\substack{\tilde{r}^{(k)} \in \Hp_{L+\iota}^C \\ \text{s.t. } \tilde{r} \in J}} \big| \langle \tilde{r}^{(k)} | \widehat{R}^{(L+ \iota)} y^{(j)} \rangle \big|^{3s}.
\end{align*}
Since $J$ has a fixed finite number of elements, we have $\big( \sum_{j \in J} \alpha_j^{3s} \big)^\frac{1}{3} \leq c \sum_{j \in J} \alpha_j^s$ for some constant $c$ and all $\alpha_j > 0$. Using the obtained bounds for the three terms in \eqref{eq: A4 terms}, we conclude:
\begin{align*}
    A_4 \leq c \, \sum_{\substack{r^{(l)} \in \Hp_L \\ \text{s.t. } r \in J}} \, \sum_{\substack{\tilde{r}^{(k)} \in \Hp_{L+\iota}^C \\ \text{s.t. } \tilde{r} \in J}} \widehat{\E} \E \Big( \big| \langle x^{(i)} \, | \, \widehat{R}^{(L)} r^{(l)} \rangle \big|^{s} \, \big| \langle \tilde{r}^{(k)} \, | \, \widehat{R}^{(L+ \iota)} \, y^{(j)} \rangle \big|^{s} \Big).
\end{align*}
We note that the last line only depends on the random variables $(\omega_r)_{r \notin J}$ and $(\widehat{\omega}_r)_{r \in J}$. Since $(\omega_r)_{r \in J}$ and $(\widehat{\omega}_r)_{r \in J}$ are independent and identically distributed, we can replace the latter by the former. The two scalar products inside the expectation are then independent random variables by the same reasoning as for $A_1$. This yields:
\begin{align} \label{eq: A4 estimate}
    A_4 \leq c \, \sum_{\substack{r^{(l)} \in \Hp_L \\ \text{s.t. } r \in J}} \E \Big( \big| \langle | x^{(i)} \, | \, R^{(L)} r^{(l)} \rangle \big|^{s} \Big) \, \sum_{\substack{\tilde{r}^{(k)} \in \Hp_{L+\iota}^C \\ \text{s.t. } \tilde{r} \in J}}  \E \Big( \big| \langle \tilde{r}^{(k)} \, | \, R^{(L+ \iota)} \, y^{(j)} \rangle \big|^{s} \Big).
\end{align}

\subsection{Combining the estimates}
We can interpolate the arguments used for $A_1$ and $A_4$ to obtain similar estimates for $A_2$ and $A_3$. By definition of $v, \tilde{v}, u, \tilde{u}$ we have $\{r \in J \, | \, \exists \, l \text{ s.t. } r^{(l)} \in \Hp_L \} = \{v, \tilde{v} \}$ and $\{ r \in J \, | \, \exists \, l \text{ s.t. } r^{(l)} \in \Hp_{L+\iota}^C \} = \{ u, \tilde{u} \}$. This yields for $A_4$:
\begin{align*}
    \sum_{v,\tilde{v},u,\tilde{u}} A_4 (v,\tilde{v},u,\tilde{u}) \leq &c \bigg( \sum_{u^{(n)} \in \partial(\Lambda_L)} \E \Big( \big| \langle x^{(i)} \, | \, R^{(L)} u^{(n)} \rangle \big|^{s} \Big) \bigg) \\ &\times \bigg( \sum_{\tilde{u}^{(m)} \in \partial(\Lambda_{L+\iota}^C)} \E \Big( \big| \langle \tilde{u}^{(m)} \, | \, R^{(L+ \iota)} \, y^{(j)} \rangle \big|^{s} \Big) \bigg).
\end{align*}
Similar estimates hold for the easier terms $A_1, A_2$ and $A_3$. Using \eqref{eq: before first resampling} and \eqref{eq: splitting first res}, we finally obtain \eqref{eq: first resampling}. We stress that all constants obtained in the appendix \ref{sec: first resampling argu} are independent of $L$.

\section{Second re-sampling argument} \label{sec: second resampling argu}
Analogously to the proof of Proposition 13.2 in \cite{HJS:09}, we derive equation \eqref{eq: second resampling} from \eqref{eq: before second resampling}. Since the methods used are similar to the ones in appendix \ref{sec: first resampling argu}, we will omit details. Note that the term containing $u^{(n)}$ in \eqref{eq: before second resampling} is already in the shape of \eqref{eq: second resampling}, so we only have to deal with the second term. We let $c > 0$ denote some constant independent of $L$ that may change from line to line. Fixing $\tilde{u}^{(m)} \in \partial (\Lambda_{L+\iota}^C)$, $(v^{(k)}, \tilde{v}^{(l)}) \in B_{L+\iota}$ with $v^{(k)} \in \Hp_{L+\iota}^C$, we define $\Tilde{J} = \{ \tilde{u}, v, \tilde{v} \}$. Similar to our approach in appendix \ref{sec: first resampling argu} we re-sample the random variables on all sites in $\tilde{J}$. Let $(\tilde{\omega}_r)_{r \in \tilde{J}}$ be an i.i.d. family of random variables uniformly distributed on the torus $\T$. Let $D_{\omega, \tilde{\omega}}$ use the phase $\tilde{\omega}_r$ on all sites $r \in \tilde{J}$ and $\omega_r$ on all others. Let $U_{\omega, \tilde{\omega}} = D_{\omega, \tilde{\omega}} U$ and $\tilde{R}$ be the re-sampled Quantum Walk and its resolvent. To simplify the notation we define $\tilde{D} = \sum_{r \in \tilde{J}} \big( e^{i \omega_r} - e^{i \tilde{\omega}_r} \big) \, P_r$. Using the resolvent identity $R = \tilde{R} - R \, \tilde{D} \, U \, \tilde{R}$, we obtain for the last expectation in \eqref{eq: before second resampling}:
\begin{align} \label{eq: split in B1 and B2}
\begin{split}
    &\E \bigg( \left|\langle \tilde{u}^{(m)} \, | \, R^{(L+\iota)} \, v^{(k)} \rangle \right|^s \, \left| \langle \tilde{v}^{(l)} \, | \, R \, y^{(j)} \rangle \right|^s \bigg) \\
    &\overset{s < 1}{\leq} \tilde{\E} \E \bigg( \left|\langle \tilde{u}^{(m)} \, | \, R^{(L+\iota)} \, v^{(k)} \rangle \right|^s \, \big| \langle \tilde{v}^{(l)} \, | \, \tilde{R} \, y^{(j)} \rangle \big|^s \bigg) \\ &+ \tilde{\E} \E \bigg( \left|\langle \tilde{u}^{(m)} \, | \, R^{(L+\iota)} \, v^{(k)} \rangle \right|^s \, \big| \langle \tilde{v}^{(l)} \, | \, R \tilde{D} U \tilde{R} \, y^{(j)} \rangle \big|^s \bigg) = B_1 + B_2.
\end{split}
\end{align}
We write $\E (... | \tilde{J})$ for the conditional expectation with respect to the $\sigma$-algebra generated by $(\omega_r)_{r \notin \tilde{J}}$. We proceed as in section \ref{sec: est for A1} and apply the law of total expectation (equation \eqref{eq: total expectation}) to $B_1$, use that $\tilde{R}$ is independent of the random variables $(\omega_r)_{r \in \tilde{J}}$ and apply Theorem \ref{thm: frac mom bound}. We can then replace the random variables $(\tilde{\omega}_r)_{r \in \tilde{J}}$ by $(\omega_r)_{r \in \tilde{J}}$ to obtain:
\begin{align} \label{eq: B1 estimate}
    B_1 \leq c \, \E \Big( \big| \langle \tilde{v}^{(l)} \, | \, R \, y^{(j)} \rangle \big|^s \Big).
\end{align}
The law of total expectation, Hölder's inequality and Theorem \ref{thm: frac mom bound} give for $B_2$:
\begin{align} \label{eq: B2 first estimate}
    B_2 \leq c \, \tilde{\E} \E \Big( \E \Big( \big| \langle \tilde{v}^{(l)} \, | \, R \tilde{D} U \tilde{R} \, y^{(j)} \rangle \big|^{2s} \, \Big| \, \tilde{J} \Big)^\frac{1}{2} \Big).
\end{align}
Using the definition of $\tilde{D}$ and Theorem \ref{thm: frac mom bound}, we obtain for the inner expectation:
\begin{align} \label{eq: B2 second estimate}
    \E &\Big( \big| \langle \tilde{v}^{(l)} \, | \, R \tilde{D} U \tilde{R} \, y^{(j)} \rangle \big|^{2s} \, \Big| \, \tilde{J} \Big) \leq c \sum_{\substack{r^{(k)} \in \Hp \\ \text{s.t. } r \in \tilde{J}}} \big| \langle r^{(k)} \, | \, U \tilde{R} \, y^{(j)} \rangle \big|^{2s}.
\end{align}
Since $y^{(j)} \in \Hp_{L+2\iota}^C$, we have $y \notin \tilde{J}$, by the same argument as in appendix \ref{sec: first resampling argu}. Again, we stress that the requirement $y^{(j)} \in \Hp_{L+\iota}^C$ would not be enough. This yields:
\begin{align} \label{eq: B2 third estimate}
    \big| \langle r^{(k)} | \, U \tilde{R} \, y^{(j)} \rangle \big| = \big| \langle r^{(k)} | \, I + z \tilde{R} \, y^{(j)} \rangle \big| = |z| \, | \langle r^{(k)} | \, \tilde{R} \, y^{(j)} \rangle |.
\end{align}
Since $|\tilde{J}| \leq 3$, we have $\big( \sum_{j \in J} \alpha_j^{2s} \big)^\frac{1}{2} \leq c \sum_{j \in J} \alpha_j^{s}$ for all $\alpha_j > 0$. Using $|z|<1$, we insert \eqref{eq: B2 second estimate} and \eqref{eq: B2 third estimate} into \eqref{eq: B2 first estimate} and replace $\tilde{\omega}_r$ by $\omega_r$ to obtain:
\begin{align} \label{eq: B2 estimate}
    B_2 \leq c \sum_{\substack{r^{(k)} \in \Hp \\ \text{s.t. } r \in \tilde{J}}} \E \Big( | \langle r^{(k)} \, | \, R \, y^{(j)} \rangle |^{s} \Big).
\end{align}
Plugging the estimates \eqref{eq: split in B1 and B2}, \eqref{eq: B1 estimate} and \eqref{eq: B2 estimate} into \eqref{eq: before second resampling} we obtain:
\begin{align*}
    \E \big( \big| &\langle x^{(i)} | \, R \, y^{(j)} \rangle \big|^s \big) \leq c \, \bigg( \sum_{u^{(n)} \in \partial (\Lambda_L)} \E \left( \big| \langle x^{(i)} | \, R^{(L)} \, u^{(n)} \rangle \big|^s \right) \bigg) \\ &\times \sum_{\tilde{u}^{(m)} \in \partial (\Lambda_{L+\iota}^C)} \Bigg( \E \left( \big| \langle \tilde{u}^{(m)} | \, R \, y^{(j)} \rangle \big|^s \right) + c \sum_{\substack{(v^{(k)}, \tilde{v}^{(l)}) \in B_{L+\iota}, \\ v^{(k)} \in \Hp_{L+\iota}^C}} \E \Big( \big| \langle \tilde{v}^{(l)} | \, R \, y^{(j)} \rangle \big|^s \Big) \\ &\hphantom{{} \times \sum_{\tilde{u}^{(m)} \in \partial (\Lambda_{L+\iota}^C)} \Bigg( } + \sum_{\substack{r^{(k)} \in \Hp \\ r \in \{ v,\tilde{v}, \tilde{u}\} }} \E \Big( | \langle r^{(k)} | \, R \, y^{(j)} \rangle |^{s} \Big) \Bigg).
\end{align*}
We stress that $| \partial(\Lambda_{L+\iota}^C)| = \mathcal{O}(|L|)$ and $|B_{L+\iota}| = \mathcal{O}(|L|)$ and note that all scalar products above are of the form $\langle r^{(k)} | R \, y^{(j)} \rangle$ with $r^{(k)} \in \overline{\partial (\Lambda_{L+\iota}^C)}$. Rearranging the terms in the sum above, this finishes the proof of equation \eqref{eq: second resampling}.

\section*{Declarations}
\textbf{Funding:} \\
This work was supported by the MSCA Cofund QuanG (Grant Nr: 101081458), funded by the European Union. The views and opinions expressed are those of the author only and do not necessarily reflect those of the European Union or the granting authority. Neither the European Union nor the granting authority can be held responsible for them. This work was partially supported by the ANR grant number ANR-24-CE40-5714-02.
\vspace{2mm} \\
\textbf{Competing Interests:} \\
The author has no competing interests to declare that are relevant to the content of this article.
\vspace{2mm} \\
\textbf{Data Availability:} \\
This manuscript has no associated data.

\bibliographystyle{plain}  
\bibliography{Bibliography}  

\begin{thebibliography}{10}

\bibitem{Ahlbrecht:2011}
A.~Ahlbrecht, H.~Vogts, A.~H. Werner, and R.~F. Werner.
\newblock Asymptotic evolution of quantum walks with random coin.
\newblock {\em Journal of Mathematical Physics}, 52(4), 2011.

\bibitem{Ahlbrecht:2012}
Andre Ahlbrecht, Christopher Cedzich, Robert Matjeschk, Volkher~B. Scholz,
  Albert~H. Werner, and Reinhard~F. Werner.
\newblock Asymptotic behavior of quantum walks with spatio-temporal coin
  fluctuations.
\newblock {\em Quantum Information Processing}, 11(5):1219–1249, March 2012.

\bibitem{ASW:2011}
Andre Ahlbrecht, Volkher~B. Scholz, and Albert~H. Werner.
\newblock Disordered quantum walks in one lattice dimension.
\newblock {\em Journal of Mathematical Physics}, 52(10), 10 2011.

\bibitem{AENSS:2005}
Michael Aizenman, Alexander Elgart, Serguei Naboko, Jeffrey~H. Schenker, and
  Gunter Stolz.
\newblock Moment analysis for localization in random schrödinger operators.
\newblock {\em Inventiones mathematicae}, 163(2):343–413, October 2005.

\bibitem{AizenmanMolchanov:1993}
Michael Aizenman and Stanislav Molchanov.
\newblock {Localization at large disorder and at extreme energies: an
  elementary derivation}.
\newblock {\em Communications in Mathematical Physics}, 157(2):245 -- 278,
  1993.

\bibitem{AizenmanSchenker:2001}
Michael Aizenman, Jeffrey~H. Schenker, Roland~M. Friedrich, and Dirk
  Hundertmark.
\newblock Finite-volume fractional-moment criteria for anderson localization.
\newblock {\em Communications in Mathematical Physics}, 224(1):219–253,
  November 2001.

\bibitem{Warzel:15}
Michael Aizenman and Simone Warzel.
\newblock {\em Random Operators: Disorder Effects on Quantum Spectra and
  Dynamics}, volume 168 of {\em Graduate Studies in Mathematics}.
\newblock AMS, 2015.

\bibitem{Ambainis:2019}
Andris Ambainis, Andr\'{a}s Gily\'{e}n, Stacey Jeffery, and Martins Kokainis.
\newblock Quadratic speedup for finding marked vertices by quantum walks.
\newblock In {\em Proceedings of the 52nd Annual ACM SIGACT Symposium on Theory
  of Computing}, STOC 2020, page 412–424. Association for Computing
  Machinery, 2020.

\bibitem{Molfetta:2018}
Pablo Arrighi, Giuseppe Di~Molfetta, Iv\'an M\'arquez-Mart\'{\i}n, and Armando
  P\'erez.
\newblock Dirac equation as a quantum walk over the honeycomb and triangular
  lattices.
\newblock {\em Phys. Rev. A}, 97:062111, 06 2018.

\bibitem{ABJ:2017}
Joachim Asch, O.~Bourget, and Alain Joye.
\newblock Chirality induced interface currents in the chalker–coddington
  model.
\newblock {\em Journal of Spectral Theory}, 2017.

\bibitem{ABJ:2010}
Joachim Asch, Olivier Bourget, and Alain Joye.
\newblock Localization properties of the chalker coddington model.
\newblock {\em Annales Henri Poincaré}, 11(7):1341–1373, November 2010.

\bibitem{ABJ:2012}
Joachim Asch, Olivier Bourget, and Alain Joye.
\newblock Dynamical localization of the chalker-coddington model far from
  transition.
\newblock {\em Journal of Statistical Physics}, 147(1):194–205, April 2012.

\bibitem{ABJ:2015}
Joachim Asch, Olivier Bourget, and Alain Joye.
\newblock Spectral stability of unitary network models.
\newblock {\em Reviews in Mathematical Physics}, 27(07):1530004, 2015.

\bibitem{ABJ:2020}
Joachim Asch, Olivier Bourget, and Alain Joye.
\newblock On stable quantum currents.
\newblock {\em Journal of Mathematical Physics}, 61(9):092104, 09 2020.

\bibitem{AschMouneime:2019}
Joachim Asch and Mohamed Mouneime.
\newblock Examples for stable quantum currents, 2019.

\bibitem{ASS:94}
J.~Avron, R.~Seiler, and B.~Simon.
\newblock The index of a pair of projections.
\newblock {\em Journal of Functional Analysis}, 120(1):220--237, 1994.

\bibitem{Boumaza:2025}
Hakim Boumaza and Amine Khouildi.
\newblock Dynamical localization for a random scattering zipper.
\newblock {\em Letters in Mathematical Physics}, 115(64), 05 2025.

\bibitem{Bourget:2003}
Olivier Bourget, James~S. Howland, and Alain Joye.
\newblock Spectral analysis of unitary band matrices.
\newblock {\em Communications in Mathematical Physics}, 234(2):191–227, March
  2003.

\bibitem{Cedzich:2017}
C.~Cedzich, T.~Geib, F.~A. Grünbaum, C.~Stahl, L.~Velázquez, A.~H. Werner,
  and R.~F. Werner.
\newblock The topological classification of one-dimensional symmetric quantum
  walks.
\newblock {\em Annales Henri Poincaré}, 19(2):325–383, November 2017.

\bibitem{Cedzich:2018}
C.~Cedzich, T.~Geib, C.~Stahl, L.~Velázquez, A.~H. Werner, and R.~F. Werner.
\newblock Complete homotopy invariants for translation invariant symmetric
  quantum walks on a chain.
\newblock {\em Quantum}, 2:95, September 2018.

\bibitem{CedWer:2021}
C.~Cedzich and A.~H. Werner.
\newblock Anderson localization for electric quantum walks and skew-shift cmv
  matrices.
\newblock {\em Communications in Mathematical Physics}, 387, 11 2021.

\bibitem{CedFilLi:2024}
Christopher Cedzich, Jake Fillman, Long Li, Darren~C Ong, and Qi~Zhou.
\newblock Exact mobility edges for almost-periodic cmv matrices via gauge
  symmetries.
\newblock {\em International Mathematics Research Notices}, 2024(8):6906--6941,
  12 2023.

\bibitem{CedFil:2023}
Christopher Cedzich, Jake Fillman, and Darren~C. Ong.
\newblock Almost everything about the unitary almost mathieu operator.
\newblock {\em Communications in Mathematical Physics}, 403, 10 2023.

\bibitem{Chalker:1988}
J~T Chalker and P~D Coddington.
\newblock Percolation, quantum tunnelling and the integer hall effect.
\newblock {\em Journal of Physics C: Solid State Physics}, 21(14):2665, 1988.

\bibitem{DNSS:2005}
Anne~Boutet de~Monvel, Serguei Naboko, Peter Stollmann, and Günter Stolz.
\newblock Localization near fluctuation boundaries via fractional moments and
  applications.
\newblock {\em Journal d'Analyse Mathématique}, 100:83--116, 2006.

\bibitem{Molfetta:2022}
Giuseppe Di~Molfetta and Victor Deng.
\newblock Geodesic quantum walks.
\newblock {\em Phys. Rev. A}, 105(6):062420, 2022.

\bibitem{Fefferman:2012}
Charles Fefferman and Michael Weinstein.
\newblock Honeycomb lattice potentials and dirac points.
\newblock {\em Journal of the American Mathematical Society}, 25, 02 2012.

\bibitem{FroehlichSpencer:1983}
J{\"u}rg Fr{\"o}hlich and Thomas Spencer.
\newblock {Absence of diffusion in the Anderson tight binding model for large
  disorder or low energy}.
\newblock {\em Communications in Mathematical Physics}, 88(2):151 -- 184, 1983.

\bibitem{Grover:1996}
Lov~K. Grover.
\newblock A fast quantum mechanical algorithm for database search.
\newblock In {\em Proceedings of the Twenty-Eighth Annual ACM Symposium on
  Theory of Computing}, STOC '96, page 212–219. Association for Computing
  Machinery, 1996.

\bibitem{HJ:2012}
Eman Hamza and Alain Joye.
\newblock Correlated markov quantum walks.
\newblock {\em Annales Henri Poincaré}, 13(8):1767–1805, March 2012.

\bibitem{HJ:2014}
Eman Hamza and Alain Joye.
\newblock Spectral transition for random quantum walks on trees.
\newblock {\em Communications in Mathematical Physics}, 326(2):415–439,
  February 2014.

\bibitem{HJS:2006}
Eman Hamza, Alain Joye, and Günter Stolz.
\newblock Localization for random unitary operators.
\newblock {\em Letters in Mathematical Physics}, 75(3):255–272, January 2006.

\bibitem{HJS:09}
Eman Hamza, Alain Joye, and Günter Stolz.
\newblock Dynamical localization for unitary anderson models.
\newblock {\em Mathematical Physics, Analysis and Geometry}, 12(4):381–444,
  September 2009.

\bibitem{Higuchi:2014}
Yusuke Higuchi, Norio Konno, Iwao Sato, and Etsuo Segawa.
\newblock Spectral and asymptotic properties of grover walks on crystal
  lattices.
\newblock {\em Journal of Functional Analysis}, 267(11):4197--4235, 2014.

\bibitem{Higuchi:2018}
Yusuke Higuchi and Etsuo Segawa.
\newblock Quantum walks induced by dirichlet random walks on infinite trees.
\newblock {\em Journal of Physics A: Mathematical and Theoretical},
  51(7):075303, January 2018.

\bibitem{Joye:2005}
Alain Joye.
\newblock Fractional moment estimates for random unitary operators.
\newblock {\em Letters in Mathematical Physics}, 72(1):51–64, April 2005.

\bibitem{Joye_2011}
Alain Joye.
\newblock Random time-dependent quantum walks.
\newblock {\em Communications in Mathematical Physics}, 307(1):65–100, July
  2011.

\bibitem{Joye:2011}
Alain Joye.
\newblock Random unitary models and their localization properties.
\newblock {\em Contemporary Mathematics}, 552:117--134, 04 2011.

\bibitem{Joye:2012}
Alain Joye.
\newblock Dynamical localization for d-dimensional random quantum walks.
\newblock {\em Quantum Information Processing}, 11:1251--1269, 2012.

\bibitem{joye:2024}
Alain Joye.
\newblock Unitary and open scattering quantum walks on graphs, 2024.

\bibitem{Joye:2014}
Alain Joye and Laurent Marin.
\newblock Spectral properties of quantum walks on rooted binary trees.
\newblock {\em Journal of Statistical Physics}, 155(6):1249–1270, April 2014.

\bibitem{JM:2010}
Alain Joye and Marco Merkli.
\newblock Dynamical localization of quantum walks in random environments.
\newblock {\em Journal of Statistical Physics}, 140:1025--1053, 2010.

\bibitem{Karski:2009}
Michal Karski, Leonid Förster, Jai-Min Choi, Andreas Steffen, Wolfgang Alt,
  Dieter Meschede, and Artur Widera.
\newblock Quantum walk in position space with single optically trapped atoms.
\newblock {\em Science}, 325(5937):174–177, July 2009.

\bibitem{Kato:66}
Tosio Kato.
\newblock {\em Perturbation Theory for Linear Operators}.
\newblock 1966.

\bibitem{Kempe:2003}
J~Kempe.
\newblock Quantum random walks: An introductory overview.
\newblock {\em Contemporary Physics}, 44(4):307–327, July 2003.

\bibitem{Klausen:2023}
Frederik~Ravn Klausen.
\newblock Random problems in mathematical physics, 2023.

\bibitem{Koshovets:1991}
I.~A. Koshovets.
\newblock Unitary analog of the anderson model. purely point spectrum.
\newblock {\em Theoretical and Mathematical Physics}, pages 1240 -- 1270, 1991.

\bibitem{Molfetta:2024}
Ugo Nzongani, Nathana\"el Eon, Iv\'an M\'arquez-Mart\'{\i}n, Armando P\'erez,
  Giuseppe Di~Molfetta, and Pablo Arrighi.
\newblock Dirac quantum walk on tetrahedra.
\newblock {\em Phys. Rev. A}, 110:042418, 10 2024.

\bibitem{Portugal:2013}
Renato Portugal.
\newblock {\em Quantum Walks and Search Algorithms}.
\newblock Springer Publishing Company, Incorporated, 2013.

\bibitem{Qiang:2024}
Xiaogang Qiang, Shixin Ma, and Haijing Song.
\newblock Quantum walk computing: Theory, implementation, and application.
\newblock {\em Intelligent Computing}, 3:0097, 2024.

\bibitem{Richard:2017}
S.~Richard, A.~Suzuki, and R.~Tiedra~de Aldecoa.
\newblock Quantum walks with an anisotropic coin i: spectral theory.
\newblock {\em Letters in Mathematical Physics}, 108(2):331–357, September
  2017.

\bibitem{felix:22}
Félix Rouveyre.
\newblock Marches quantiques sur un réseau hexagonal.
\newblock Master's thesis, Université Grenoble Alpes, 2022.

\bibitem{Shenvi:2003}
Neil Shenvi, Julia Kempe, and K.~Birgitta Whaley.
\newblock Quantum random-walk search algorithm.
\newblock {\em Physical Review A}, 67(5), 2003.

\bibitem{Stollmann:2001}
P.~Stollmann.
\newblock {\em Caught by Disorder: Bound States in Random Media}.
\newblock Caught by Disorder: Bound States in Random Media. Birkh{\"a}user,
  2001.

\bibitem{Szegedy:2004}
M.~Szegedy.
\newblock Quantum speed-up of markov chain based algorithms.
\newblock In {\em 45th Annual IEEE Symposium on Foundations of Computer
  Science}, pages 32--41, 2004.

\bibitem{Venancio:2023}
B.~F. Venancio, H.~S. Ghizoni, and M.~G.~E. da~Luz.
\newblock Scattering quantum walk framework for two-dimensional materials: The
  case of honeycomb lattice structures.
\newblock {\em Phys. Rev. B}, 108:094303, 2023.

\bibitem{VenegasAndraca:2012}
Salvador~Elías Venegas-Andraca.
\newblock Quantum walks: a comprehensive review.
\newblock {\em Quantum Information Processing}, 11(5):1015–1106, July 2012.

\bibitem{WangDam:2019}
Fengpeng Wang and David Damanik.
\newblock Anderson localization for quasi-periodic cmv matrices and quantum
  walks.
\newblock {\em Journal of Functional Analysis}, 276(6):1978--2006, 2019.

\bibitem{Zaehringer:2010}
F.~Zähringer, G.~Kirchmair, R.~Gerritsma, E.~Solano, R.~Blatt, and C.~F. Roos.
\newblock Realization of a quantum walk with one and two trapped ions.
\newblock {\em Physical Review Letters}, 104(10), 2010.

\end{thebibliography}

\end{document}